\documentclass[aoas]{imsart}

\RequirePackage{amsthm,amsmath,amsfonts,amssymb}
\RequirePackage[authoryear]{natbib}
\usepackage{xcolor}
\RequirePackage{graphicx}

\RequirePackage{algpseudocode}
\RequirePackage{enumerate}
\RequirePackage{url}
\RequirePackage{algorithm}
\RequirePackage{subfigure}
\RequirePackage{tabularx} 
\RequirePackage{threeparttable}
\RequirePackage{microtype}
\RequirePackage{array}
\RequirePackage{mathrsfs}
\RequirePackage{bm}
\RequirePackage{multirow}
\RequirePackage{cite}
\RequirePackage{booktabs}
\RequirePackage{float}
\usepackage{arydshln}

\startlocaldefs
\theoremstyle{plain}

\newtheorem{proposition}{Proposition}
\newtheorem{theorem}{Theorem}[section]
\newtheorem{lemma}[theorem]{Lemma}
\newtheorem{corollary}{Corollary}
\newcommand{\tabincell}[2]{\begin{tabular}{@{}#1@{}}#2\end{tabular}}

\theoremstyle{remark}

\endlocaldefs

\begin{document}

\begin{frontmatter}
	\title{Applied Regression Analysis of Correlations for Correlated Data}
	\runtitle{Regression Analysis of Correlations}
	
	\begin{aug}
		\author[A]{\fnms{Jie} \snm{Hu}\ead[label=e1]{hujie@mail.ustc.edu.cn}},
		\author[A]{\fnms{Yu} \snm{Chen}\ead[label=e2,mark]{cyu@ustc.edu.cn}},
		\author[B]{\fnms{Chenlei} \snm{Leng}\ead[label=e3,mark]{C.Leng@warwick.ac.uk}},
		\and
		\author[C]{\fnms{Cheng Yong} \snm{Tang}\ead[label=e4,mark]{yongtang@temple.edu}}
		\address[A]{International institute of  Finance, School of Management, 
			University of Science and Technology of China,
			\printead{e1,e2}}

		\address[B]{Department of Statistics, 
			University of Warwick,
			\printead{e3}}
		
		\address[C]{Department of Statistics, Operations, and Data Science, 
			Temple University,
			\printead{e4}}
	\end{aug}
	
	\begin{abstract}
		Correlated data are ubiquitous in today's data-driven society. While regression models for analyzing means and variances of responses of interest are relatively well-developed, the development of these models for analyzing the correlations is largely confined to longitudinal data, a special form of  sequentially correlated data. This paper proposes a new method for the analysis of correlations
		to fully exploit the use of covariates for general correlated data. In a renewed analysis of the Classroom data, a highly unbalanced multilevel clustered data with within-class and within-school correlations, our method  reveals  informative insights on these structures not previously known. In another analysis of the malaria immune response data in Benin, a longitudinal study with time-dependent covariates where the exact times of the observations are not available, our approach again provides promising new results. At the heart of our approach is a new generalized z-transformation that converts correlation matrices constrained to be positive definite to vectors with unrestricted support, and is order-invariant. These two properties enable us to develop regression analysis incorporating covariates for the modelling of correlations via the use of maximum likelihood.

	\end{abstract}
	
	\begin{keyword}
		\kwd{Correlogram}
		\kwd{Correlated data analysis}
		\kwd{Correlation matrix}
		\kwd{Generalized z-transformation}
		\kwd{Regression modeling}
		\kwd{Testing correlation structures}
	\end{keyword}
	
\end{frontmatter}

\section{Introduction}
\label{sec:intro}

\subsection{Background} 
\label{sec:background}
Correlated data arise in a variety of forms in many applied fields including epidemiology, social science, biology, public health, psychology, and economics. This paper is motivated by the insufficiency of existing approaches for analyzing two such datasets. The first, as detailed in Section \ref{dataexample}, is a highly unbalanced multilevel clustered classroom dataset with within-class and within-school correlations, together with subject-related covariates such as race, gender, socioeconomic status, and more. The second, with full detail provided in Section S.4  of the Supplementary Material, comes from a longitudinal study where multiple sequential observations are collected on individuals over time, together with a collection of associated covariates. The observations in these two examples are not independent. 
Understanding, accounting for, modelling, and utilizing the correlations in datasets like these are fundamental for valid inference in statistical analysis. Equally importantly, there are numerous occasions where correlations can be of central scientific interest to draw insightful inferential conclusions. Regression type of models that can fully exploit covariate information for the analysis of correlations however are underdeveloped.

For analyzing clustered data, 
one of the most celebrated approaches is the mixed-effects model approach \citep{Laird:Ware:1982} 
that incorporates random effects, typically  additive to the fixed effects. 
As a result, the marginal variance of the response scales quadratically with the random effects; and the between-observations correlations are also implicitly determined.  
This intrinsic tie between the marginal variances and the correlations may restrict the capacity to adequately accommodate a broad range of correlation structures. 
While mixed-effects models are suitable for fitting categorical random effects such as grouping variables, they have limited ability to incorporate heterogeneity in variance caused by continuous random effects. 
In a great many applications, it is more reasonable to allow 
for the investigation of these components -- mean, variance, and correlation -- separately; see our informative comparisons with  the mixed-effects model 
in a real data analysis in Section  \ref{dataexample}.  

While the versatility of the mixed-effects model has been increasingly recognized in many applied fields,  
there are several challenging issues, including specifying and fitting complex correlation models,  identifying the number of parameters involved, determining the degrees of freedom for various statistics, and quantifying their asymptotic distributions under the null hypothesis  when the parameter lies on the boundary of its support \citep[cf.]{Muller:2013,Barr:2013,Bates:2015}.
This sometimes gives rise to confusing, and even contradictory recommendations about what tests to use, and how and when, not only by applied workers but also by professional statisticians \citep{GLMM:FAQ, Luke:2017}.

In the area of analyzing longitudinal data, numerous  approaches have been developed for modeling the covariances; see the monograph \citet{Fitzmaurice2008}, among the abundant literature. 
Related to the effort of this paper for incorporating covariates, there is a line of research on explicitly and parsimoniously modelling the covariances with an unconstrained parametrization via  Cholesky-type decompositions. 
The recipe is to  transform the restricted supports of the covariance matrices to unrestricted ones; then regression modelling approaches are exploited by utilizing the transformed parameters; see   \citet{Pourah:1999, Pourah:2000, Pan:2003, Ye:2006, Leng:2010, Zhang:2012}.
For explicit regression analysis of correlations in longitudinal data,   \citet{Zhang:2015} explored the use of hyperspherical coordinates to parametrize a correlation matrix via angles.  
Since these Cholesky-type decompositions are tied to the known ordering of the observations, they are  unfortunately unsuitable for general correlated data such as clustered data where no ordering information is available. Indeed, as demonstrated in  our real data example in Section S.4,  missing or erroneous ordering  of the longitudinal data may lead to inadequate model fitting
for the approaches of this type.

\subsection{Our study}

Naturally, a statistician's take on modeling the correlations is to apply a regression technique accounting for the potential contributions from covariates. 
However, there are several major interrelated challenges to doing so:  
\begin{itemize}
	\item[1.] The matrix containing  all the pairwise correlations is inherently positive definite and has ones on its main diagonal; thus the parameters lie in a highly constrained space;
	\item[2.]  The number of correlations is related to  that of the observations in each subject/cluster;  as a result, the correlation matrices for all the subjects/clusters have different dimensions when the data are unbalanced;
	\item[3.]  Correlated data  are not necessarily ordered. 
\end{itemize}

To meet the challenges, this paper makes a dedicated effort to propose and establish a novel,  simple,  flexible, unified  inferential tool for applied regression  modelling of correlations for general correlated data, regardless of whether they are ordered.  
As a regression model,  our approach relates the entries in correlation matrices to any covariates via a new unrestricted parametrization of parameters and thus simultaneously addresses the three aforementioned challenges.  
Analogous to extending the support of the correlation coefficient from the unit interval to the whole real line, the new parametrization extends the support of a correlation matrix from a restricted space to an unrestricted one.  
A remarkable advantage of the new parametrization is its order-invariance: re-ordering the variables in the correlation matrix results in the same re-ordering of the components in the new parametrization. 
The two merits -- unrestricted support and order-invariance -- make the new parametrization an ideal device for modelling the correlation structures of generic correlated data. 

This new framework enables unparalleled convenience in applications for correlation model building, as we establish new tools for both exploratory analysis (see the graphical tool  in Section \ref{sec:corro}) and statistical inference.   It not only provides a framework  for multiple layers of random effects, similar to that of various mixed-effects models - see  Section \ref{sec:cor},  but also offers opportunities for new discoveries, as we have demonstrated through our data examples.  In Section \ref{dataexample} for analyzing the classroom data of \citet{hill2005effects},  which aims at finding factors affecting the math performance of elementary school students,  our study reveals interesting new insights on the between-students correlations:  there is significant evidence that  the co-movement of their performances is affected by their social-economic status, besides the benchmarking scores of their earlier performances.   We believe that this finding provides a new and fresh perspective on understanding the impact of different factors on the progress of students' math performance.  Apart from clustered data, our framework provides competitive, sometimes much better, performance when it comes to analyzing ordered observations, such as longitudinal data,  as we have demonstrated by  real data analysis and  simulation studies presented in  the Supplementary Material. In particular,  our approach is competitive to existing regression tools using Cholesky-type decompositions when the ordering of the data is properly incorporated, while it substantially outperforms them when the ordering is not available or erroneous.

Our framework conveniently bridges clustered data analysis and the golden standard statistical inference built upon likelihood.  
Our approach to model correlations via an explicit regression model is rooted in the likelihood functions where the regularity conditions are easily satisfied in our setting, eliminating any ambiguity in specifying quantities such as the limiting distributions of  test statistics and their degrees of freedom. 
For instance, the approach proposed in this paper makes the likelihood ratio test  extremely easy and accessible in applications, as we have demonstrated in  the real data example in Section \ref{dataexample}. 
We remark further that although only linear models for correlations are studied in this paper,  our framework allows viable extensions to generalized linear models,  semi-parametric and non-parametric models.   Therefore, we advocate our framework for its great flexibility and substantial potential as an applied tool.  

The rest of the paper is organized as follows. Section  \ref{sec:meth} introduces the generalized z-transformation, outlines the correlation model and compares it to the mixed-effects model, presents maximum likelihood for parameter estimation, and provides the theoretical justification. We present   two data analyses in Section \ref{sec:data} to illustrate the advantages of our framework over mixed-effects models when observations are not ordered, and over Cholesky-type decompositions when analyzing ordered data. A brief conclusion is made in Section \ref{conc}. All technical details, one more real data analysis, and simulation examples are found in the Supplementary Material.

\section{Methodology} \label{sec:meth}

\subsection{Generalized z-transformation}

We introduce some notation first. For a symmetric matrix $\mathbf{A}\in \mathbb{R}^{m \times m}$, the operator $\operatorname{vecl}(\mathbf{A}) \in \mathbb{R}^{m \times (m-1)/2}$ stacks the lower off-diagonal elements of $\mathbf{A}$ into a vector. The operator $\operatorname{diag}(\cdot)$ is used in two ways. When applied to a vector $\bm{v}=(v_1, \ldots, v_m)^{\prime} \in \mathbb{R}^m$, $\operatorname{diag}(\bm{v})$ becomes a diagonal matrix with diagonal terms being $v_1, \ldots, v_m$. When applied to $\mathbf{A}$, $\operatorname{diag}(\mathbf{A})$ extracts the diagonals of $\mathbf{A}$ to return a length-$m$ vector.   We use $e^{\mathbf{A}}$ to represent the matrix exponential of $\mathbf{A}$, and use $\log \mathbf{A}$ to denote the matrix logarithm of $\mathbf{A}$, assuming $\mathbf{A}$ is positive definite, defined respectively as $e^\mathbf{A}= \mathbf{Q} \operatorname{diag} (e^{\lambda_1}, \ldots, e^{\lambda_m})\mathbf{Q}^{\prime}$ and $\log \mathbf{A}= \mathbf{Q} \operatorname{diag}(\log \lambda_1, ..., \log \lambda_m) \mathbf{Q}^{\prime}$, where $\mathbf{A}=\mathbf{Q} \operatorname{diag}(\lambda_1, \ldots, \lambda_m) \mathbf{Q}^{\prime}$ is the eigen-decomposition of $\mathbf{A}$, with $\mathbf{Q}$ being an orthonormal matrix.

For a correlation $\rho \in (-1, 1)$, the Fisher's z-transformation is  $z=\frac{1}{2}\log \frac{1+\rho}{1-\rho} \in \mathbb{R}$,  transforming a restricted parameter to an unrestricted one.  
\citet{Archakov:2020} recently discovered a matrix operation in the same spirit,  transforming the restricted support of a correlation matrix to an unrestricted one.  The transformation defines a mapping $f$ from  a correlation matrix $\mathbf{R}\in \mathbb{R}^{m \times m}$ to an $m(m-1)/2$-dimensional vector denoted as $\bm{\gamma}$  via
\begin{align}\label{eq:ft}
	\bm{\gamma}=f(\mathbf{R})=\operatorname{vecl} (\log \mathbf{R}).
\end{align} 
Hereafter, we  refer to this transformation as generalized z-transformation.
It has the following  remarkable properties.
\begin{itemize}
	\item[(a)]\textbf{One-to-one mapping between $\mathbf{R}$ and $\bm{\gamma}=f(\mathbf{R})$}.  \citet{Archakov:2020} show that
	for any real symmetric matrix $\mathbf{G} \in \mathbb{R}^{m \times m}$, there exists a unique vector $\mathbf{x}^{*} \in \mathbb{R}^{m}$, such
	that $e^{\mathbf{G}[\mathbf{x}^{*}]}$ is a correlation matrix, where $\mathbf{G}[\mathbf{x}^{*}]$ denotes the matrix $\mathbf{G}$ with $\mathbf{x}^*$ replacing its diagonal.
	This ensures the existence and uniqueness of the inverse mapping in \eqref{eq:ft}. 	To find $\mathbf{x}^{*}$,  \citet{Archakov:2020}  show that
	the sequence 
	$
	\mathbf{x}_{(k)}=\mathbf{x}_{(k-1)}-\log \operatorname{diag}\left(e^{\mathbf{G}[\mathbf{x}_{(k-1)}]}\right)$ 
	converges to $\mathbf{x}^{*}$ as $k\to \infty$
	with arbitrary $\mathbf{x}_{(0)} \in \mathbb{R}^{m}$.
	
	This suggests that to find $\mathbf{R}= f^{-1}(\bm{\gamma})$ from a  $\bm{\gamma}$, whose support is contraint-free, one starts from a symmetric matrix $\mathbf{G}$ = $\operatorname{vecl}^{-1}(\bm{\gamma})$  with arbitrary diagonals. Then upon determining 
	$\mathbf{x}^{*}$ with the above algorithm,  the corresponding correlation matrix is $\mathbf{R}=e^{\mathbf{G}[\mathbf{x}^{*}]}$. 
	
	\item[(b)]\textbf{Order-invariance}.  
	Let  $\bm{y}=(y_1,\dots,y_m)'$ and $\bm{x}=(x_1,\dots,x_m)'$ be random vectors such that $\bm{y}=\mathbf{P} \bm{x} $ where $\mathbf{P}$ is a permutation matrix.    The correlations of these two vectors satisfy 
	$\operatorname{corr}(\bm{y})=\mathbf{R}_{y}=\mathbf{P}\operatorname{corr}(\bm{x})\mathbf{P}'=\mathbf{P} \mathbf{R}_{x} \mathbf{P}'$.   Hence, the correlations are order-invariant in the sense that if the $i$th and $j$th components of $\bm{y}$ are the  $k$th and $l$th components of $\bm{x}$ before the permutation, then
	$\operatorname{corr}(y_i,y_j)=\operatorname{corr}(x_k,x_l)$. 
	Following simple calculations,  the corresponding generalized z-transformations, in their matrix forms, are also order-invariant since 
	$\log (\mathbf{R}_{y} ) =  \mathbf{P}  \log(\mathbf{R}_{x}) \mathbf{P}' $.

\end{itemize}

\subsection{Parsimonious modelling of the correlation matrix }\label{sec:cor}

We assume 
$n$ generic groups of dependent data,   each consisting of $m_i$ observations for the $i$th group $(i=1, \ldots, n)$.  The correlation matrix of each group is 
$\mathbf{R}_i=(\rho_{ijk})$, and its generalized z-transformation leads to  $\bm{\gamma}_i=(\gamma_{ijk})$  $(i=1, \ldots, n; 1\le k<j\le m_i)$. Let $y_{ij}$ be the $j$th observation of group $i$ associated with covariate $\mathbf{x}_{ij}$; then given $\mathbf{x}_{ij}$ and $\mathbf{x}_{ik}$, ${\rho}_{ijk}=\text{corr}(y_{ij},y_{ik})$.
Our proposal for the correlation model is simply
\begin{align} \label{eq:cm}
	\gamma_{i j k}=\mathbf{w}_{i j k}^{\prime} \bm{\alpha},
\end{align}
where $\bm{\alpha}$ is an unknown parameter which will be referred to as the matrix log-correlation parameter, 
and $\mathbf{w}_{i j k}\in {\mathbb R}^d$ are observations 
or constructions of covariates associated with $y_j$ and $y_k$ that used to model the correlation. For example, a reasonable choice of constructing $\mathbf{w}_{i j k}$ is to take the difference of the covariates at observations $j$ and $k$ for subject $i$. This ensures that the resulting correlation matrix is stationary if the covariate is time. 
For more general continuous covariates other than time, this choice is advantageous in avoiding concerns when extrapolation is needed.  Our correlation model enables  a highly parsimonious device  -- one parameter $\bm{\alpha}$ accounts for $n$ correlation matrices of arbitrary sizes for modelling general dependent data.

We would like to highlight that the linearity in  \eqref{eq:cm} is inspired by and analogous to similar ideas in many widely used statistical models, most notably the generalized linear model. In the latter, it is achieved by relating the mean via a link function to a linear function of the covariates. In ours, linearity is attained via a matrix-log transformation.

Of course, the construction of $\mathbf{w}_{ijk}$ is flexible and can be tailor-made to suit  practical scenarios; a few examples are provided below and more can be found in Section \ref{sec:data}.
In some of these examples, with some abuse of notations, we use the indices $i, j$ and $k$ differently from those in \eqref{eq:cm} in to highlight the structures of intended correlated data under discussion.

\textit{Example 1}:
The device in \eqref{eq:cm} can conveniently handle the correlation structure of temporally dependent data.   For example, for modelling longitudinal data, the covariates $\mathbf{w}_{ijk}$ can be simply taken from the $j$th and $k$th measurements of the $i$th subject ($i\in\{1,\dots, n\}; j,k\in\{1,\dots, m_i\}$) including their observed time $t_{ij}$ and $t_{ik}$ and other time-dependent variables. In this case, \eqref{eq:cm}  can be applied in a similar spirit as in \citet{Pourah:1999} and \citet{Zhang:2015}. 	

\textit{Example 2}: 
For clustered data with nested random effects, a typical mixed-effects model can be denoted as
$
y_{ijk}=g(\mathbf{x}_{ijk})+u_{i}+v_{ij}+\epsilon_{ijk}$,
where three indices $i, j$ and $k$ are needed for the notation to maintain the usual interpretations of the indices in the context of mixed-effects models.
Here $g(\mathbf{x}_{ijk})$ is the fixed effects properly specified as a function of covariates $\mathbf{x}_{ijk}$. The additive random errors are independent and satisfy  $u_i\sim N(0,\sigma^2_u)$, $v_{ij}\sim N(0,\sigma^2_v)$,  and $\epsilon_{ijk}\sim N(0,\sigma^2)$.
Then $\operatorname{cov}(y_{ijk}, y_{ijk'})=\sigma_{u}^2+\sigma_v^2$ for two observations sharing two common random effects, and $\operatorname{cov}(y_{ijk}, y_{ij'k'})=\sigma_u^2$ for two observations sharing one common random effect. All the other cases are uncorrelated.

The correlation structure from this mixed-effects model can be modelled via the device in  \eqref{eq:cm} by taking $\mathbf{w}_{i j k}=(w_{ijk,0}, w_{ijk,1})^\prime$ where $w_{ijk,0}=1$ corresponds to the intercept coefficient, and
$w_{ijk,1}=1$ if  the $j$th and the $k$th observations are in the same group sharing two common random effects, and $w_{ijk,1}=0$  otherwise. By construction,
$\gamma_{i j k}=\alpha_{0}+\alpha_{1}$ or $\gamma_{i j k}=\alpha_{0}$. 

\textit{Example 3}:  As a generalization of Example 2,
one may consider  $v_{ij_l}\sim N(0,\sigma_{v_l}^2), l=1, \ldots, L$  in 
$
y_{ij_lk}=g(\mathbf{x}_{ij_lk})+
u_{i}+v_{ij_l}+\epsilon_{ijk}
$ to allow heterogeneous random effects in $v$.
Then $\operatorname{cov}(y_{ij_l k}, y_{ij_lk'})=\sigma_{u}^2+\sigma_{v_l}^2$ for two observations sharing two common random effects, and other covariances remain the same as in Example 2. 

The correlation structure from this mixed-effects model can be modelled via the device in  \eqref{eq:cm} by taking $\mathbf{w}_{i j k}=(w_{ijk,0}, w_{ijk,1}, \ldots, w_{ijk,L})^\prime$ where $w_{ijk,0}=1$ is the intercept, and $w_{ijk,l}$ is the dummy variable defined as $w_{ijk,l}=1$ if  the $j$th and the $k$th observations are in the same group sharing the same $v_{ij_l}$ and $w_{ijk,l}=0$  otherwise.

\textit{Example 4}: For multiple-level  clustered data with $L+1$ levels,  a linear mixed-effects model is
\[
y_{i_0i_1\ldots i_L}=g(\mathbf{x}_{i_0i_1\ldots i_L})+
u_{i_0}+u_{i_0i_1}+\ldots+\epsilon_{i_0i_1\ldots i_L}\]
where additive random effects are independent, satisfying $u_{i_0\ldots i_l}\sim N(0,\sigma^2_{u_l})$ for $l=0, \ldots, L-1$.
Then $\operatorname{cov}(y_{i_0\ldots i_li_{l+1}\ldots i_L}, y_{i_0\ldots i_li'_{l+1}\ldots i'_L})=\sum_{i=0}^{l}\sigma_{u_l}^2$ for two observations sharing common $l+1$   levels of random effects.

Analogously, by appropriate coding  in \eqref{eq:cm},   one can set $\gamma_{i j k}=\sum_{i=0}^{l} \alpha_{l}$ for two observations sharing common $l+1$   levels of random effects.

\textit{Example 5}: In crossed random-effects models, the units at the same level of a  hierarchy are simultaneously classified by more than one factor. For example, a two-level additive variance components model can be denoted as
$
y_{i(jk)}=g(\mathbf{x}_{i(jk)})+
u_{j}+v_{k}+\epsilon_{i(jk)}$, 
where additive random effects are independent and satisfy $u_{j}\sim N(0,\sigma^2_{u})$, $v_{k}\sim N(0,\sigma^2_{v})$ and $\epsilon_{i(jk)}\sim N(0,\sigma^2)$. Then $\operatorname{cov}(y_{i(jk)}, y'_{i(jk)})=\sigma_{u}^2+\sigma_{v}^2$, $\operatorname{cov}(y_{i(jk)}, y'_{i(jk')})=\sigma_{u}^2 $ and $\operatorname{cov}(y_{i(jk)}, y'_{i(j'k)})=\sigma_{v}^2 $.

The  correlation structure from this mixed-effects model can be modelled via the device in  \eqref{eq:cm} by taking $\mathbf{w}_{i j k}=(w_{ijk,0}, w_{ijk,1}, w_{ijk,2})^\prime$ where $w_{ijk,0}=1$ is the intercept. The two dummy variables are defined as $w_{ijk,1}=1$ if the two individuals share the same $u$ effect and $w_{ijk,1}=0$  otherwise, and $w_{ijk,2}=1$ if the two individuals share the same $v$ effect and $w_{ijk,2}=0$  otherwise.

Examples 2-5 belong to a class of blocked correlation matrices {arising from various  clustered data structures.  
	We remark that although our approach can characterize block-structured correlation matrices, the resulting models are not necessarily equivalent to the linear mixed effect models; we refer to Section \ref{dataexample} for an example. 
	Therefore, caution should be exercised when making  comparisons between the two classes of approaches, and we recommend that researchers be mindful of this distinction.

	In the full generality of the proposed framework, the form of the correlation model in \eqref{eq:cm} can be of ANOVA or ANCOVA type, similar to analogous mean models, to deal with categorical and continuous variables that may influence the correlations in a unified fashion; see our real data examples in Section \ref{sec:data}. 
	
	\subsection{Generalized z-transformation correlogram for continuous covariates}\label{sec:corro}
	While  categorical variables 
	may easily be included in our correlation models,  exploring and deciding possible continuous variables for the model is more challenging in applications, especially when data are highly unbalanced across groups and not ordered.
	Below we establish a key property -- whose proof is given in the Supplementary Material -- of the generalized z-transformation that facilitates a convenient descriptive tool to visually examine the relevance of a continuous variable.  
	\begin{proposition} \label{pro3} Denote by $\bm{\rho}=\operatorname{vecl}(\mathbf{R})$  the vector containing the lower off-diagonal elements of $\mathbf{R}$. 
		The diagonal elements of  $\frac{\partial \bm{\rho}}{\partial \bm{\gamma}}$ are all nonnegative, i.e., $\frac{\partial \rho_{jk}}{\partial \gamma_{jk}}\ge 0 \quad ( 1\le k<j\le m)$.  
	\end{proposition}
	Proposition \ref{pro3} establishes a monotonic relationship between the correlations $\rho_{jk}$ and the corresponding $\gamma_{jk}$: each pair of them moves in the same direction, when all the other parameters are fixed.  For those off-diagonal elements, the relationships between $\rho_{jk}$ and $\gamma_{lm}$ are case by case; see the examples Section S.2 of the Supplementary Material after the proof of Proposition \ref{pro3}. 
	The monotonicity established in Proposition \ref{pro3} is valuable to gain insight and guide exploratory analysis. However, the parameters in $\bm \gamma$ are only equivalent to the correlations after a non-linear matrix transformation and their analytical expressions in matrix form are not generally intuitive. Therefore, we recommend conducting numerical comparison between the original correlations and the transformed parameters based on individual settings and models when the need arises. For examples of such comparison, please refer to Section S.2 of the Supplementary Material.
	
	To examine  the relevance of some covariates  $\mathbf{w}$ for modelling correlations, a straightforward approach is to plot some empirical estimate of $\bm\gamma$ versus  those variables --  an approach sharing the same spirit of the commonly used scatterplots in multiple regressions  for assessing the pairwise associations between variables. 
	However, for highly unbalanced data,  obtaining an estimate of $\bm\gamma$ is difficult. Proposition \ref{pro3} suggests the inspection  of the association  between some empirical estimate of $\bm\rho$, which is more readily available, and the covariates $\mathbf{w}$.  We call a plot towards this purpose a \textit{generalized z-transformation correlogram}, or \textit{GZT-correlogram} hereinafter,  as 
	an analogue to the variogram in analyzing longitudinal data for examining the pattern in variances.  
	
	For balanced ordered data, empirical correlation matrices of the residuals can be computed after fitting a suitable mean model.  From there the corresponding GZT-correlogram can be constructed in a straightforward manner. 
	For handling more difficult situations including highly unbalanced data, we outline
	the following procedure.  
	\begin{itemize}
		\item[Step 1.] Fit a suitable mean and variance model and obtain (standardized) residuals;
		\item[Step 2.] 	For a  given continuous covariate,  create sub-groups of the data by stratification;  
		\item[Step 3.] Calculate the averaged pairwise correlations within  each sub-group. While there are other choices that one can use for the same purpose, for this paper, we will explore the use of local-averaging after stratification.
		\item[Step 4.] Visually examine the pattern of the resulting correlations from Step 3 against the sub-groups  obtained in Step 2.
	\end{itemize}
	If we see a systematic or monotonic trend in Step 4, according to Proposition \ref{pro3}, it will manifest itself in the relationship between $\gamma_{ijk}$ and the corresponding covariates. An example is shown in Figure \ref{fig3}, where 
	plots  depict decreasing patterns in the strength of correlations with two continuous variables, 
	leading us  to investigate the usefulness of these variables in modelling correlations.

	\subsection{Model specification and estimation}
	Having developed a regression model for the transformed correlations, we complement it by including the following regression models of the mean and the log-variances
	\begin{equation}\label{eq:meanvariance}
		\mu_{i j}=\mathbf{x}_{i j}^{\prime} \bm{\beta}, \quad \log (\sigma_{i j}^{2})=\mathbf{z}_{i j}^{\prime} \bm{\lambda},
	\end{equation}
	where $\mu_{i j}$ and $\sigma_{i j}^{2}\left(i=1, \ldots, n ; j=1, \ldots, m_{i}\right)$ are respectively the conditional mean and variance for the $j$th measurement of the $i$th subject, $\mathbf{x}_{ij}$ and $\mathbf{z}_{ij}$ are $p\times 1$ and $q\times 1$ vectors of generic covariates for modelling the mean and the log-variances respectively.  
	The specification of the three models for the mean, variance and correlation structures in \eqref{eq:cm} and \eqref{eq:meanvariance} stipulates the number of parameters unequivocally as $p+q+d$, which can then be easily utilized to decide the degrees of freedoms for various statistics of interest. In contrast, in the mixed-effects type of models,  it is not clear what one should regard as the degrees of freedom for random effects terms  \citep{Bates:2006, Baayen:2008, Faraway:2015}.
	
	Let $\bm{\mu}_{i}=\left(\mu_{i 1}, \ldots, \mu_{i m_{i}}\right)^{\prime}, \mathbf{D}_{i}=\operatorname{diag}\left(\sigma_{i 1}, \ldots, \sigma_{i m_{i}}\right)$ and  $\bm{\omega}=\left(\bm{\beta}^{\prime}, \bm{\alpha}^{\prime}, \bm{\lambda}^{\prime}\right)^{\prime}$. Write $\bm{\nu}_{i}=\mathbf{y}_{i}-\bm{\mu}_{i}$ as the random error associated with the $i$th subject. Then $\bm{\Sigma}_{i}=\operatorname{cov}(\bm{\nu}_{i})=\mathbf{D}_{i} \mathbf{R}_{i} \mathbf{D}_{i}$. If $\mathbf{y}_{i}$ follows Gaussian distribution,
	the log-likelihood is given by
	\begin{equation} \label{eq2}
		l(\bm{\omega})=-\frac{1}{2}\sum_{i=1}^{n} \left( \log \left|\mathbf{D}_{i} \mathbf{R}_{i} \mathbf{D}_{i}\right|+ \bm{\nu}_{i}^{\prime} \mathbf{D}_{i}^{-1} \mathbf{R}_{i}^{-1} \mathbf{D}_{i}^{-1} \bm{\nu}_{i} \right).
	\end{equation}
	
	Let  $\widehat{\bm{\omega}}=\left(\widehat{\bm{\beta}}^{\prime}, \widehat{\bm{\alpha}}^{\prime}, \widehat{\bm{\lambda}}^{\prime}\right)^{\prime}$ be the maximum likelihood estimator. An algorithm for obtaining $\widehat{\bm{\omega}}$ is provided in  detail in the Supplementary Material. 
	
	In establishing the properties of $\widehat{\bm{\omega}}$, we  assume the following regularity conditions. 
	\begin{itemize}
		\item[(A1)]  The dimensions $p, q$ and $d$ of covariates $\mathbf{x}_{i j}$, $\mathbf{z}_{i j}$ and $\mathbf{w}_{i j k}$ are fixed, and $\max _{1 \leqslant i \leqslant n} m_{i}$ is bounded.
		
		\item[(A2)] The parameter space $\bm{\Omega}$ of $\left(\bm{\beta}^{\prime}, \bm{\alpha}^{\prime}, \bm{\lambda}^{\prime}\right)^{\prime}$ is a compact set in $\mathbb{R}^{p+d+q}$, and the true value $\bm{\omega}_{0}=\left(\bm{\beta}_{0}^{\prime}, \bm{\alpha}_{0}^{\prime}, \bm{\lambda}_{0}^{\prime}\right)^{\prime}$ is in the interior of $\bm{\Omega}$.
		
		\item[(A3)]  As $n \rightarrow \infty, n^{-1} \mathbf{I}\left(\bm{\omega}_{0}\right)$ converges to a positive definite matrix $\mathcal{I}\left(\bm{\omega}_{0}\right)$, where $\mathbf{I}\left(\bm{\omega}_{0}\right)$ is the Fisher information matrix at $\bm{\omega}_{0}$.	
	\end{itemize}
	Assumption (A1) is routinely made in the analysis of correlated data. Assumption
	(A2) is a conventional assumption for theoretical analysis of the maximum likelihood approach. Notably,  given our model formulation, it is natural to assume that the true values of the parameters are not on the boundary of the parameter space. 
	Assumption (A3) is a natural requirement for  regression analysis in unbalanced longitudinal data
	modelling. We establish the following asymptotic results for the maximum likelihood estimator, which support statistical inference associated with the model parameters.
	\begin{theorem} \label{tm:inf}
		Under regularity assumptions (A1)--(A3), as $n\rightarrow \infty$, we have that
		\begin{itemize}
			\item[(a)] the maximum likelihood estimator $\widehat{\bm{\omega}}$ is strongly consistent for the true value $\bm{\omega}_{0}$, and
			
			\item[(b)] $\widehat{\bm{\omega}}=\left(\widehat{\bm{\beta}}^{\prime}, \widehat{\bm{\alpha}}^{\prime}, \widehat{\bm{\lambda}}^{\prime}\right)^{\prime}$ is asymptotically normally distributed such that $\sqrt{n}(\widehat{\bm{\omega}}-\bm{\omega}_0)\xrightarrow{d} \mathcal{N}\left(\mathbf{0},\mathcal{I}(\bm{\omega}_{0} )^{-1}\right) $, where $\mathcal{I}(\bm{\omega}_{0} )$ is the Fisher information matrix defined in assumption (A3) and $\xrightarrow{d}$ denotes convergence in distribution.
			
		\end{itemize}	
		
	\end{theorem}
	Here $\widehat{\bm{\beta}}$ is asymptotically independent of $\widehat{\bm{\alpha}}$ and $\widehat{\bm{\lambda}}$, because $\bm{\beta}$ concerns the mean and $\bm{\alpha}$ and $\bm{\lambda}$ are parameters of the covariances. A consistent estimator of $\mathcal{I}(\bm{\omega}_{0} )$ is $\mathcal{I}(\widehat{\bm{\omega}})$ which can be used for inference.  
	Details containing the calculations and relevant intermediate results, and an algorithm for evaluating the maximum likelihood estimator are provided in  the Supplementary Material.

	From Theorem \ref{tm:inf} and \citet{Chernoff}, we establish Corollary \ref{coro1}, an essential result for statistical inference in applications. 
	
	\begin{corollary}  \label{coro1}
		Suppose that $\bm{\omega}\in \bm{\Omega}$, and that both $\bm{\Omega}_0$ and $\bm{\Omega}-\bm{\Omega}_0$ are non empty subsets of $\bm{\Omega}$. Denote the dimensions of $\bm{\Omega}$ and $\bm{\Omega}_0$ as $k$ and $r$ respectively with $k>r$. For testing
		$H_0: \bm{\omega} \in \bm{\Omega}_0~~\text{vs}~~ H_1: \bm{\omega} \in \bm{\Omega}_1=\bm{\Omega}-\bm{\Omega}_0$,
		if we define the log likelihood ratio test statistic as $2\log LR=2\left(\sup \limits_{\bm{\omega} \in \bm{\Omega}}l(\bm{\omega})-\sup \limits_{\bm{\omega} \in \bm{\Omega}_0}l(\bm{\omega})\right),$
		we have $2\log LR\xrightarrow{d} \chi_{k-r}^2$ as $n\to \infty$.
	\end{corollary}
	Thanks to our new modelling framework,  there is no ambiguity regarding the number of parameters in the model, and thus deriving the distributions of various likelihood ratio test (LRT)  statistics is straightforward. 
	We advocate the use of our approach as an appealing competitor to the mixed-effects model approach for analyzing correlated data in practice.

	In this paper, we mainly illustrate our framework with a correctly specified model to avoid digression from the main message that we want to convey. If this assumption is reasonable,  the maximum likelihood estimation method provides a natural approach for parameter estimation with desirable properties as we have shown. On the other hand, when this assumption is questionable, a considerable amount of research has been conducted on the properties of likelihood-based approaches, most notably along the line of the foundational work of \citet{White1982}.  In Section S.5 of the Supplementary Material,  we provide additional simulation studies of the impact of model misspecification. Our conclusion is that caution should be exercised in interpreting results and that there is a clear need for model diagnosis if misspecification is a concern.  We recommend the use of robust alternatives when evidence of assumption violation is detected.
	
	\section{Real Data Analysis}\label{sec:data}
	\subsection{The classroom data}
	\label{dataexample}
	
	We analyze the classroom dataset from a study evaluating math achievement scores conducted by researchers at the University of Michigan \citep{hill2005effects}. 
	The primary goal of this study is to assess the improvement in math, measured by a variable called Mathgain,   for kids in their early ages in relation to  multiple factors.    
	Since this is a  dataset with an interesting clustering structure, elucidating the effects of the covariates on the correlation of Mathgain measures is also of great interest, not only to better model the response of interest but also for understanding the correlation pattern. 
	In a new finding that was not revealed before, we demonstrate via our analysis some interesting co-movements of Mathgain for students with some similar conditions, which may shed light on further studies of this kind.  
	
	In this dataset,  first- and third-grade students were randomly selected  from classrooms in a national U.S. sample of elementary schools. Hence, there are  within-class and within-school clusters, where the former are nested within the latter. After omitting items with missing values, the dataset has 1081  students from 285 classrooms in 105 schools.  The number of students in each school varies from 2 to 31, making this dataset highly unbalanced.  We  consider  Mathgain as the response variable, which measures the change in a student's  math achievement scores from the spring of kindergarten to the spring of first grade.   There are eight  covariates including individual-level, classroom-level, and school-level variables with a mixture of categorical and continuous ones as detailed in Table \ref{tb:5}.

	\begin{table}[ht]
		\begin{center}
			\caption{Covariates in the classroom data}
			\label{tb:5} 
			
			\begin{tabular}{ll|l}
				
				\toprule[1.2pt]
				&Covariate & Description (range) \\
				\specialrule{0.05em}{0pt}{0pt}
				\multirow{4}{*}{Student-level}		&Sex &\multicolumn{1}{m{9cm}}{Indicator variable (0 = boy, 1 = girl)} \\
				&Minority &\multicolumn{1}{m{9cm}}{Indicator variable (0 = nonminority student, 1 = minority student) }\\
				&Mathkind &\multicolumn{1}{m{9cm}}{Student’s math score in the spring of their kindergarten year [290, 629]}\\
				&	Ses&\multicolumn{1}{m{9cm}}{Student socioeconomic status [-1.61, 3.21]}\\
				\hline
				\multirow{5}{*}{Classroom-level}
				&Yearstea &\multicolumn{1}{m{9cm}}{First-grade teacher’s years of teaching experience [0, 40]}\\
				&	Mathprep& \multicolumn{1}{m{9cm}}{ First-grade teacher’s mathematics preparation: number of mathematics content and methods courses [1, 6]}\\
				&	Mathknow &\multicolumn{1}{m{9cm}}{First-grade teacher’s mathematics content knowledge, higher values indicate higher content knowledge [-2.50, 2.61]} \\	
				\hline	
				School-level	&Housepov &\multicolumn{1}{m{9cm}}{Percentage of households in the neighborhood of the school below the poverty level [0.012, 0.564]}\\	
				\bottomrule[1.2pt]
				
			\end{tabular}
			
		\end{center}
	\end{table}

	As a demonstration of linear mixed-effects models, \citet{West:2014} recommended  the following model: 
	\begin{align} \label{eq:6}
		\text{Mathgain}_{ijk}&=\beta_{0}+\beta_{1}\text{Sex}+\beta_{2}\text{Minority}+\beta_{3}\text{Mathkind}+\beta_{4}\text{Ses}   \notag \\
		&+\beta_{5}\text{Yearstea}+\beta_{6}\text{Mathprep}+\beta_{7}\text{Mathknow}+s_{i}+c_{ij}+\varepsilon_{ijk},
	\end{align}
	where $i, j$, and  $k$ index schools, classes, and individuals respectively,  $s_{i}\sim N(0, \sigma_{s}^2)$ and  $c_{ij}\sim N(0, \sigma_{c}^2)$ are  independent random effects with a nested structure respectively capturing the within-school effect and within-class effect nested in schools,  and $\varepsilon_{ijk}\sim N(0, \sigma^2)$ incorporates all the remaining variations.  
	This  model  serves as a benchmark in our analysis. 
	Fitting this linear mixed-effects models via function ${\verb+lmer+}()$ in R package ${\verb+lme4+}$
	gives a log-likelihood $-5160.1$ which has a Akaike information criterion (AIC) value $98.50$. 
	
	The correlations in this case are 
	said to be blocked, a class of structures that broadly applies in analyzing clustered data.  For example,  if school $i$ has two classrooms with the first classroom having two observations and the second three, model  \eqref{eq:6} implies that 
	
	\begin{align}\label{eq:blk}
		\mathbf{R}_i=\left(\begin{array}{cc:ccc}
			1& \rho_1&\rho_2& \rho_2&\rho_2 \\
			\rho_1& 1 &\rho_2&\rho_2& \rho_2\\
			\hdashline
			\rho_2&\rho_2&1&\rho_1&\rho_1\\
			\rho_2&\rho_2&\rho_1&1&\rho_1\\
			\rho_2&\rho_2&\rho_1&\rho_1&1\\
		\end{array}\right), 
	\end{align}
	where the within-block correlation is $\rho_1={(\sigma_{s}^2+\sigma_{c}^2)}/{(\sigma_{s}^2+\sigma_{c}^2+\sigma^2)}$ and the between-block correlation is $\rho_2={\sigma_{s}^2}/{(\sigma_{s}^2+\sigma_{c}^2+\sigma^2})$.

	To compare with the benchmark linear mixed-effects model in (\ref{eq:6}), we will take the same mean model as in (\ref{eq:6}), and  model the log-variances as $\log \sigma^2=\lambda_{0}$ as in \eqref{eq:meanvariance}.   
	For the correlation model, we set
	\begin{equation} \label{eq:7}
		\gamma_{ijk}=\alpha_{0}+\alpha_{1}w_{ijk,1},
	\end{equation}
	where  $i$ is the index of the schools, $j,k$ are the $j$th and $k$th individuals therein,  $w_{ijk,1}=1$ if student $j$ and student $k$ come from the same classroom and $w_{ijk,1}=0$ otherwise. Thus,    in our correlation model,  $\alpha_{0}$ captures the between-school variability and $\alpha_1$ captures the additional variability between classrooms.  This gives rise to the same blocked structure as in \eqref{eq:blk} for $\bm{\gamma}_i$: 
	
	\begin{align}\label{eq:blk1}
		\bm{\gamma}_i=\text{vecl}\left(\begin{array}{cc:ccc}
			*&\alpha_0+\alpha_1&\alpha_0&\alpha_0&\alpha_0\\[8pt]
			\alpha_0+\alpha_1&*&\alpha_0&\alpha_0&\alpha_0\\[4pt]
			\hdashline
			\specialrule{0em}{2pt}{2pt}
			\alpha_0&\alpha_0&*&\alpha_0+\alpha_1&\alpha_0+\alpha_1\\[8pt]
			\alpha_0&\alpha_0&\alpha_0+\alpha_1&*&\alpha_0+\alpha_1 \\[8pt]
			\alpha_0&\alpha_0&\alpha_0+\alpha_1&\alpha_0+\alpha_1&*
		\end{array}\right),
	\end{align}
	and subsequently $\tilde{\mathbf R}_i=f^{-1}(\bm{\gamma}_i)$  is also block-structured. 
	This actually indicates that the generalized z-transformation maintains the blocked structure of the correlation matrix; see also
	\citet{Archakov:2020}. 
	We note here that though  model \eqref{eq:blk1}  implies a block-structured correlation matrix,  $\tilde{\mathbf R}_i$ is not exactly equivalent to $\mathbf{R}_i$ in \eqref{eq:blk} for unbalanced case with unequal block sizes.  Specifically, the within-block correlations of $\tilde{\mathbf{R}}_i$ may vary depending on the block size, whereas ${\mathbf R}_i$ assumes a common within-block correlation uniformly for all blocks as a result from the linear mixed-effect models.  This is a notable consideration to keep in mind when interpreting and comparing results between the two models. 
	
	We report the results for testing the school and classroom effects for modelling correlations  by applying likelihood ratio tests (LRT) for both models.  As can be seen in the first three columns of Table \ref{tb:8}, both the LRTs produce very small $p$-values for testing the school effect and the class effect within school, 
	indicating significant statistical evidence for the school and classroom's contribution to the correlations.

	While the linear mixed-effects model \eqref{eq:6} offers an effective mechanism to model classroom and school effects -- the two grouping factors -- as random effects,  an important yet previously unaddressed question arises: are there other variables that play a role in influencing the  correlations?   In this application, among others of the same kind, investigating  the impact of students' math score in the spring of their kindergarten year (Mathkind) and socioeconomic status (Ses) -- two continuous variables -- is of particular interest. Intuitively, it would seem that students in the same class, especially those with  similar math scores in  their kindergarten and/or similar socioeconomic statuses,  should tend to perform similarly in the future. This is simply because students with similar backgrounds may tend to undertake similar educational paths and interact more with each other. 
	To test this out in an exploratory analysis, we obtain the residuals $\widehat{\epsilon}_{ijk}$ of each student by
	subtracting the fixed effects  estimated from model \eqref{eq:6}.  We standardize $\widehat{\epsilon}_{ijk}$ by dividing it by the fitted standard deviation of model \eqref{eq:6}, obtaining standardized residuals as $\widetilde{\epsilon}_{ijk}$.
	Then we calculate the empirical correlations confined to pre-defined subgroup $S$ as
	\[\widehat{\rho}_{ij}^{S}=\frac{1}{N_{ij}^S}\sum_{k<k^{\prime}}\widetilde{\epsilon}_{ijk}\widetilde{\epsilon}_{ijk'}I(|\text{V}_{ijk}-\text{V}_{ijk'}|\in S),\]
	where $V$ is used as a generic variable, standing for Mathkind or Ses that we are concerned about,   $S$ is created by stratifying their difference, and $N_{ij}^S$ is the total number of different pairs in the subgroup.  We examine  the data by creating three subgroups -- small difference:  $S=(0, 40]$ for Mathkind 
	or  $S=(0, 0.3]$ for Ses; mid difference:  $S=(40, 120]$ for Mathkind or $S=(0.3, 0.6]$ for Ses, and large difference: $S=(120, 240]$ for Mathkind or $S=(0.6, 0.9]$ for Ses.
	The empirical distributions of $\widehat{\rho}_{ij}^{S}$ are shown by boxplots in Figure \ref{fig3}, where clear decreasing trends are seen, indicating that the mathgains  are indeed more correlated among students whose Mathkind or socioeconomic statuses differ relatively less. 
	
	\begin{figure}[htbp]
		\centering
		
		\includegraphics[width=13cm]{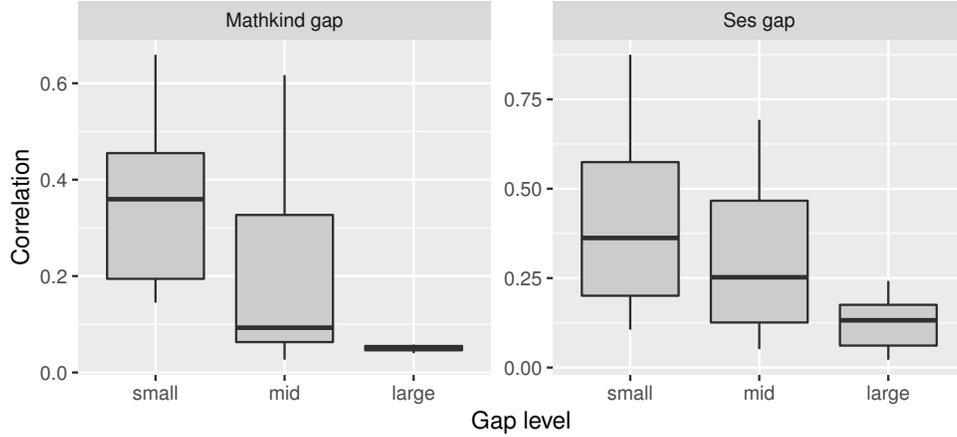}

		\caption{ GZT-Correlogram: boxplots of the pairwise correlations with different Mathkind and Ses gap.}
		\label{fig3}
	\end{figure}

	Motivated by the GZT-correlogram in Figure \ref{fig3},  we examine the effects of the four student-level variables for modelling correlations by expanding \eqref{eq:7} as 
	\begin{equation} \label{eq:m2}
		\gamma_{ijk}=\alpha_{0}+\alpha_{1}w_{ijk,1} +\alpha_{2}w_{ijk,2}.
	\end{equation} 
	In this specification,  each of these variables -- Sex, Minority, Mathkind, Ses -- is assessed individually by varying, respectively,  the definition of the dummy variable  $w_{ijk,2}$ as the  difference between the corresponding values of student $j$ and student $k$.  We then implement the LRT for testing $H_0:\alpha_2=0$, whose $p$-values are reported in the second part of Table  \ref{tb:8}. 
	
	To make a fair comparison, we also attempt a setting for the mixed-effect models comparable to \eqref{eq:m2}. Since including a continuous variable in the mixed-effects model typically means a random slope, which is different from that in \eqref{eq:m2}, we do not report the results when Mathkind or Ses is included. Instead, we focus on the two categorical variables Sex and Minority by  adding one more layer of random effects to \eqref{eq:6} as
	\begin{align}\label{lmm2}
		\text{Mathgain}_{ijkl}=\text{fixed effect} + s_{i}+c_{ij}+v_{ik}+\varepsilon_{ijkl}, 
	\end{align}
	where the subscript  $i$ indicates school, $j$ represents class, $k$ is the level of the categorical variable -- Sex or Minority, and $l$ indexes the subjects in the  respective group.  Here the random effects are assumed independent following $s_i\sim N(0,\sigma^2_s)$, $c_{ij}\sim N(0,\sigma^2_c)$, $v_{ik}\sim N(0,\sigma^2_v)$, and $\varepsilon_{ijkl}\sim N(0,\sigma^2)$. 
	We test the added random effect by validating $H_0: \sigma^2_v=0$ with the results found in Table  \ref{tb:8}. 
	
	The results in Table \ref{tb:8} are quite interesting. While both approaches see adding Sex as a random effect as not necessary, the linear mixed-effects model approach finds Minority marginally significant in modelling random effects, while our approach does not. We remark again from the previous discussion that for the former, testing the existence of random effects has known issues deeply rooted in the nature of such tests - the value under the null hypothesis lies on the boundary of the parameter space and thus the existing LRT is problematic. On the other hand, for our approach, the parameter value under the null in the LRT lies in the interior of its parameter space and thus our test can be easily carried out. More remarkably, this statement carries over easily to the inclusion of continuous variables in modelling correlations. As we can see from Table \ref{tb:8}, it is found that Mathkind is a highly useful variable for depicting correlations. In contrast, the mixed-effects model does not have a natural way of incorporating continuous variables in modelling correlations, unless they are included to have random slopes, which drastically changes the question to be addressed. We further note that the student-level variables do not have to be included in a way nested in school, and  the examinations in other ways are equally applicable by varying the specifications in  \eqref{eq:m2} and \eqref{lmm2}.  Such assessments, not reported here, do not change the conclusions on the effects of those variables. 
	
	\begin{table}[ht]
		\begin{center}
			\caption{$P$-values of LRTs by using  the linear mixed-effects model (LMM) and our proposed approach (Proposed). The null hypotheses are:  $H_0: \sigma^2_s=0$ and $H_0: \alpha_0=0$ for testing the School effect;  $H_0: \sigma^2_c=0$ and $H_0: \alpha_1=0$ in \eqref{eq:6} and \eqref{eq:7} respectively for testing the Class effect nested in the School (School:Class).   The null hypotheses are $H_0:\sigma^2_v=0$ and $H_0:\alpha_2=0$ in \eqref{lmm2} and \eqref{eq:m2} respectively for testing the student-level variables. }
			\label{tb:8} 
			\begin{tabular}{c|cc|cccc}
				\toprule[1.2pt]
				&School & School:Class & Sex & Minority & Mathkind& Ses \\
				\hline
				LMM	 & $1.08\times10^{-11}$&$0.0011^{}$ &$0.2818$& $0.0485$ & $ - $&$ - $   \\
				Proposed &$2.37\times10^{-12}$ &$0.0023^{}$ &$0.2387$& $0.1166$ & $ 0.0027^{}$&$0.1567$  \\ 	
				\bottomrule[1.2pt]
				
			\end{tabular}
			
		\end{center}
	\end{table}
	
	After examining the potential impact of the variables, we conduct  a model selection exercise in the context of our framework by using the  Akaike information criterion (AIC):  
	$
	\operatorname{AIC}=-2 \widehat{l}_{\max } / n+2(p+q+d)  / n,
	$
	where $\widehat{l}_{\max }$ is the maximum of the corresponding log-likelihood. For ease in our demonstration  and considering the finding from Table  \eqref{tb:8} , we narrow  the full model for the correlations to: 
	\begin{equation} \label{eq:8}
		\gamma_{ijk}=\alpha_{0}+\alpha_{1}w_{ijk,1}+\alpha_{2}w_{ijk,2}+\alpha_{3}w_{ijk,3},
	\end{equation}
	where 
	$w_{ijk,2}$  and  $w_{ijk,3}$ are  the differences  between students $j$ and $k$ in their  respective Mathkind and Ses variables.  For the full model on the mean and log-variance as in \eqref{eq:meanvariance},  we include all the variables. 
	The full model gives a log-likelihood $-4153.88$ and AIC value $79.50$;  
	and the optimal model  of our framework chosen by minimizing AIC gives a log-likelihood $-4156.30$ and AIC value $79.45$. 
	These values are better  than those of the benchmark model \eqref{eq:6},   
	indicating  the improvement of  our modelling approach  due to the effective inclusion of the Mathkind or Ses as covariates for the correlations.

	For completeness, we show the parameter estimates and their standard errors in Table \ref{tb:7}. For the mean model part, all the results are similar.  For the linear mixed-effects model,   we are unable to report the standard errors of the random effects parameters as they are not available in the package ${\verb+lme4+}$ used in our implementation.
	
	There are interesting results from our analysis that may be relevant in answering important questions related to assessing the improvement in math for kids in their early ages, the primary goal of this study. 
	First, the coefficients of both  Mathkind and Ses are found negative in our correlation model,  which is consistent with our empirical observation from  Figure \ref{fig3}. In particular, the negative coefficients suggest that for the kids who had similar initial performance, or similar socioeconomic backgrounds,  their later math improvements tend to move in the same direction, provided that all their other conditions are controlled at the same level.  Thus, these variables, especially Mathkind, may help to estimate math performance and quantify its uncertainty more faithfully to the data. 
	Though seemingly intuitive, this finding was not discovered via other existing approaches.   
	Moreover, there are also interesting observations from the log-variance model part. The coefficient of Mathkind is negative, suggesting that kids with  higher math scores in kindergarten show lower  level of variation in their later math improvement. In contrast, the coefficient of Ses is positive, implying that higher socioeconomic status  score is associated with higher variance in measuring the math improvement.
	
	We report another real data example with analyzing longitudinal data, upon
	applying our framework to a malaria immune response data set studied by \citet{Adjakossa:2016},  whose detail is provided in Section S.4 of the Supplementary Material.

	\begin{table}[ht]
		\begin{center}
			\caption{Analysis of the classroom data: The estimated values of parameters and their standard errors. LMM: Linear mixed-effects model (AIC: 98.50); Our full: our full model (AIC: 79.50); 
				Our AIC: Our model chosen by AIC (AIC: 79.45).}
			\label{tb:7} 		\resizebox{\textwidth}{26mm}{
				\begin{tabular}{lrrrrrrrr}
					
					\toprule[1.2pt]
					\multicolumn{9}{l}{Mean: all approaches}  \\
					& Intercept & Sex& Minority &Mathkind& Ses&Yearstea&Mathprep&Mathknow\\	
					\specialrule{0.1em}{3pt}{3pt}
					LMM& $282.02 (11.70)$ &$ -1.34 (1.72)$  &$ -7.87 (2.43)$  &$ -0.48 (0.02)$ &$  5.42 (1.28)$ & $ 0.04 (0.12)$ & $ 1.09 (1.15)$ & $ 1.91 (1.15)$ \\
					Our full& $ 277.05 (12.73)$ &$ -1.62 (1.68)$  &$ -7.46 (2.48)$  &$ -0.46 (0.02)$ & $ 5.13 (1.35)$ & $ 0.05 (0.11)$ & $ 0.94 (1.09)$ & $ 2.20 (1.12)$ \\
					Our AIC& $ 275.60 (12.79)$ &$ -1.18 (1.68)$ &$ -7.22 (2.49)$  &$ -0.46 (0.02)$ & $ 5.20 (1.28)$ & $ 0.06 (0.11)$ & $ 0.90 (1.13)$ & $ 2.16 (1.14)$ \\
					\specialrule{0.1em}{3pt}{3pt}
					\multicolumn{9}{l}{Log-variance: our approach only; results for LMM are not applicable}\\
					& Intercept  & Sex  & Minority  &Mathkind & Ses &Yearstea&Mathprep &Mathknow  \\
					\specialrule{0.1em}{3pt}{3pt}	
					Our full& $8.333 (0.546)$ &$-0.093 (0.083)$ &$-0.188 (0.099)$  &$-0.003 (0.001)$ & $0.108 (0.060)$ & $-0.004 (0.004)$ & $-0.036 (0.043)$ & $-0.044 (0.043)$ \\
					Our AIC& $8.071 (0.520)$  &  &$-0.163 ( 0.096) $  &$ -0.003 (0.001)$ & $0.108 (0.059)$ &  &  &  \\
					\specialrule{0.1em}{3pt}{3pt}
					\multicolumn{9}{l}{Matrix log-correlation: our approach only; results for LMM are not applicable }\\
					& School & Classroom &Mathkind gap & Ses gap & &  & & \\
					\specialrule{0.1em}{3pt}{3pt}
					Our full& $0.112 (0.024)$ &$0.081 (0.025)$ &$-0.00073 (0.0003)$  &$-0.025 (0.019)$ &  &  &  &  \\
					Our AIC& $0.096 (0.019)$ & $0.081 (0.025)$ & $-0.00075 (0.0003)$ & &  &  &  &  \\
					\bottomrule[1.2pt]	
			\end{tabular}}
			
		\end{center}
	\end{table}

	\section{Conclusion}\label{conc}
	We have proposed a novel regression analysis analysis of correlations for general correlated data and illustrated its wide applicability via the analysis of a clustered and a longitudinal dataset. Our model can deal with highly unbalanced clusters and groups and provide a parsimonious characterization of various correlation structures. Our approach builds on the generalized z-transformation that permits unrestricted parameters, to relate quantities in this transformation to covariates via a regression model.  Together with a mean model and a model for the logarithm of the marginal variances, the proposed method represents a flexible and attractive framework with easy and accessible inferential tools rooted in maximum likelihood. Through simulations found in the Supplementary Material and data analysis, we have demonstrated that our modelling framework can be more robust to model misspecification and offer a valuable alternative to the mixed-effects modelling approach.  Additionally, our approach provides a simple and effective means for conducting statistical inference, especially when examining the impact of various factors on correlation structures. This can be a difficult problem for mixed-effects models to handle, making our approach a valuable alternative in such situations.

	Our framework utilizes the likelihood approach.  Though we find it reasonably robust, we advocate caution when there is strong evidence that the distribution assumption is violated. Specifically, in a simulation study reported in the Supplementary Material, we found that the Gaussian likelihood approach works reasonably well when the underlying error distribution deviate moderately away from Gaussianity, but its performance worsens when the errors are much heavy-tailed.  In case there is evidence against the model error distribution, a more robust approach such as the  generalized estimating equations (GEE) can be applied as an appealing alternative.  Related to this, how to develop a more robust alternative approach for estimating the correlation model becomes an interesting open research problem. 
	
	We identify several directions for future work. First, since the parameters in our parametrization are unconstrained, it is natural to model the matrix log-correlations nonparametrically or semiparametrically. Second, we only consider the scenario when the response variable is Gaussian. When departure from normality happens, it will be interesting to consider a wider family of distributions such as multivariate t-distributions for modelling correlated data. Moreover, we have only considered continuous response variables in this paper. It will be interesting to extend the developed framework to deal with categorical responses  in a broad context of generalized linear models. 
	Besides modeling and inference, it is also interesting to extensively investigate the predictions incorporating broad correlation structures, e.g., the kind of the recent study in \citet{Mandel2022}.
	These and other generalizations of the method in this paper will be reported elsewhere.
	
	\begin{acks}[Acknowledgments]
		The authors would like to thank the anonymous referees, an Associate
		Editor and the Editor for their constructive comments that improved the
		quality of this paper. 
	\end{acks}

	\bibliographystyle{imsart-nameyear} 
	\bibliography{paper-ref}

\begin{thebibliography}{29}

\bibitem[\protect\citeauthoryear{Adjakossa et~al.}{2016}]{Adjakossa:2016}
\begin{barticle}[author]
\bauthor{\bsnm{Adjakossa},~\bfnm{Eric~Houngla}\binits{E.~H.}},
  \bauthor{\bsnm{Sadissou},~\bfnm{Ibrahim}\binits{I.}},
  \bauthor{\bsnm{Hounkonnou},~\bfnm{Mahouton~Norbert}\binits{M.~N.}} \AND
  \bauthor{\bsnm{Nuel},~\bfnm{Gregory}\binits{G.}}
(\byear{2016}).
\btitle{Multivariate longitudinal analysis with bivariate correlation test}.
\bjournal{PloS one}
\bvolume{11}
\bpages{e0159649}.
\end{barticle}
\endbibitem

\bibitem[\protect\citeauthoryear{Archakov and Hansen}{2021}]{Archakov:2020}
\begin{barticle}[author]
\bauthor{\bsnm{Archakov},~\bfnm{Ilya}\binits{I.}} \AND
  \bauthor{\bsnm{Hansen},~\bfnm{Peter~Reinhard}\binits{P.~R.}}
(\byear{2021}).
\btitle{A new parametrization of correlation matrices}.
\bjournal{Econometrica}
\bvolume{89}
\bpages{1699--1715}.
\end{barticle}
\endbibitem

\bibitem[\protect\citeauthoryear{Baayen, Davidson and
  Bates}{2008}]{Baayen:2008}
\begin{barticle}[author]
\bauthor{\bsnm{Baayen},~\bfnm{R~Harald}\binits{R.~H.}},
  \bauthor{\bsnm{Davidson},~\bfnm{Douglas~J}\binits{D.~J.}} \AND
  \bauthor{\bsnm{Bates},~\bfnm{Douglas~M}\binits{D.~M.}}
(\byear{2008}).
\btitle{Mixed-effects modeling with crossed random effects for subjects and
  items}.
\bjournal{Journal of memory and language}
\bvolume{59}
\bpages{390--412}.
\end{barticle}
\endbibitem

\bibitem[\protect\citeauthoryear{Barr et~al.}{2013}]{Barr:2013}
\begin{barticle}[author]
\bauthor{\bsnm{Barr},~\bfnm{Dale~J}\binits{D.~J.}},
  \bauthor{\bsnm{Levy},~\bfnm{Roger}\binits{R.}},
  \bauthor{\bsnm{Scheepers},~\bfnm{Christoph}\binits{C.}} \AND
  \bauthor{\bsnm{Tily},~\bfnm{Harry~J}\binits{H.~J.}}
(\byear{2013}).
\btitle{Random effects structure for confirmatory hypothesis testing: Keep it
  maximal}.
\bjournal{Journal of memory and language}
\bvolume{68}
\bpages{255--278}.
\end{barticle}
\endbibitem

\bibitem[\protect\citeauthoryear{Bates}{2006}]{Bates:2006}
\begin{bmisc}[author]
\bauthor{\bsnm{Bates},~\bfnm{Douglas}\binits{D.}}
(\byear{2006}).
\btitle{lmer, p-values and all that}.
\end{bmisc}
\endbibitem

\bibitem[\protect\citeauthoryear{Bates et~al.}{2015a}]{Bates:2015}
\begin{barticle}[author]
\bauthor{\bsnm{Bates},~\bfnm{Douglas}\binits{D.}},
  \bauthor{\bsnm{Kliegl},~\bfnm{Reinhold}\binits{R.}},
  \bauthor{\bsnm{Vasishth},~\bfnm{Shravan}\binits{S.}} \AND
  \bauthor{\bsnm{Baayen},~\bfnm{Harald}\binits{H.}}
(\byear{2015}a).
\btitle{Parsimonious mixed models}.
\bjournal{arXiv preprint arXiv:1506.04967}.
\end{barticle}
\endbibitem

\bibitem[\protect\citeauthoryear{Bates et~al.}{2015b}]{Bates:etal:2014}
\begin{barticle}[author]
\bauthor{\bsnm{Bates},~\bfnm{Douglas}\binits{D.}},
  \bauthor{\bsnm{M{\"a}chler},~\bfnm{Martin}\binits{M.}},
  \bauthor{\bsnm{Bolker},~\bfnm{Ben}\binits{B.}} \AND
  \bauthor{\bsnm{Walker},~\bfnm{Steve}\binits{S.}}
(\byear{2015}b).
\btitle{Fitting linear mixed-effects models using lme4}.
\bjournal{Journal of Statistical Software}
\bvolume{67}
\bpages{1--48}.
\end{barticle}
\endbibitem

\bibitem[\protect\citeauthoryear{Bolker et~al.}{2020}]{GLMM:FAQ}
\begin{bmisc}[author]
\bauthor{\bsnm{Bolker},~\bfnm{Ben}\binits{B.}} \betal{et~al.}
(\byear{2020}).
\btitle{GLMM FAQ}.
\end{bmisc}
\endbibitem

\bibitem[\protect\citeauthoryear{Chen and Liang}{2010}]{ChenLiang:2010}
\begin{barticle}[author]
\bauthor{\bsnm{Chen},~\bfnm{Yong}\binits{Y.}} \AND
  \bauthor{\bsnm{Liang},~\bfnm{Kung-Yee}\binits{K.-Y.}}
(\byear{2010}).
\btitle{On the asymptotic behaviour of the pseudolikelihood ratio test
  statistic with boundary problems}.
\bjournal{Biometrika}
\bvolume{97}
\bpages{603--620}.
\end{barticle}
\endbibitem

\bibitem[\protect\citeauthoryear{Chernoff}{1954}]{Chernoff}
\begin{barticle}[author]
\bauthor{\bsnm{Chernoff},~\bfnm{Herman}\binits{H.}}
(\byear{1954}).
\btitle{On the distribution of the likelihood ratio}.
\bjournal{The Annals of Mathematical Statistics}
\bpages{573--578}.
\end{barticle}
\endbibitem

\bibitem[\protect\citeauthoryear{Crainiceanu and
  Ruppert}{2004}]{Crainiceanu:2004}
\begin{barticle}[author]
\bauthor{\bsnm{Crainiceanu},~\bfnm{Ciprian~M}\binits{C.~M.}} \AND
  \bauthor{\bsnm{Ruppert},~\bfnm{David}\binits{D.}}
(\byear{2004}).
\btitle{Likelihood ratio tests in linear mixed models with one variance
  component}.
\bjournal{Journal of the Royal Statistical Society: Series B}
\bvolume{66}
\bpages{165--185}.
\end{barticle}
\endbibitem

\bibitem[\protect\citeauthoryear{Faraway}{2015}]{Faraway:2015}
\begin{bmisc}[author]
\bauthor{\bsnm{Faraway},~\bfnm{Julian~J}\binits{J.~J.}}
(\byear{2015}).
\btitle{Document Accompanying ``Extending the Linear Model with R''}.
\end{bmisc}
\endbibitem

\bibitem[\protect\citeauthoryear{Fitzmaurice et~al.}{2008}]{Fitzmaurice2008}
\begin{bbook}[author]
\beditor{\bsnm{Fitzmaurice},~\bfnm{Garrett~M.}\binits{G.~M.}},
  \beditor{\bsnm{Davidian},~\bfnm{M.}\binits{M.}},
  \beditor{\bsnm{Verbeke},~\bfnm{G.}\binits{G.}} \AND
  \beditor{\bsnm{Molenberghs},~\bfnm{G.}\binits{G.}}, eds.
(\byear{2008}).
\btitle{Longitudinal Data Analysis}.
\bpublisher{Chapman \& Hall/CRC}.
\end{bbook}
\endbibitem

\bibitem[\protect\citeauthoryear{Hill, Rowan and Ball}{2005}]{hill2005effects}
\begin{barticle}[author]
\bauthor{\bsnm{Hill},~\bfnm{Heather~C}\binits{H.~C.}},
  \bauthor{\bsnm{Rowan},~\bfnm{Brian}\binits{B.}} \AND
  \bauthor{\bsnm{Ball},~\bfnm{Deborah~Loewenberg}\binits{D.~L.}}
(\byear{2005}).
\btitle{Effects of teachers' mathematical knowledge for teaching on student
  achievement}.
\bjournal{American educational research journal}
\bvolume{42}
\bpages{371--406}.
\end{barticle}
\endbibitem

\bibitem[\protect\citeauthoryear{Laird and Ware}{1982}]{Laird:Ware:1982}
\begin{barticle}[author]
\bauthor{\bsnm{Laird},~\bfnm{Nan~M}\binits{N.~M.}} \AND
  \bauthor{\bsnm{Ware},~\bfnm{James~H}\binits{J.~H.}}
(\byear{1982}).
\btitle{Random-effects models for longitudinal data}.
\bjournal{Biometrics}
\bpages{963--974}.
\end{barticle}
\endbibitem

\bibitem[\protect\citeauthoryear{Leng, Zhang and Pan}{2010}]{Leng:2010}
\begin{barticle}[author]
\bauthor{\bsnm{Leng},~\bfnm{Chenlei}\binits{C.}},
  \bauthor{\bsnm{Zhang},~\bfnm{Weiping}\binits{W.}} \AND
  \bauthor{\bsnm{Pan},~\bfnm{Jianxin}\binits{J.}}
(\byear{2010}).
\btitle{Semiparametric mean--covariance regression analysis for longitudinal
  data}.
\bjournal{Journal of the American Statistical Association}
\bvolume{105}
\bpages{181--193}.
\end{barticle}
\endbibitem

\bibitem[\protect\citeauthoryear{Linton and McCrorie}{1995}]{Lint:1995}
\begin{barticle}[author]
\bauthor{\bsnm{Linton},~\bfnm{Oliver}\binits{O.}} \AND
  \bauthor{\bsnm{McCrorie},~\bfnm{J~Roderick}\binits{J.~R.}}
(\byear{1995}).
\btitle{Differentiation of an exponential matrix function}.
\bjournal{Econometric Theory}
\bvolume{11}
\bpages{1182--1185}.
\end{barticle}
\endbibitem

\bibitem[\protect\citeauthoryear{Luke}{2017}]{Luke:2017}
\begin{barticle}[author]
\bauthor{\bsnm{Luke},~\bfnm{Steven~G}\binits{S.~G.}}
(\byear{2017}).
\btitle{Evaluating significance in linear mixed-effects models in R}.
\bjournal{Behavior research methods}
\bvolume{49}
\bpages{1494--1502}.
\end{barticle}
\endbibitem

\bibitem[\protect\citeauthoryear{Mandel, Ghosh and Barnett}{2022}]{Mandel2022}
\begin{barticle}[author]
\bauthor{\bsnm{Mandel},~\bfnm{Francesca}\binits{F.}},
  \bauthor{\bsnm{Ghosh},~\bfnm{Riddhi~Pratim}\binits{R.~P.}} \AND
  \bauthor{\bsnm{Barnett},~\bfnm{Ian}\binits{I.}}
(\byear{2022}).
\btitle{Neural networks for clustered and longitudinal data using mixed effects
  models}.
\bjournal{Biometrics}
\bvolume{to appear}.
\bdoi{10.1111/biom.13615}
\end{barticle}
\endbibitem

\bibitem[\protect\citeauthoryear{M{\"u}ller, Scealy and
  Welsh}{2013}]{Muller:2013}
\begin{barticle}[author]
\bauthor{\bsnm{M{\"u}ller},~\bfnm{Samuel}\binits{S.}},
  \bauthor{\bsnm{Scealy},~\bfnm{Janice~L}\binits{J.~L.}} \AND
  \bauthor{\bsnm{Welsh},~\bfnm{Alan~H}\binits{A.~H.}}
(\byear{2013}).
\btitle{Model selection in linear mixed models}.
\bjournal{Statistical Science}
\bvolume{28}
\bpages{135--167}.
\end{barticle}
\endbibitem

\bibitem[\protect\citeauthoryear{Pan and Mackenzie}{2003}]{Pan:2003}
\begin{barticle}[author]
\bauthor{\bsnm{Pan},~\bfnm{Jianxin}\binits{J.}} \AND
  \bauthor{\bsnm{Mackenzie},~\bfnm{Gilbert}\binits{G.}}
(\byear{2003}).
\btitle{On modelling mean-covariance structures in longitudinal studies}.
\bjournal{Biometrika}
\bvolume{90}
\bpages{239--244}.
\end{barticle}
\endbibitem

\bibitem[\protect\citeauthoryear{Pourahmadi}{1999}]{Pourah:1999}
\begin{barticle}[author]
\bauthor{\bsnm{Pourahmadi},~\bfnm{Mohsen}\binits{M.}}
(\byear{1999}).
\btitle{Joint mean-covariance models with applications to longitudinal data:
  Unconstrained parameterisation}.
\bjournal{Biometrika}
\bvolume{86}
\bpages{677--690}.
\end{barticle}
\endbibitem

\bibitem[\protect\citeauthoryear{Pourahmadi}{2000}]{Pourah:2000}
\begin{barticle}[author]
\bauthor{\bsnm{Pourahmadi},~\bfnm{Mohsen}\binits{M.}}
(\byear{2000}).
\btitle{Maximum likelihood estimation of generalised linear models for
  multivariate normal covariance matrix}.
\bjournal{Biometrika}
\bvolume{87}
\bpages{425--435}.
\end{barticle}
\endbibitem

\bibitem[\protect\citeauthoryear{Self and Liang}{1987}]{Self:Liang:1987}
\begin{barticle}[author]
\bauthor{\bsnm{Self},~\bfnm{Steven~G}\binits{S.~G.}} \AND
  \bauthor{\bsnm{Liang},~\bfnm{Kung-Yee}\binits{K.-Y.}}
(\byear{1987}).
\btitle{Asymptotic properties of maximum likelihood estimators and likelihood
  ratio tests under nonstandard conditions}.
\bjournal{Journal of the American Statistical Association}
\bvolume{82}
\bpages{605--610}.
\end{barticle}
\endbibitem

\bibitem[\protect\citeauthoryear{West, Welch and Galecki}{2006}]{West:2014}
\begin{bbook}[author]
\bauthor{\bsnm{West},~\bfnm{Brady~T}\binits{B.~T.}},
  \bauthor{\bsnm{Welch},~\bfnm{Kathleen~B}\binits{K.~B.}} \AND
  \bauthor{\bsnm{Galecki},~\bfnm{Andrzej~T}\binits{A.~T.}}
(\byear{2006}).
\btitle{Linear Mixed Models: A Practical Guide Using Statistical Software}.
\bpublisher{CRC Press}.
\end{bbook}
\endbibitem

\bibitem[\protect\citeauthoryear{White}{1982}]{White1982}
\begin{barticle}[author]
\bauthor{\bsnm{White},~\bfnm{Halbert}\binits{H.}}
(\byear{1982}).
\btitle{Maximum likelihood estimation of misspecified models}.
\bjournal{Econometrica}
\bvolume{50}
\bpages{1--25}.
\bdoi{10.2307/1912526}
\end{barticle}
\endbibitem

\bibitem[\protect\citeauthoryear{Ye and Pan}{2006}]{Ye:2006}
\begin{barticle}[author]
\bauthor{\bsnm{Ye},~\bfnm{Huajun}\binits{H.}} \AND
  \bauthor{\bsnm{Pan},~\bfnm{Jianxin}\binits{J.}}
(\byear{2006}).
\btitle{Modelling covariance structures in generalized estimating equations for
  longitudinal data}.
\bjournal{Biometrika}
\bvolume{93}
\bpages{927--941}.
\end{barticle}
\endbibitem

\bibitem[\protect\citeauthoryear{Zhang and Leng}{2012}]{Zhang:2012}
\begin{barticle}[author]
\bauthor{\bsnm{Zhang},~\bfnm{Weiping}\binits{W.}} \AND
  \bauthor{\bsnm{Leng},~\bfnm{Chenlei}\binits{C.}}
(\byear{2012}).
\btitle{A moving average Cholesky factor model in covariance modelling for
  longitudinal data}.
\bjournal{Biometrika}
\bvolume{99}
\bpages{141--150}.
\end{barticle}
\endbibitem

\bibitem[\protect\citeauthoryear{Zhang, Leng and Tang}{2015}]{Zhang:2015}
\begin{barticle}[author]
\bauthor{\bsnm{Zhang},~\bfnm{Weiping}\binits{W.}},
  \bauthor{\bsnm{Leng},~\bfnm{Chenlei}\binits{C.}} \AND
  \bauthor{\bsnm{Tang},~\bfnm{Cheng~Yong}\binits{C.~Y.}}
(\byear{2015}).
\btitle{A joint modelling approach for longitudinal studies}.
\bjournal{Journal of the Royal Statistical Society: Series B}
\bvolume{77}
\bpages{219--238}.
\end{barticle}
\endbibitem

\end{thebibliography}

\newpage

\renewcommand\thesection{S}

\begin{supplement}
\stitle{}
\sdescription{}
This Supplementary Material contains the algorithm for fitting our model,  technical proofs, an application of malaria immune response data in Benin, and simulation. 

\subsection{The algorithm} 

First, recall that the log-likelihood is given by
\begin{equation} \label{eq2}
	l(\bm{\omega})=-\frac{1}{2}\sum_{i=1}^{n} \left( \log \left|\mathbf{D}_{i} \mathbf{R}_{i} \mathbf{D}_{i}\right|+ \bm{\nu}_{i}^{\prime} \mathbf{D}_{i}^{-1} \mathbf{R}_{i}^{-1} \mathbf{D}_{i}^{-1} \bm{\nu}_{i} \right).
\end{equation}

Let $\bm{\Delta}_{i}=\bm{\Delta}_{i}\left(\mathbf{X}_{i} \bm{\beta}\right)=\operatorname{diag}\left\{\dot{g}^{-1}\left(\mathbf{x}_{i1}^{\prime} \bm{\beta}\right), \ldots, \dot{g}^{-1}\left(\mathbf{x}_{i m_{i}}^{\prime} \bm{\beta}\right)\right\}$ where $\dot{g}^{-1}(\cdot)$ is
the derivative of the inverse link function $g^{-1}(\cdot)$ and we note that $\mu(\cdot)=g^{-1}(\cdot)$. Define $\widehat{\mathbf{R}}_i=\mathbf{D}_i^{-1}\bm{\nu}_i \bm{\nu}_i^{\prime}\mathbf{D}_i^{-1}$, and $\mathbf{h}_i=\operatorname{diag}(\mathbf{R}_i^{-1}\widehat{\mathbf{R}}_i)$. Then the following score equations based on the log-likelihood (\ref{eq2}) can be obtained by direct calculation:
\begin{align} \label{eq:3}
	\mathbf{S}_{1}(\bm{\beta};\bm{\alpha},\bm{\lambda} )&=\sum_{i=1}^{n} \mathbf{X}_{i}^{\prime} \bm{\Delta}_{i} \bm{\Sigma}_{i}^{-1}\left(\mathbf{y}_{i}-\boldsymbol{\mu}_{i}\right)=\mathbf{0}, \nonumber \\
	\mathbf{S}_{2}(\bm{\alpha};\bm{\beta},\bm{\lambda} )&=\sum_{i=1}^{n}\mathbf{W}_{i}^{\prime} (\frac{\partial \bm{\rho}_i}{\partial \bm{\gamma}_i})^{\prime} \operatorname{vecl} \left(\mathbf{R}_i^{-1}\widehat{\mathbf{R}}_i\mathbf{R}_i^{-1}-\mathbf{R}_i^{-1} \right) =\mathbf{0}, \\ \nonumber
	\mathbf{S}_{3}(\bm{\lambda};\bm{\beta},\bm{\alpha} )&=\frac{1}{2} \sum_{i=1}^{n} \mathbf{Z}_{i}^{\prime}\left(\mathbf{h}_{i}-\mathbf{1}_{m_{i}}\right)=\mathbf{0},
\end{align}
where $\mathbf{X}_{i}, \mathbf{W}_{i}$ and $\mathbf{Z}_{i}$ are respectively the $m_i\times p, m_i(m_i-1)/2\times d$ and $m_i \times q $ matrices that contain the relevant observed covariates, and $\mathbf{1}_{m_{i}}$ is the $m_i\times 1$ vector with elements 1.

We define the negative expected Hessian matrix $\mathbf{I}(\bm{\omega})=-\mathbb{E}(\frac{\partial^2 l}{ \partial\bm{\omega}\partial \bm{\omega}^{\prime}} )$. Following  (\ref{eq:3}),  the block expression of $\mathbf{I}(\bm{\omega})$ satisfy
\begin{align*}
	&\mathbf{I}_{11}(\bm{\omega})=\sum_{i=1}^{n} \mathbf{X}_{i}^{\prime} \bm{\Delta}_{i} \bm{\Sigma}_{i}^{-1} \bm{\Delta}_{i} \mathbf{X}_{i},
	\mathbf{I}_{22}(\bm{\omega})=\sum_{i=1}^{n}\mathbf{W}_{i}^{\prime} (\frac{\partial \bm{\rho}_i}{\partial \bm{\gamma}_i})^{\prime} \mathbf{J}_i \frac{\partial \bm{\rho}_i}{\partial \bm{\gamma}_i} \mathbf{W}_{i},  \\
	& \mathbf{I}_{33}(\boldsymbol{\bm{\omega}})=\frac{1}{4} \sum_{i=1}^{n} \mathbf{Z}_{i}^{\prime}\left(\mathbf{R}_{i}^{-1} \circ \mathbf{R}_{i}+\mathbf{I}_{m_{i}}\right) \mathbf{Z}_{i},
	\mathbf{I}_{12}(\bm{\omega})=\mathbf{I}_{21}^{\prime}(\bm{\omega})=\mathbf{0}, \\
	& \mathbf{I}_{13}(\bm{\omega})=\mathbf{I}_{31}^{\prime}(\bm{\omega})=\mathbf{0},
	\mathbf{I}_{23}(\bm{\omega})=\mathbf{I}_{32}^{\prime}(\bm{\omega})=\frac{1}{2} \sum_{i=1}^{n}\mathbf{W}_{i}^{\prime} (\frac{\partial \bm{\rho}_i}{\partial \bm{\gamma}_i})^{\prime} \mathbf{H}_i \mathbf{Z}_i,
\end{align*}
where $\circ$ denotes the Hadamard product. Denote by $\bm{\eta}_i=\operatorname{vecl} (\mathbf{R}_i^{-1}\widehat{\mathbf{R}}_i\mathbf{R}_i^{-1}-\mathbf{R}_i^{-1} )=(\eta_{ijk}), 1\le k<j\le m_i$ and $\bm{\phi}_i=\mathbf{h}_{i}-\mathbf{1}_{m_{i}}=(\phi_{il}), 1\le l\le m_i $ for $i=1, \dots, n$. Then, the $\frac{m_i(m_i-1)}{2}\times \frac{m_i(m_i-1)}{2}$ matrix $\mathbf{J}_i$ in $\mathbf{I}_{22}(\bm{\omega})$  and $\frac{m_i(m_i-1)}{2}\times m_i$ matrix $\mathbf{H}_i$ in $\mathbf{I}_{23}(\bm{\omega})$ can be respectively expressed as $\mathbf{J}_i=\mathbb{E}(\bm{\eta}_i \bm{\eta}_i^\prime) $ and $\mathbf{H}_i=\mathbb{E}(\bm{\eta}_i \bm{\phi}_i^\prime) $, since the negative expected Hessian matrix is equate to the Fisher information matrix $\mathbf{I}(\bm{\omega})=\mathbb{E}(\frac{\partial l} {\partial\bm{\omega}}\frac{ \partial l}{\partial \bm{\omega}}^{\prime} )$. The calculation of each element of $\mathbf{J}_i$ and $\bm{H}_i$ is given in Lemma \ref{pro1} below.

We then estimate $\bm{\omega}$ by maximizing the log-likelihood (\ref{eq2}) via an iterative Newton-Raphson algorithm.  An application of the quasi-Fisher scoring algorithm on Equation (\ref{eq:3}) directly yields the numerical solutions for these parameters. Since the  negative expected Hessian matrix $\mathbf{I}(\bm{\omega})$ is block diagonal consisting of one block corresponding to $\bm{\beta}$ and the other to $\bm{\alpha}$ and $\bm{\lambda}$, it is natural to iterate between updating $\bm{\beta}$ and $(\bm{\alpha}^\prime, \bm{\lambda}^\prime)^\prime$. The computation needed to find the solution is summarized in Algorithm \ref{alg1}.

\begin{algorithm}[h]
	\caption{Quasi-Fisher Scoring Algorithm}
	\label{alg1}
	\begin{algorithmic}[1]
		\Require
		Starting value: $\bm{\beta}^{(0)}, \bm{\alpha}^{(0)}$ and $\bm{\lambda}^{(0)}$, set $k=0$,
		\Ensure  An estimate of
		$\bm{\omega}$.
		
		\Repeat
		\State Compute $\bm{\Sigma}_{i}$ by using $\bm{\alpha}^{(k)}$ and $\bm{\lambda}^{(k)}$. Update $\bm{\beta}^{(k+1)}$ as
		\begin{equation}
			\bm{\beta}^{(k+1)}=\bm{\beta}^{(k)}+\left.\mathbf{I}_{11}^{-1}(\bm{\omega}) \mathbf{S}_{1}(\bm{\beta}; \bm{\alpha},\bm{\lambda} )\right|_{\bm{\beta}=\bm{\beta}^{(k)}}. \notag
		\end{equation}
		\State Given $\bm{\beta}=\bm{\beta}^{(k+1)}$, update $\bm{\alpha}^{(k+1)}$ and $\bm{\lambda}^{(k+1)}$ by using
		\begin{equation}
			\left(\begin{array}{l}
				\bm{\alpha}^{(k+1)}\\
				\bm{\lambda}^{(k+1)}
			\end{array}\right)=\left(\begin{array}{l}
				\bm{\alpha}^{(k)} \\
				\bm{\lambda}^{(k)}
			\end{array}\right)+\left.\left[\left(\begin{array}{ll}
				\mathbf{I}_{22}(\bm{\omega}) & \mathbf{I}_{23}(\bm{\omega}) \\
				\mathbf{I}_{32}(\bm{\omega}) & \mathbf{I}_{33}(\bm{\omega})
			\end{array}\right)^{-1}\left(\begin{array}{l}
				\mathbf{S}_{2}(\bm{\alpha} ; \bm{\beta}, \bm{\lambda}) \\
				\mathbf{S}_{3}(\bm{\lambda} ; \bm{\beta}, \bm{\alpha})
			\end{array}\right)\right]\right|_{\bm{\alpha}=\bm{\alpha}^{(k)}, \bm{\lambda}=\bm{\lambda}^{(k)}}. \notag
		\end{equation}
		\State Set $k=k+1$.
		\Until{ a desired convergence criterion is met. }
	\end{algorithmic}
\end{algorithm}

Since the likelihood function is not a global convex function of the parameters on their support, it can only be guaranteed that the algorithm converges to a local optimum. 
Similar to the conventional theory for the MLE,  strictly speaking, the properties in Theorem 2.1 are established for the so-called  consistent root of the likelihood score equation.  In practice, we advocate using multiple initial values to ensure that the algorithm find an optimal solution. 

To choose one reasonable initial value, we can take $\bm{\Sigma}_i$ as identity matrices initially and use the least-squares estimator as the
initial value of $\bm{\beta}$ in the first equation of (\ref{eq2}). Then we initiate $\bm{\alpha}$ and $\bm{\lambda}$ using the least-squares estimation based on the residuals.  Our numerical experience shows that this iterative algorithm converges very quickly {\color{black} if we stop the iteration when $\parallel \omega^{(k+1)}-\omega^{(k)}\parallel <1e^{-7}$}, usually in a few iterations. {\color{black} In addition, we also tried using different initial values and found that the results were not affected, unless  extreme initial values that are far away from the true parameters were used.}

\subsection{Properties}\label{sec:prop}

We summarize the computation of $\mathbf{J}_i$ and $\bm{H}_i$ in the negative expected Hessian matrix in the following lemma:
\begin{lemma} \label{pro1}
	Let $a_{ijk}$ be the $(j, k)$th element of $\mathbf{R}_i^{-1}$. Then the $(\frac{(2n-k)(k-1)}{2}+j-k, \frac{(2n-s)(s-1)}{2}+l-s)$th element of $\mathbf{J}_i$ is given by $\mathbb{E}(\eta_{ijk}\eta_{ils})=a_{ijl}a_{iks}+a_{ijs}a_{ikl}  (1\le k<j\le m_i; 1\le s<l\le m_i)$, and the $(\frac{(2n-k)(k-1)}{2}+j-k, l)$th element of $\mathbf{H}_i$ is given by $\mathbb{E}(\eta_{ijk}\phi_{il})=a_{ijl}\delta_{jl}+a_{ikl}\delta_{kl}  (1\le k<j\le m_i; 1\le l\le m_i) $, where $ \delta_{jk}$ is unity when $j=k$ and zero otherwise.
\end{lemma}

The proof of this lemma will be given later. We first present 
some useful formula. 
We start from a formula  in \citet{Archakov:2020}:    \[\frac{\partial \bm{\rho}}{\partial \bm{\gamma}}=\mathbf{E}_{l}\left(\mathbf{I}-\mathbf{A} \mathbf{E}_{d}^{T}\left(\mathbf{E}_{d} \mathbf{A} \mathbf{E}_{d}^{T}\right)^{-1} \mathbf{E}_{d}\right) \mathbf{A}\left(\mathbf{E}_{l}+\mathbf{E}_{u}\right)^{T},\] where $\mathbf{A}= \partial \operatorname{vec} \mathbf{R}/\partial \operatorname{vec} \mathbf{G}$, $\operatorname{vec}$ is the matrix column vectorization operator, $\mathbf{G}=\log \mathbf{R} $ and
the matrices $\mathbf{E}_{l}, \mathbf{E}_{u}$ and $\mathbf{E}_{d}$ are elimination matrices, such that $\operatorname{vecl} \mathbf{R}=\mathbf{E}_{l} \operatorname{vec} \mathbf{R}, \operatorname{vecl} \mathbf{R}^{T}=\mathbf{E}_{u} \operatorname{vec} \mathbf{R}$
and $\operatorname{diag} \mathbf{R}=\mathbf{E}_{d} \operatorname{vec} \mathbf{R}$.	
Let $\mathbf{G}=\mathbf{Q} \bm{\Lambda} \mathbf{Q}^{T}$, where $\bm{\Lambda}$ is the diagonal matrix containing the eigenvalues $\lambda_{1}, \ldots, \lambda_{m}$ of $\mathbf{G}$,  and $\mathbf{Q}$ is an orthonormal matrix (i.e. $\mathbf{Q}^{T}=\mathbf{Q}^{-1}$ ) containing the corresponding eigenvectors. 

Following \citet{Lint:1995},    we have $\partial \operatorname{vec} \mathbf{R} =\mathbf{A} \partial \operatorname{vec} \mathbf{G}$, where
$$
\mathbf{A} =(\mathbf{Q} \otimes \mathbf{Q}) \bm{\Xi} (\mathbf{Q}\otimes \mathbf{Q})^{T}
$$
is a $m^{2} \times m^{2}$ matrix and $\bm{\Xi}$ is an $m^{2} \times m^{2}$ diagonal matrix whose elements are given by
$$
\xi_{ jk}=\bm{\Xi}_{ (j-1) m+k, (j-1) m+k}=\left\{\begin{array}{lll}
	e^{\lambda_{j}}, & \text { if } & \lambda_{j}=\lambda_{k} \\
	\frac{e^{\lambda_{j}}-e^{\lambda_{k}}}{\lambda_{j}-\lambda_{k}}, & \text { if } & \lambda_{j} \neq \lambda_{k}
\end{array}\right.
$$
for $j, k=1, \ldots, m$. Clearly, we have $\xi_{jk}=\xi_{kj}$ for all $(j,k)$. Moreover, $\mathbf{A}$ is a symmetric positive definite matrix, because all the diagonal elements of $\bm{\Xi}$ are strictly positive. For convenience, in the following proofs, we use $(\mathbf{M})_{ij}$  for the $(i,j)$th element of the matrix $\mathbf{M}$.

\begin{proof}[Proof of Proposition 1]
	Since  $\mathbf{A}$ is a symmetric positive definite matrix, we can decompose $\mathbf{A}$ as $\mathbf{A}=\mathbf{Q}\mathbf{Q}^{T}$, where $\mathbf{Q}$ is a positive matrix. Then, we have
	$$\mathbf{I}-\mathbf{A} \mathbf{E}_{d}^{T}\left(\mathbf{E}_{d} \mathbf{A}\mathbf{E}_{d}^{T}\right)^{-1} \mathbf{E}_{d}=\mathbf{I}-\mathbf{Q}\mathbf{Q}^{T} \mathbf{E}_{d}^{T}\left(\mathbf{E}_{d} \mathbf{Q}\mathbf{Q}^{T} \mathbf{E}_{d}^{T}\right)^{-1} \mathbf{E}_{d}\mathbf{Q}\mathbf{Q}^{-1},$$
	where $\mathbf{Q}^{T} \mathbf{E}_{d}^{T}\left(\mathbf{E}_{d} \mathbf{Q}\mathbf{Q}^{T} \mathbf{E}_{d}^{T}\right)^{-1} \mathbf{E}_{d}\mathbf{Q}$ is a symmetric idempotent matrix whose eigenvalues are either 0 or 1. Therefore, the eigenvalues of $\mathbf{Q}\mathbf{Q}^{T} \mathbf{E}_{d}^{T}\left(\mathbf{E}_{d} \mathbf{Q}\mathbf{Q}^{T} \mathbf{E}_{d}^{T}\right)^{-1} \mathbf{E}_{d}\mathbf{Q}\mathbf{Q}^{-1}$ are 0 or 1, which implies that the eigenvalues of $\mathbf{I}-\mathbf{A} \mathbf{E}_{d}^{T}\left(\mathbf{E}_{d} \mathbf{A} \mathbf{E}_{d}^{T}\right)^{-1} \mathbf{E}_{d}$ are 0 or 1, too. Denote by $\mathbf{B}=\mathbf{A}-\mathbf{A} \mathbf{E}_{d}^{T}\left(\mathbf{E}_{d} \mathbf{A} \mathbf{E}_{d}^{T}\right)^{-1} \mathbf{E}_{d}\mathbf{A}$. Thus, $\mathbf{B}$ is a semi-positive definite symmetric matrix.
	
	Since $\partial \bm{\rho}/\partial \bm{\gamma} =\mathbf{E}_{l}\mathbf{B}\left(\mathbf{E}_{l}+\mathbf{E}_{u}\right)^{T}$, the diagonal element $\partial \rho_{jk}/\partial \gamma_{jk}$ $(1\le k<j\le m)$ is given by
	$$ \frac{\partial\rho_{jk}}{\partial\gamma_{jk}}=\mathbf{B}_{j+m(k-1),j+m(k-1)}+\mathbf{B}_{j+m(k-1),k+m(j-1)}.$$
	Note that $\mathbf{B}$ is also the Jacobian of $\partial \operatorname{vec} \mathbf{R}/\partial \operatorname{vec} \mathbf{G}$, but the diagonal elements of $\mathbf{R}$ are constrained to one. Because of the symmetry of $\partial \operatorname{vec} \mathbf{R}/\partial \operatorname{vec} \mathbf{G}$, it is easy to verify that
	$$\mathbf{B}_{j+m(k-1),j+m(k-1)}=\mathbf{B}_{k+m(j-1),k+m(j-1)}.$$
	The semi-positive definiteness of $\mathbf{B}$ implies that the principal sub-matrix of $\mathbf{B}$  satisfies
	$$\mathbf{B}_{j+m(k-1),j+m(k-1)}\mathbf{B}_{k+m(j-1),k+m(j-1)}\ge \mathbf{B}_{j+m(k-1),k+m(j-1)}^2,$$
	so that
	$$\mathbf{B}_{j+m(k-1),j+m(k-1)}^2\ge\mathbf{B}_{j+m(k-1),k+m(j-1)}^2,$$
	which implies that $\partial\rho_{jk}/\partial\gamma_{jk}=\mathbf{B}_{j+m(k-1),j+m(k-1)}+\mathbf{B}_{j+m(k-1),k+m(j-1)}\ge 0.$
\end{proof}

As for $\partial \rho_{jk}/\partial \gamma_{lm}$ $(j\neq l, k\neq m)$, their relationships are case by case,  depending on the correlation matrix itself. For example, for correlation matrix  AR(0.5) with dimension of 3, we have
$$
\mathbf{R}=\left(\begin{array}{ccc}
	1&0.5&0.25\\
	0.5&1&0.5\\
	0.25&0.5&1
\end{array}\right)  \rightarrow 	
\frac{\partial \bm{\rho}}{\partial \bm{\gamma}}=\left(\begin{array}{ccc}
	0.736&0.188&0.014\\
	0.188&0.910&0.188\\
	0.014&0.188&0.736
\end{array}\right),
$$
while for AR(-0.5) we have
$$
\mathbf{R}=\left(\begin{array}{ccc}
	1&-0.5&0.25\\
	-0.5&1&-0.5\\
	0.25&-0.5&1
\end{array}\right)  \rightarrow 	
\frac{\partial \bm{\rho}}{\partial \bm{\gamma}}=\left(\begin{array}{ccc}
	0.736&-0.188&0.014\\
	-0.188&0.910&-0.188\\
	0.014&-0.188&0.736
\end{array}\right).
$$

Since the one-to-one correspondence between the matrix log-correlation parameters and the correlations  is only through a matrix generalized z-transformation, in general their analytical expressions can be complicated. In practice,  we recommend examining this correspondence via numerical evaluation. As examples,  here we plot the associations between $\gamma$ and $\rho$ in Figure \ref{fig6} for three commonly used correlations structures: exchangeable,  AR(1), and banded (whose non-zero entries are confined to a diagonal band of width 1) structure with parameter $\rho$. For example, when the dimension $m$ is 3, the correlation matrices are respectively 
$$\left(\begin{array}{ccc}
	1&\rho&\rho\\
	\rho&1&\rho\\
	\rho&\rho&1\\
\end{array}\right), \qquad \left(\begin{array}{ccc}
	1&\rho&\rho^2\\
	\rho&1&\rho\\
	\rho^2&\rho&1\\
\end{array}\right),  \quad \text{and} \quad
\left(\begin{array}{ccc}
	1&\rho&0\\
	\rho&1&\rho\\
	0&\rho&1\\
\end{array}\right). 
$$ 
We note that $(2,1)$th components of all these three matrices are the same. As the representative case, we plot the corresponding $(2,1)$th element after the matrix-log transformation versus $\rho$ in Figure \ref{fig6}. 
We vary the dimension of the corresponding correlation matrices as $m=2$, $5$, $10$ or $20$. 

There are clear monotonic relationships between the parameters for all the three structures; the patterns of monotonicity are not identical. In addition, we found that $m$ plays a role in determining the extent of the monotonicity. These observations were consistent across all components in the correlation matrices and their matrix-log transformations, as further analysis showed (results not presented here). Overall, these findings align with the discussion presented in the paper.

Additionally, it is important to note that the relationships between the $(i,j)$th component of the correlation matrix and its corresponding $(k,l)$th component of the matrix-log transformation can vary on a case-by-case basis, unless $i=j$ and  $k=l$, as  discussed in the paper. To illustrate this point, we present Figure \ref{fig7} which displays the $(3,1)$th component of the matrix-log transformation plotted against the $(2,1)$th component of the correlation matrix for both the AR(1) and banded structures. The exchangeable structure has the same pattern shown in Figure \ref{fig7} and thus is not plotted. Notably, the patterns for the AR(1) and banded structures are opposite to each other, with one increasing while the other decreasing.

\begin{figure}[htbp]
	\centering
	\subfigure[Exchangeable ]{
		\includegraphics[width=0.3\textwidth]{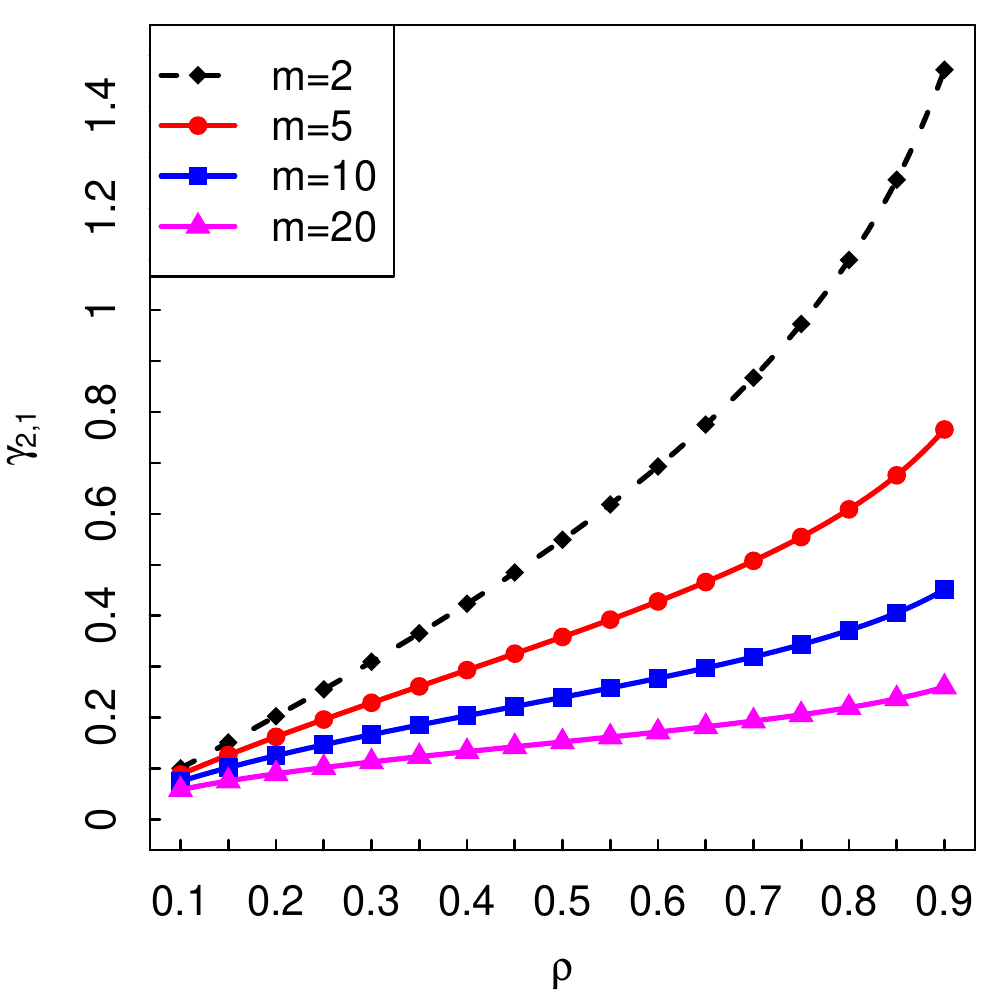}
	}	\quad
	\subfigure[AR(1)]{
		\includegraphics[width=0.3\textwidth]{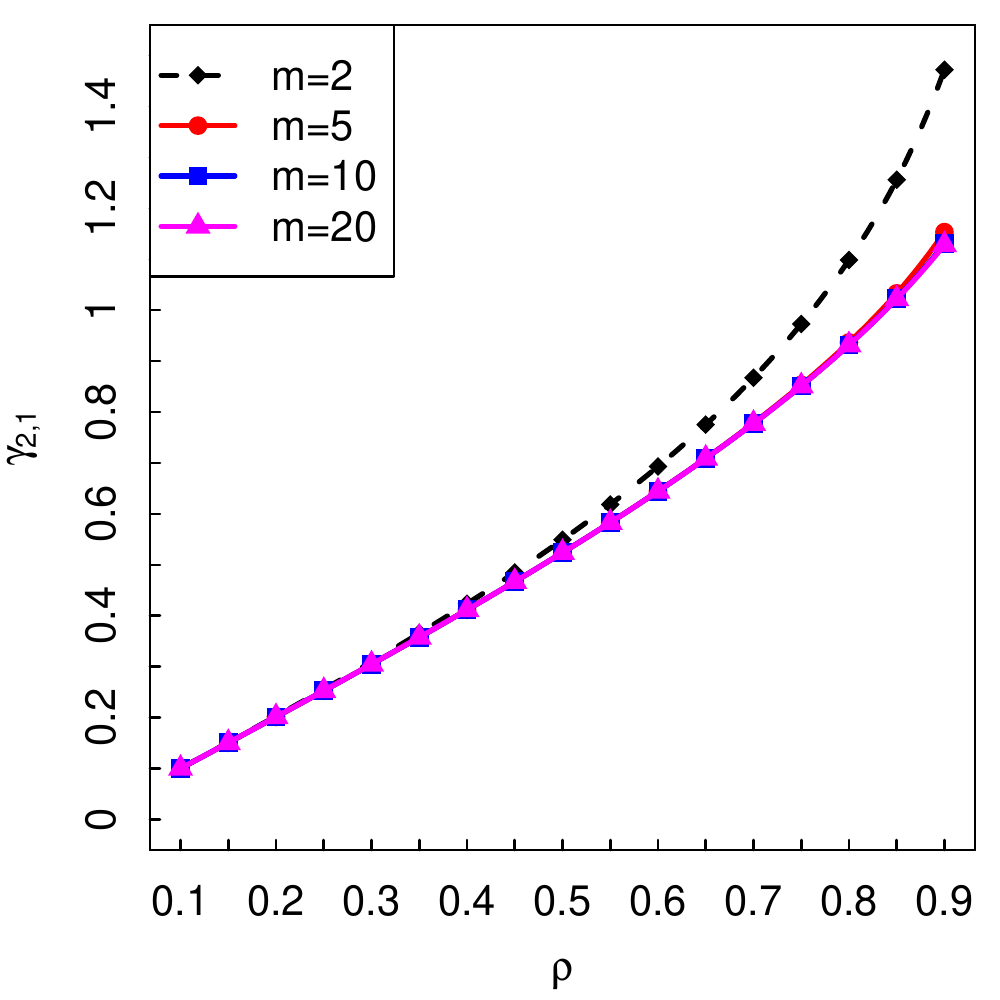}
	}\quad
	\subfigure[Banded]{
		\includegraphics[width=0.3\textwidth]{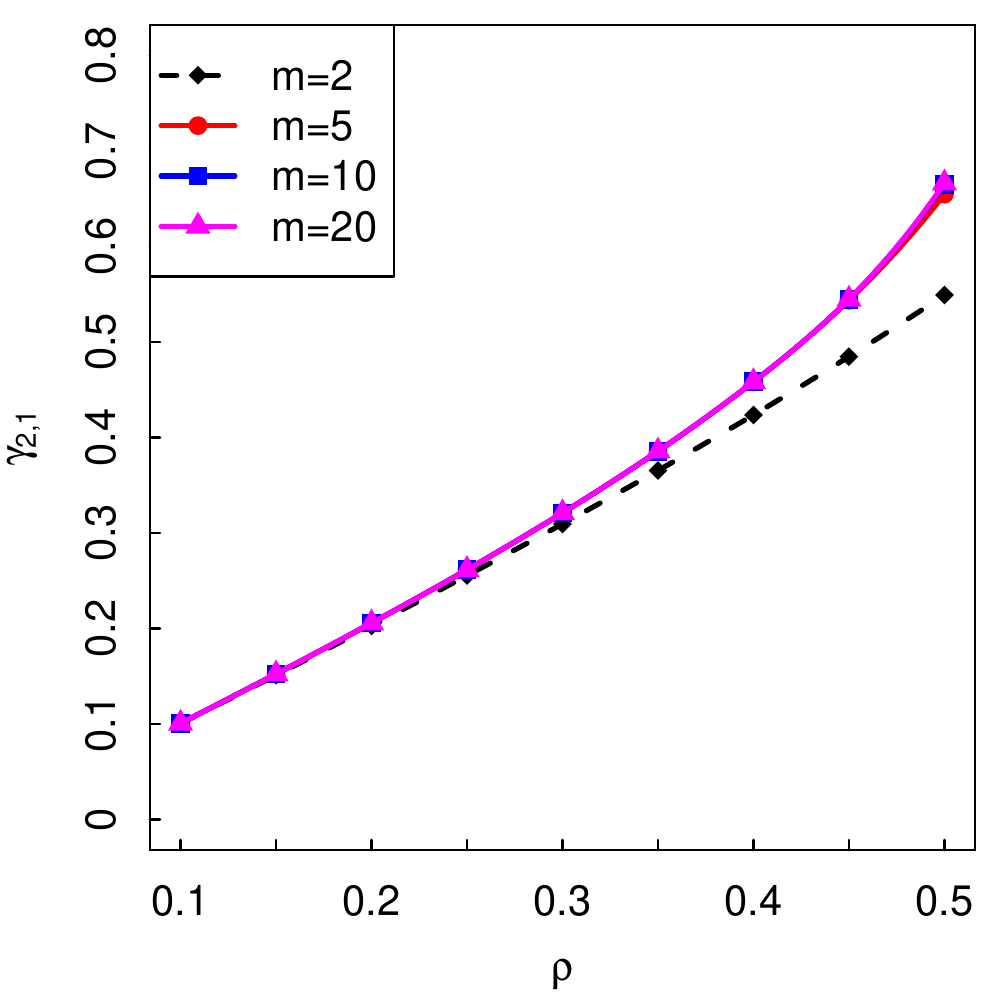}
	}
	\caption{ {$\{\log(\mathbf{R})\}_{2,1}$ versus $\mathbf{R}_{2,1}$  under three commonly used correlation structures. }}
	\label{fig6}
\end{figure}

\begin{figure}[htbp]
	\centering
	\subfigure[AR(1)]{
		\includegraphics[width=0.3\textwidth]{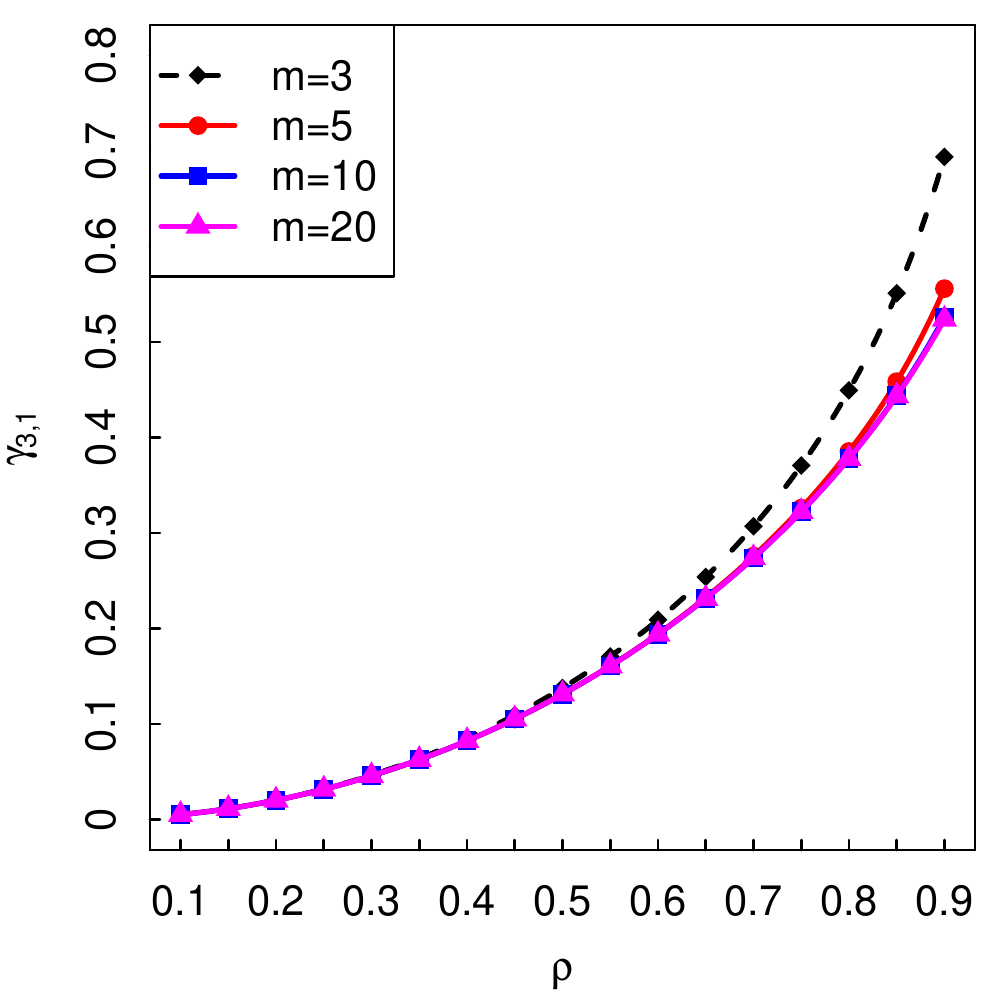}
	}\quad
	\subfigure[Banded]{
		\includegraphics[width=0.3\textwidth]{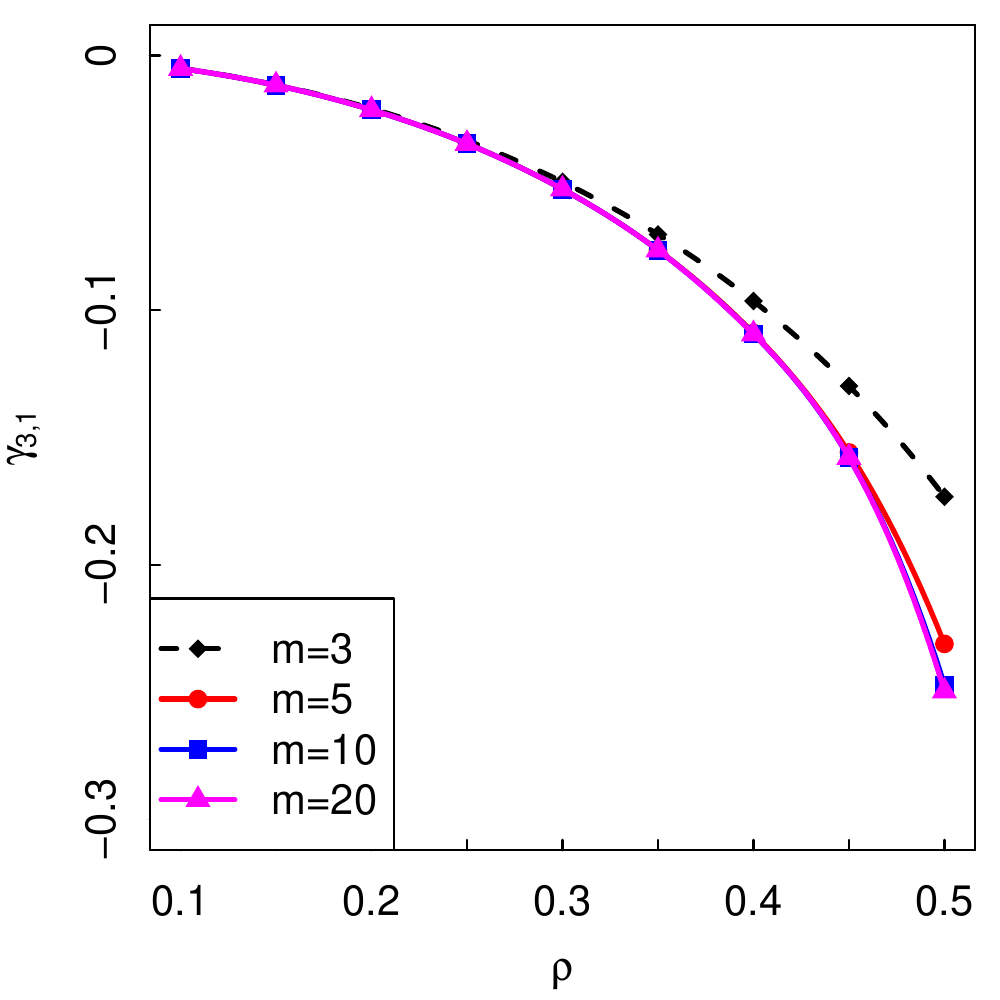}
	}
	\caption{$\{\log(\mathbf{R})\}_{3,1}$ versus $\mathbf{R}_{2,1}$  under two commonly used correlation structures. }
	\label{fig7}
\end{figure}

\subsection{More technical details}

We have the following intermediate result that is useful in our derivations. 

\begin{lemma} \label{lem1}
	Suppose that $\bm{\epsilon}\in \mathbb{R}^d$ and $\bm{\epsilon}\sim \mathcal{N}(\mathbf{0}, \mathbf{I}_d)$. Then for any $d\times d$ matrix $\mathbf{B}$, $\mathbb{E}(\bm{\epsilon}\bm{\epsilon}^{\prime}\mathbf{B}\bm{\epsilon}\bm{\epsilon}^{\prime})=\mathbf{B}+\mathbf{B}^{\prime}+\operatorname{tr}(\mathbf{B}) \mathbf{I}_d$, where $\mathbf{I}_d$ is an identity matrix and $\operatorname{tr}(\cdot)$ is the trace of a matrix.
\end{lemma}
\begin{proof}[Proof]
	Without loss of generality, we assume $\bm{\epsilon}=(\epsilon_1, \ldots, \epsilon_d)^{\prime}$, so that $\epsilon_i, i=1, \ldots, d$ are i.i.d standard normal variables, and we have $\mathbb{E}(\epsilon_i^2)=1 $ and $\mathbb{E}(\epsilon_i^4)=3$. Denote by $\mathbf{B}=(b_{ij})_{i,j=1}^d$. Then it is easy to compute
	$$\bm{\epsilon}^{\prime}\mathbf{B}\bm{\epsilon}=\sum_{m=1}^{d}\sum_{n=1}^{d}b_{mn}\epsilon_m\epsilon_n.$$
	Thus, the $(i, j)$th element of $\mathbb{E}(\bm{\epsilon}\bm{\epsilon}^{\prime}\mathbf{B}\bm{\epsilon}\bm{\epsilon}^{\prime})$ is given by
	\begin{equation}\label{eq:A1}
		\mathbb{E}[\epsilon_i(\sum_{m=1}^{d}\sum_{n=1}^{d}b_{mn}\epsilon_m\epsilon_n)\epsilon_j ]=\left\{\begin{array}{lll}
			2b_{ii}+ \sum_{k=1}^{d}b_{kk}, & \text { if } & i=j \\
			b_{ij}+b_{ji}, & \text { if } & i \neq j
		\end{array}.\right.
	\end{equation}
	Then, it is easy to rewrite (\ref{eq:A1}) in matrix form $\mathbb{E}(\bm{\epsilon}\bm{\epsilon}^{\prime}\mathbf{B}\bm{\epsilon}\bm{\epsilon}^{\prime})=\mathbf{B}+\mathbf{B}^{\prime}+\operatorname{tr}(\mathbf{B}) \mathbf{I}_d$.
\end{proof}	

We then compute the score equations below.
\begin{proof}[Score equations]
	The calculation of $\mathbf{S}_{1}(\bm{\beta};\bm{\alpha},\bm{\lambda} )$ is straightforward and is omitted. 
	
	To calculate $\mathbf{S}_{2}(\bm{\alpha}; \bm{\beta}, \bm{\lambda} )$ we rewrite the log-likelihood (\ref{eq2}) as 
	\begin{equation} \label{eq:A2}
		l(\bm{\omega})=-\frac{1}{2}\sum_{i=1}^{n} \left( \log \left|\mathbf{D}_{i}^2\right|+\log \left|\mathbf{R}_{i} \right|+\operatorname{tr}\left(\mathbf{R}_{i}^{-1}\widehat{\mathbf{R}}_i \right)   \right).
	\end{equation}
	Since $\bm{\rho}_i=\operatorname{vecl}(\bm{R}_{i})=(\rho_{ijk}), (1\le k<j \le m_i)$ and $\frac{\partial \log \left|\mathbf{R}_{i} \right| }{\rho_{ijk}}=2(\mathbf{R}_i^{-1})_{jk}$, by the chain rule we have
	$$\frac{\partial \log \left|\mathbf{R}_{i} \right| }{\partial \bm{\gamma}_i}=\sum_{k<j} \left(\frac{\partial \log \left|\mathbf{R}_{i} \right| }{\rho_{ijk}}\right) \frac{\partial \rho_{ijk}}{\partial \bm{\gamma}_i}=2\sum_{k<j}(\mathbf{R}_i^{-1})_{jk} \frac{\partial \rho_{ijk}}{\partial \bm{\gamma}_i}=2\left( \frac{\partial \bm{\rho}_i}{\partial \bm{\gamma}_i}\right)^{\prime} \operatorname{vecl}(\mathbf{R}_i^{-1}),  $$
	and 
	\begin{align} \label{eq:A3}
		\frac{\partial\operatorname{tr}\left(\mathbf{R}_{i}^{-1}\widehat{\mathbf{R}}_i \right)}{\partial \bm{\gamma}_i}=&\sum_{j,k=1}^{m_i} \left(\frac{\partial \operatorname{tr}\left(\mathbf{R}_{i}^{-1}\widehat{\mathbf{R}}_i \right)}{\partial \mathbf{R}_{i}^{-1}} \right)_{jk}\frac{\partial (\mathbf{R}_{i}^{-1})_{jk}}{\partial \bm{\gamma}_i}=\sum_{j,k=1}^{m_i}(\widehat{\mathbf{R}}_i)_{jk} \frac{\partial (\mathbf{R}_{i}^{-1})_{jk}}{\partial \bm{\gamma}_i} \notag \\
		=& \left( \frac{\partial \bm{\rho}_i}{\partial \bm{\gamma}_i}\right)^{\prime} \sum_{j,k=1}^{m_i}(\widehat{\mathbf{R}}_i)_{jk} \frac{\partial (\mathbf{R}_{i}^{-1})_{jk}}{\partial \bm{\rho}_i}.
	\end{align}
	Because
	$\frac{\partial (\mathbf{R}_{i}^{-1})_{jk}}{\partial \mathbf{R}_i}=\frac{\partial \mathbf{R}_i^{-1}}{\partial (\mathbf{R}_i)_{jk}}=-\mathbf{R}_i^{-1}\mathbf{E}_{jk}\mathbf{R}_i^{-1} $, where $\mathbf{E}_{jk}$ is the selection matrix with 1 in its $(j,k)$th element and zero otherwise, by the definition of  $\bm{\rho}_i=\operatorname{vecl}(\bm{R}_{i})$ we have
	\begin{equation} \label{eq:A4}
		\frac{\partial (\mathbf{R}_{i}^{-1})_{jk}}{\partial \bm{\rho}_i}=-2 \operatorname{vecl}(\mathbf{R}_i^{-1}\mathbf{E}_{jk}\mathbf{R}_i^{-1} ).
	\end{equation}
	Here, the 2 in the right-hand side of  (\ref{eq:A4}) is due to a change in an element of $\bm{\rho}_i$ affecting two symmetric entries in the matrix $\mathbf{R}_i$. Substituting equation (\ref{eq:A4}) in equation (\ref{eq:A3}), we obtain
	\begin{align}
		\frac{\partial\operatorname{tr}\left(\mathbf{R}_{i}^{-1}\widehat{\mathbf{R}}_i \right)}{\partial \bm{\gamma}_i}=&-2\left( \frac{\partial \bm{\rho}_i}{\partial \bm{\gamma}_i}\right)^{\prime} \sum_{j,k=1}^{m_i}(\widehat{\mathbf{R}}_i)_{jk}\operatorname{vecl}(\mathbf{R}_i^{-1}\mathbf{E}_{jk}\mathbf{R}_i^{-1} ) \notag \\
		=&-2\left( \frac{\partial \bm{\rho}_i}{\partial \bm{\gamma}_i}\right)^{\prime} \operatorname{vecl}(\mathbf{R}_i^{-1}\widehat{\mathbf{R}}_i \mathbf{R}_i^{-1} ).\notag
	\end{align}
	Therefore,
	$$
	\frac{\partial l(\bm{\omega}) }{\partial \bm{\alpha}}=-\frac{1}{2}\sum_{i=1}^{n}\mathbf{W}_i^{\prime}\left(\frac{\partial \log \left|\mathbf{R}_{i} \right| }{\partial \bm{\gamma}_i}+\frac{\partial\operatorname{tr}\left(\mathbf{R}_{i}^{-1}\widehat{\mathbf{R}}_i \right)}{\partial \bm{\gamma}_i}  \right) =\sum_{i=1}^{n}\mathbf{W}_i^{\prime} \left( \frac{\partial \bm{\rho}_i}{\partial \bm{\gamma}_i}\right)^{\prime} \operatorname{vecl}(\mathbf{R}_i^{-1}\widehat{\mathbf{R}}_i \mathbf{R}_i^{-1}-\mathbf{R}_i^{-1}),	
	$$
	and this establishes $\mathbf{S}_{2}(\bm{\alpha};\bm{\beta},\bm{\lambda} )$. 
	
	To calculate $\mathbf{S}_{2}(\bm{\lambda};\bm{\beta},\bm{\alpha})$, it is easy to see that
	$\frac{\partial \log \left|\mathbf{D}_{i}^2\right|}{\partial \bm{\lambda}}= \mathbf{Z}_i^{\prime} \mathbf{1}_{m_i},$
	where $\mathbf{1}_{m_i}$ is $m_i\times 1$ vector with elements 1. Since the parameter $\bm{\lambda}$ is only on the diagonal elements of $\mathbf{D}_i$ for $i=1,\ldots,n$, by the chain rule, we have
	\begin{align}
		\frac{\partial \operatorname{tr}\left(\mathbf{R}_i^{-1} \mathbf{D}_i^{-1}\bm{\nu}_i \bm{\nu}_i^{\prime}\mathbf{D}_i^{-1}\right)}{\partial \bm{\lambda}}=&\sum_{j=1}^{m_i} \left(\frac{\partial \operatorname{tr}\left(\mathbf{R}_i^{-1} \mathbf{D}_i^{-1}\bm{\nu}_i \bm{\nu}_i^{\prime}\mathbf{D}_i^{-1}\right)}{\partial \mathbf{D}_i^{-1}}\right)_{jj} \frac{\partial \sigma_{ij}^{-1}}{\partial \bm{\lambda}} \notag \\
		=&-\sum_{j=1}^{m_i}\left(\mathbf{R}_i^{-1} \mathbf{D}_i^{-1}\bm{\nu}_i \bm{\nu}_i^{\prime}\mathbf{D}_i^{-1}\right)_{jj} \frac{2\partial \sigma_{ij}}{\sigma_{ij}\partial \bm{\lambda}} \notag\\
		=&-\sum_{j=1}^{m_i}\left(\mathbf{R}_i^{-1} \mathbf{D}_i^{-1}\bm{\nu}_i \bm{\nu}_i^{\prime}\mathbf{D}_i^{-1}\right)_{jj} \frac{\partial \log(\sigma_{ij}^2)}{\partial \bm{\lambda}} \notag \\
		=&-\mathbf{Z}_i^{\prime} \operatorname{diag}\left(\mathbf{R}_i^{-1} \mathbf{D}_i^{-1}\bm{\nu}_i \bm{\nu}_i^{\prime}\mathbf{D}_i^{-1}\right).  \notag
	\end{align}
	Thus
	$$\frac{\partial l(\bm{\omega}) }{\partial \bm{\lambda}}=\frac{1}{2}\sum_{i=1}^{n}\mathbf{Z}_i^{\prime}\left( \operatorname{diag}(\mathbf{R}_i^{-1} \mathbf{D}_i^{-1}\bm{\nu}_i \bm{\nu}_i^{\prime}\mathbf{D}_i^{-1})-\mathbf{1}_{m_i} \right), $$
	and we have completed the derivation.
\end{proof}

\begin{proof}[The proof of  Lemma \ref{pro1} and the calculation of the Fisher information matrix]

	The calculation of $\mathbf{I}_{11}(\bm{\omega})$ is trivial.  
	Since $\bm{\Sigma}_i, i=1, \ldots, n$ only depend on $\bm{\alpha}$ and $\bm{\lambda}$, it is easy to see that $\mathbf{I}_{12}(\bm{\omega})=\mathbf{0} $ and $\mathbf{I}_{13}(\bm{\omega})=\mathbf{0}$. 
	
	Recall that $\bm{\eta}_i=\operatorname{vecl} (\mathbf{R}_i^{-1}\widehat{\mathbf{R}}_i\mathbf{R}_i^{-1}-\mathbf{R}_i^{-1} )=(\eta_{ijk}), 1\le k<j\le m_i$. For $\mathbf{I}_{22}(\bm{\omega})$, the key is to compute the $\frac{m_i(m_i-1)}{2}\times \frac{m_i(m_i-1)}{2}$ matrix  $\mathbf{J}_i=\mathbb{E}(\bm{\eta}_i \bm{\eta}_i^\prime) $. Since $\mathbb{E}((\mathbf{R}_i^{-1}\widehat{\mathbf{R}}_i\mathbf{R}_i^{-1})_{jk} )=(\mathbb{E}(\mathbf{R}_i^{-1}\widehat{\mathbf{R}}_i\mathbf{R}_i^{-1}))_{jk}=(\mathbf{R}_i^{-1})_{jk} $, the $(\frac{(2n-k)(k-1)}{2}+j-k, \frac{(2n-s)(s-1)}{2}+l-s)$th element of $\mathbf{J}_i$ for $1\le k<j\le m_i, 1\le s<l\le m_i $ is given by
	$$\mathbb{E}(\eta_{ijk}\eta_{ils})=\mathbb{E}\left[(\mathbf{R}_i^{-1}\widehat{\mathbf{R}}_i\mathbf{R}_i^{-1})_{jk} (\mathbf{R}_i^{-1}\widehat{\mathbf{R}}_i\mathbf{R}_i^{-1})_{ls} \right]-a_{ijk}a_{ils}, $$
	where $a_{ijk}$ is the $(j,k)$th element of $\mathbf{R}_i^{-1}$. Denote by $\bm{\epsilon}_i=\mathbf{R}_i^{-\frac{1}{2}}\mathbf{D}_i^{-1}\bm{\nu}_i $ so that $\bm{\epsilon}_i\sim \mathcal{N}(\mathbf{0}, \mathbf{I}_{m_i}) $	and  $\mathbf{R}_i^{-1}\widehat{\mathbf{R}}_i\mathbf{R}_i^{-1}=\mathbf{R}_i^{-\frac{1}{2}}\bm{\epsilon}_i\bm{\epsilon}_i^{\prime}\mathbf{R}_i^{-\frac{1}{2}} $. Let $\mathbf{T}_{ij} $ be the $j$th column of $\mathbf{R}_i^{-\frac{1}{2}}$. We have $\mathbf{T}_{ij}^{\prime}\mathbf{T}_{ik}=a_{ijk}$, and $(\mathbf{R}_i^{-1}\widehat{\mathbf{R}}_i\mathbf{R}_i^{-1})_{jk}= \mathbf{T}_{ij}^{\prime}\bm{\epsilon}_i\bm{\epsilon}_i^{\prime} \mathbf{T}_{ik}$. Thus $$\mathbb{E}\left[(\mathbf{R}_i^{-1}\widehat{\mathbf{R}}_i\mathbf{R}_i^{-1})_{jk} (\mathbf{R}_i^{-1}\widehat{\mathbf{R}}_i\mathbf{R}_i^{-1})_{ls} \right]=\mathbf{T}_{ij}^{\prime}\mathbb{E}\left( \bm{\epsilon}_i\bm{\epsilon}_i^{\prime} \mathbf{T}_{ik}\mathbf{T}_{il}^{\prime}\bm{\epsilon}_i\bm{\epsilon}_i^{\prime} \right)\mathbf{T}_{is}.$$
	From Lemma \ref{lem1},  
	$$\mathbb{E}(\eta_{ijk}\eta_{ils})=\mathbf{T}_{ij}^{\prime} \mathbf{T}_{ik}\mathbf{T}_{il}^{\prime}\mathbf{T}_{is}+\mathbf{T}_{ij}^{\prime} \mathbf{T}_{il}\mathbf{T}_{ik}^{\prime}\mathbf{T}_{is}+\mathbf{T}_{ij}^{\prime} \operatorname{tr}(\mathbf{T}_{ik}\mathbf{T}_{il}^{\prime})\mathbf{T}_{is}-a_{ijk}a_{ils}=a_{ijl}a_{iks}+a_{ijs}a_{ikl},  $$	
	and this proves the first part of Lemma \ref{pro1}.
	
	Similarly, recall that $\bm{\phi}_i=\mathbf{h}_{i}-\mathbf{1}_{m_{i}}=(\phi_{il}), 1\le l\le m_i $, where $\mathbf{h}_{i}=\operatorname{diag}(\mathbf{R}_i^{-1} \mathbf{D}_i^{-1}\bm{\nu}_i \bm{\nu}_i^{\prime}\mathbf{D}_i^{-1})$. Then for  $\mathbf{I}_{23}(\bm{\omega})$, the key is to compute the $\frac{m_i(m_i-1)}{2}\times m_i$ matrix  $\mathbf{H}_i=\mathbb{E}(\bm{\eta}_i \bm{\phi}_i^\prime) $. Since $\mathbb{E}(\mathbf{h}_{i} )=\mathbf{1}_{m_i} $, the $(\frac{(2n-k)(k-1)}{2}+j-k, l)$th element of $\mathbf{H}_i$ for $1\le k<j\le m_i, 1\le l\le m_i $ is given by
	$$\mathbb{E}(\eta_{ijk}\phi_{il})=\mathbb{E}\left[(\mathbf{R}_i^{-1}\widehat{\mathbf{R}}_i\mathbf{R}_i^{-1})_{jk} (\mathbf{R}_i^{-1}\widehat{\mathbf{R}}_i)_{ll}\right]-a_{ijk}. $$
	Note $\mathbf{R}_i^{-1} \mathbf{D}_i^{-1}\bm{\nu}_i \bm{\nu}_i^{\prime}\mathbf{D}_i^{-1}= \mathbf{R}_i^{-\frac{1}{2}}\bm{\epsilon}_i\bm{\epsilon}_i^{\prime}\mathbf{R}_i^{\frac{1}{2}}$,  and denote by $\mathbf{P}_{ij}$  the $j$th column of $\mathbf{R}_i^{\frac{1}{2}}$. We have $\mathbf{P}_{ij}^{\prime}\mathbf{P}_{ik}=\rho_{ijk}$, $\mathbf{T}_{ij}^{\prime}\mathbf{P}_{ik}=\delta_{jk}$, and $(\mathbf{R}_i^{-1}\widehat{\mathbf{R}}_i)_{ll}=\mathbf{T}_{il}^{\prime}\bm{\epsilon}_i\bm{\epsilon}_i^{\prime} \mathbf{P}_{il} $. Thus
	$$\mathbb{E}\left[(\mathbf{R}_i^{-1}\widehat{\mathbf{R}}_i\mathbf{R}_i^{-1})_{jk} (\mathbf{R}_i^{-1}\widehat{\mathbf{R}}_i)_{ll}\right]=\mathbf{T}_{ij}^{\prime}\mathbb{E}\left( \bm{\epsilon}_i\bm{\epsilon}_i^{\prime} \mathbf{T}_{ik}\mathbf{T}_{il}^{\prime}\bm{\epsilon}_i\bm{\epsilon}_i^{\prime} \right)\mathbf{P}_{il}. $$
	By using Lemma \ref{lem1} we have
	$$\mathbb{E}(\eta_{ijk}\phi_{il})=\mathbf{T}_{ij}^{\prime} \mathbf{T}_{ik}\mathbf{T}_{il}^{\prime}\mathbf{P}_{il}+\mathbf{T}_{ij}^{\prime} \mathbf{T}_{il}\mathbf{T}_{ik}^{\prime}\mathbf{P}_{il}+\mathbf{T}_{ij}^{\prime} \operatorname{tr}(\mathbf{T}_{ik}\mathbf{T}_{il}^{\prime})\mathbf{P}_{il}-a_{ijk}=a_{ijl}\delta_{jl}+a_{ikl}\delta_{kl}. $$		
	This completes the proof of Lemma \ref{pro1}.
	
	For $\mathbf{I}_{33}(\bm{\omega})$, the key is to compute the $m_i \times m_i$ matrix $ \mathbb{E}(\bm{\phi}_i \bm{\phi}_i^\prime)$. According to the proof above, the $(j,k)$th element of $ \mathbb{E}(\bm{\phi}_i \bm{\phi}_i^\prime)$ for $j,k=1,\ldots,m_i$ can be calculated as
	$$\mathbb{E}(\phi_{ij}\phi_{ik})=\mathbb{E}\left[(\mathbf{R}_i^{-1}\widehat{\mathbf{R}}_i)_{jj} (\mathbf{R}_i^{-1}\widehat{\mathbf{R}}_i)_{kk}\right]-1, $$
	and
	$$\mathbb{E}\left[(\mathbf{R}_i^{-1}\widehat{\mathbf{R}}_i)_{jj} (\mathbf{R}_i^{-1}\widehat{\mathbf{R}}_i)_{kk}\right] =\mathbf{T}_{ij}^{\prime}\mathbb{E}\left( \bm{\epsilon}_i\bm{\epsilon}_i^{\prime} \mathbf{P}_{ij}\mathbf{T}_{ik}^{\prime}\bm{\epsilon}_i\bm{\epsilon}_i^{\prime} \right)\mathbf{P}_{ik}.$$ 	
	Similarly, using Lemma \ref{lem1} we have
	$$\mathbb{E}(\phi_{ij}\phi_{ik})=\mathbf{T}_{ij}^{\prime} \mathbf{P}_{ij}\mathbf{T}_{ik}^{\prime}\mathbf{P}_{ik}+\mathbf{T}_{ij}^{\prime} \mathbf{T}_{ik}\mathbf{P}_{ij}^{\prime}\mathbf{P}_{ik}+\mathbf{T}_{ij}^{\prime} \operatorname{tr}(\mathbf{P}_{ij}\mathbf{T}_{ik}^{\prime})\mathbf{P}_{ik}-1=a_{ijk}\rho_{ijk}+\delta_{jk}. $$
	Thus, we have $\mathbb{E}(\bm{\phi}_i \bm{\phi}_i^\prime)=\mathbf{R}_i^{-1}\circ \mathbf{R}_i+\mathbf{I}_{m_i}$, where $\circ$ denotes the Hadamard product.	
\end{proof}

\begin{proof}[Proof of Theorem 1]
	The proof follows standard steps  under the regularity conditions; we omit the details here.	
\end{proof}

\subsection{Malaria immune response data in Benin} \label{ex:long}
As an application to analyze longitudinal data,
we  apply our framework to a malaria immune response data set studied by \citet{Adjakossa:2016}, where  the primary goal of the study is analyzing the malaria incidence in children in their early months via  examining various  antigens' measurements related to the  immune reactions against the disease.
In this study,  multiple measurements are available for the same subject, and the between-measurements correlations is known important for the analysis.  An interesting aspect of this dataset is that the measurements of each child were sequentially ordered but the exact time of a measurement was not recorded.  Additionally,  two time-dependent variables -- mosquito exposures (pred\_trim) and  nutrition scores (nutri\_trim) -- are available; and assessing the effect of their inclusion in the model is an aim of our study.  In addition, we assess the effect on model fitting due to missing or erroneous ordering of the data, so as to show the robustness of our method. For comparison purposes, we include existing methods  that are designed  for longitudinal data. 

The response  variable is the level of the protein coded as IgG1\_A1 in the children, assessed at 3, 6, 9, 12, 15, and 18 months. This variable was obtained by using two recombinant \textit{P. falciparum} antigens to perform antibody quantification by Enzyme-Linked ImmunoSorbent Assay  standard methods developed for evaluating malaria vaccines by the African Malaria Network Trust (AMANET [www.amanet148trust.org]).
Due to missing data, the number of measurements for each individual varies from 2 to 5. In total,
this data set contains 316 individuals with 1292 measurements.  Together, there are seven  covariates as described in Table \ref{tb:4}. 
\begin{table}[ht]
	\begin{center}
		\caption{Covariates in the malaria immune response data}
		\label{tb:4} 
		\resizebox{\textwidth}{18mm}{
			\begin{tabular}{l|l}
				
				\toprule[1.2pt]
				Covariate & Description (range) \\
				\specialrule{0.05em}{0pt}{0pt}
				CO.IgG1\_A1 &measured concentration of IgG1\_A1 in the umbilical cord blood [-4.59, 8.21]\\
				M3.IgG1\_A1 &predicted concentration of IgG1\_A1 in the child’s peripheral blood at 3 months [-5.83, 4.13] \\
				ap &placental apposition (0=apposition, 1=non apposition)\\
				hb&hemoglobin level [5.7, 17.1]\\
				inf\_trim &number of malaria infections in the previous 3 months [1, 5]\\
				pred\_trim &quarterly average number of mosquito child is exposed to [0.028, 24.25]\\
				nutri\_trim &quarterly average nutrition scores [0, 1]\\	
				
				\bottomrule[1.2pt]
				
		\end{tabular}}
		
	\end{center}
\end{table}

Since the exact times of the measurements were not recorded, when applying the methods for analyzing longitudinal data, we take $\mathbf{t}_i=(1,\ldots, m_i)^{\prime}$ to be the sequential indices of the observations. 
For exploratory analysis, we select all the individuals with 4 and 5 measurements,  and  calculate the respective $4\times 4$ and $5\times 5$ empirical correlation matrices; these correlation matrices allow us to produce the GZT-correlogram. 
By examining empirical versions of $\bm{\gamma}$ versus the differences in two variables -- time $\mathbf{t}_i$ and  mosquito exposures (pred\_trim) --  respectively in Figure \ref{fig4}, we see that ${\gamma}$ decreases with  the difference  in mosquito exposures,  while the trend in the time lag is not clear. This makes sense since malaria is caused by a parasite which is passed to humans through mosquito bites, so that the mosquito exposures are a pronounced common factor affecting the responses.  Meanwhile, the time-lag does not seem very informative, this might be due to the fact that we do not have the exact times available.  
Additionally,  since mosquito bites are seasonal,  the time-lag alone may offer relatively limited insight in revealing the correlation between measurements.

\begin{figure}[htbp]
	\centering
	\subfigure[]{
		\includegraphics[width=6cm]{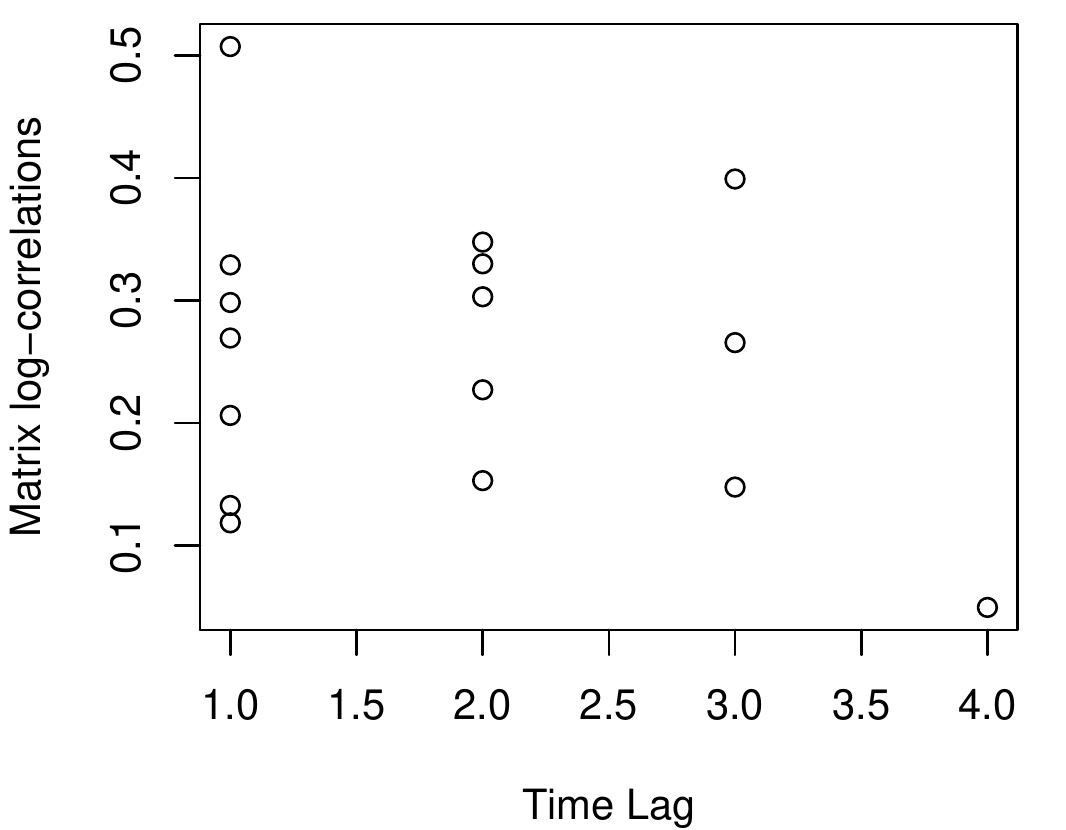}
	}
	\subfigure[]{
		\includegraphics[width=6cm]{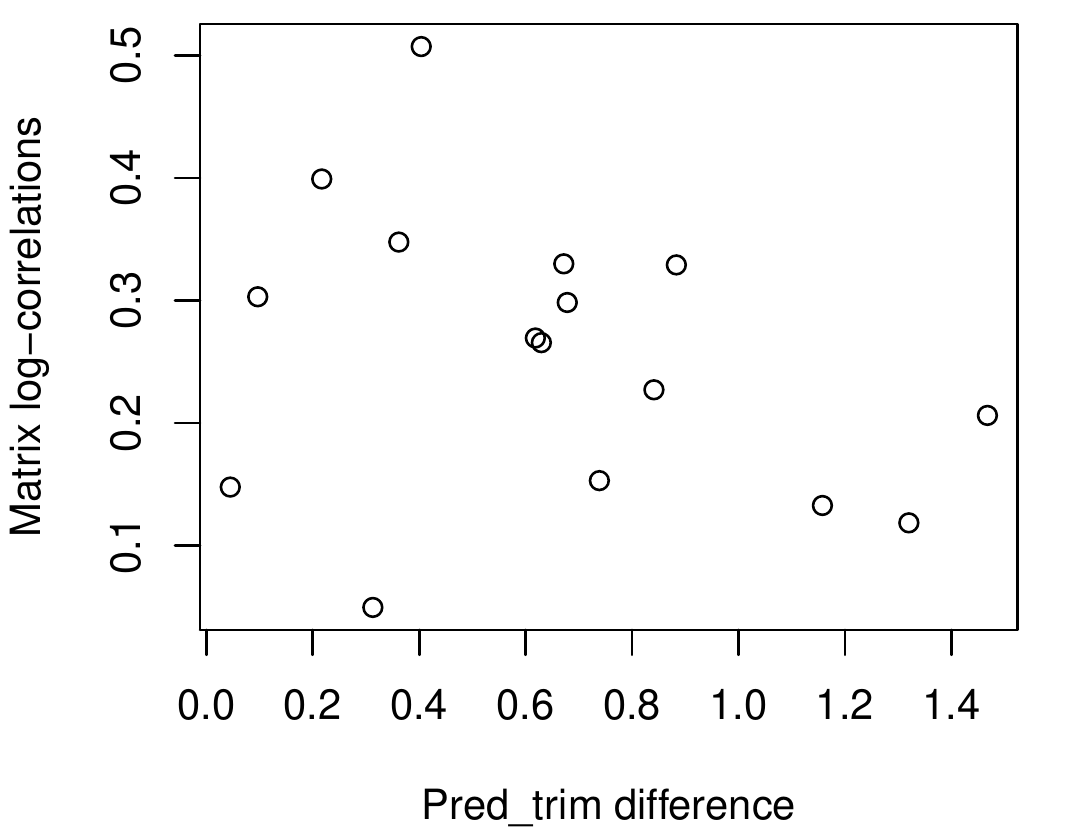}
	}
	
	\caption{ Malaria immune response data. (a): the empirical matrix log-correlations against time lag; (b) the empirical matrix log-correlations against pred\_trim difference.}
	\label{fig4}
\end{figure}
For comparison, we apply 
the approaches in \citet{Pourah:1999} and \citet{Zhang:2015} that require ordering of the observations.  
Furthermore, it is interesting to compare the performance between approaches when the ordering is not maintained.  Having these in mind,   we conduct the following analysis.  We begin with analyzing the data by using their correct ordering, 
then the robustness of different methods is examined by applying them on the data  with a random permutation that breaks down the ordering of the data from the same subject. 
In particular,
we consider  the following three specifications, taking the same linear mean model including all the covariates:

\begin{itemize}
	\item ``time model'':  This specification favors approaches requiring ordering. We use a polynomial of  $t_{ij}$ for the log-variance model and a polynomial of the time lag for the correlation model for our approach and that in \citet{Zhang:2015}, and use a polynomial of  $t_{ij}$ for the log-innovation model and a polynomial of the time lag for the autoregressive model for the modified Cholesky decomposition in \citet{Pourah:1999}. The optimal models all turn out to have $q=1$ and $d=1$ when using BIC as the criterion to choose an optimal model;
	\item ``pred\_trim model'':  This specification maintains the original ordering of the data. We use pred\_trim for the log-variance model  and the difference between the $j$th and $k$th  pred\_trim of individual $i$ for the correlation model for our approach and that in \citet{Zhang:2015}. For the modified Cholesky decomposition in \citet{Pourah:1999}, 
	we use pred\_trim for its log innovation model and  the pred\_trim difference for its autoregressive model;
	\item ``model with permuted data'':
	This specification  randomly permutes the ordering, mimicing  un-ordered clustered data. 
	The model  specification is the same as the previous setting,  but  the measurements of each child are randomly permuted. All the approaches under comparison will use the permuted order as the true order to construct their models. 
	This random permutation is conducted 50 times and the results are averaged. 
\end{itemize}

The results of the comparison are presented in Table \ref{tb:2}. We can see clearly that  the proposed method outperforms the other two approaches,  with or without knowing the ordering of the measurements.   Furthermore,  the methods of \citet{Pourah:1999} and \citet{Zhang:2015} are satisfactory only when the ordering is known; but  they clearly fail to fit the data  if the correct ordering is not maintained. This observation raises a cautionary flag when applying those ordering-sensitive methods in cases when such information may be erroneous or unknown.

We plot selective fitted variances and correlations in Figure \ref{fig2}, and elaborate some  findings that could be useful in further study designs and data analysis of this kind.    
First, consistent with the observations from Figure \ref{fig4}, we observe a decreasing trend in correlations when the differences between pred\_trim becomes smaller.   This is probably due to the fact that the measurements tend to move in the same direction for the subjects sharing similar levels of  mosquito exposures. 
The mosquito exposure is also seen to have an increasing trend for the log-variances as in panel (d) of  Figure \ref{fig4}: higher levels of exposures are associated with higher levels of variations in the measurements.  
Intuitively, when the mean of mosquito exposure increases, one expects that the variance of its impact also increases, due to heterogeneous conditions of the children at their early months
A possible reason is that not all the  families are subject to the same public health conditions, causing their responses to   higher levels of mosquito exposure more varying,  while these variations are  lower when there is less severe mosquito exposure.   We observe that the effect due to mosquito exposure is clearly dominating the pattern in the variance. 
After  the contribution of mosquito exposure is accounted for, a decreasing trend in the variance over time is seen; but the range of the variations  is very small. 
This is understandable as  there is  likely some seasonality in this problem and there is clear data evidence that a more effective variable is available to explain the variations of the response. 
Another interesting aspect of this data analysis is that the exact time the observations were obtained are unknown, so that using the time variable alone could lead to less accurate results. Our approach provides a useful alternative in case there are other variables that carry information about the time, and can incorporate those variables effortlessly as we have demonstrated.

\newcolumntype{H}{>{\setbox0=\hbox\bgroup}c<{\egroup}@{}}
\begin{table}[!htb]
	\begin{center}
		\caption{Malaria immune response data: Comparison of different models}
		\label{tb:2} 
		\begin{tabular}{cHHcccccccc}
			
			\toprule[1.2pt]
			
			&\multirow{2}*{ $(q, d)$}	&\multirow{2}*{  \tabincell{c}{ Number of\\ parameters}}&\multicolumn{2}{c}{Our approach}&& \multicolumn{2}{c}{\citet{Pourah:1999}}   && \multicolumn{2}{c}{\citet{Zhang:2015}}                   \\
			
			\cline{4-5} \cline{7-8} \cline{10-11}
			\specialrule{0em}{3pt}{3pt}
			Model& &  & log-likelihood &BIC &&log-likelihood  &BIC&&log-likelihood  &BIC\\
			\specialrule{0.05em}{3pt}{3pt}
			\multirow{1}*{time}
			& $(1,1)$  & 11 & $-1241.65$ & 8.059 & &$-1244.09$ & 8.093 && $-1244.96$ & 8.098 \\
			\specialrule{0.05em}{3pt}{3pt}
			pred\_trim & $(1,1)$ & 11 & $-1233.04$ & 8.004 & &$-1240.16$ & 8.050& & $-1237.37$ & 8.032\\
			\specialrule{0.05em}{3pt}{3pt}
			permuted & $(1,1)$ & 11 & $-1233.04$ & 8.004 & &$-1707.99$ & 11.01& & $-6361.83$ & 40.47\\
			\bottomrule[1.2pt]
		\end{tabular}
		
	\end{center}
\end{table}

\begin{figure}[htbp]
	\centering
	\subfigure[ ]{
		\includegraphics[width=5cm]{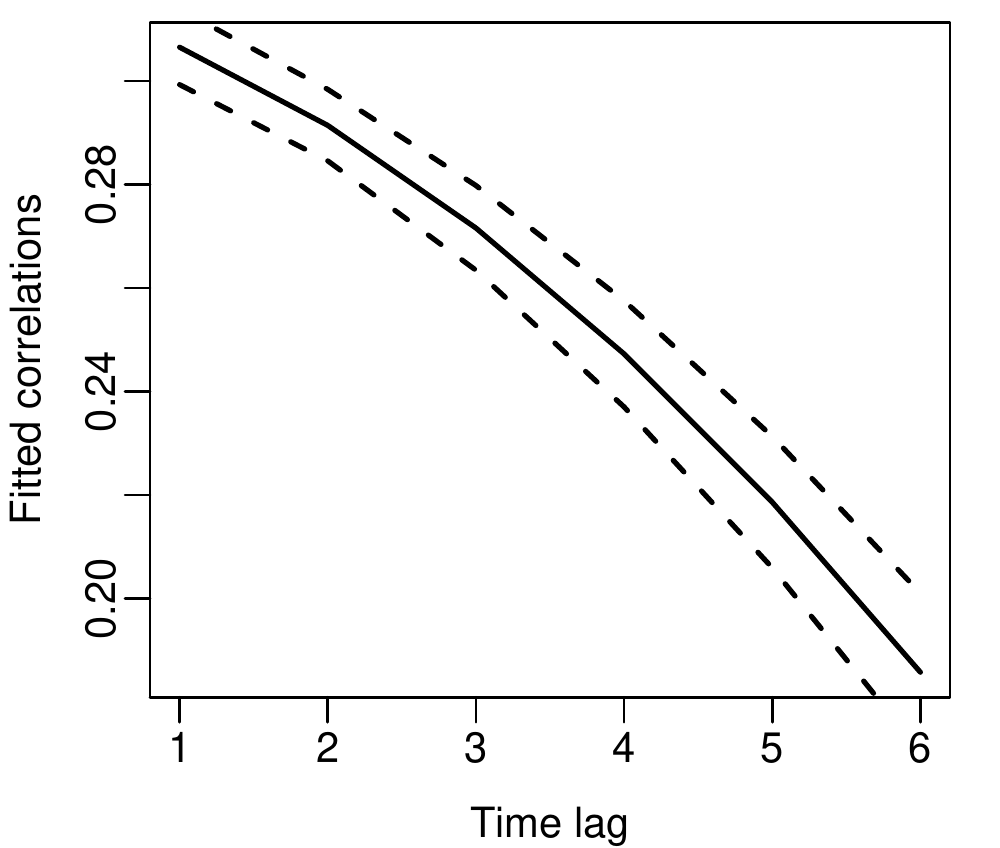}
	}\quad
	\subfigure[]{
		\includegraphics[width=5cm]{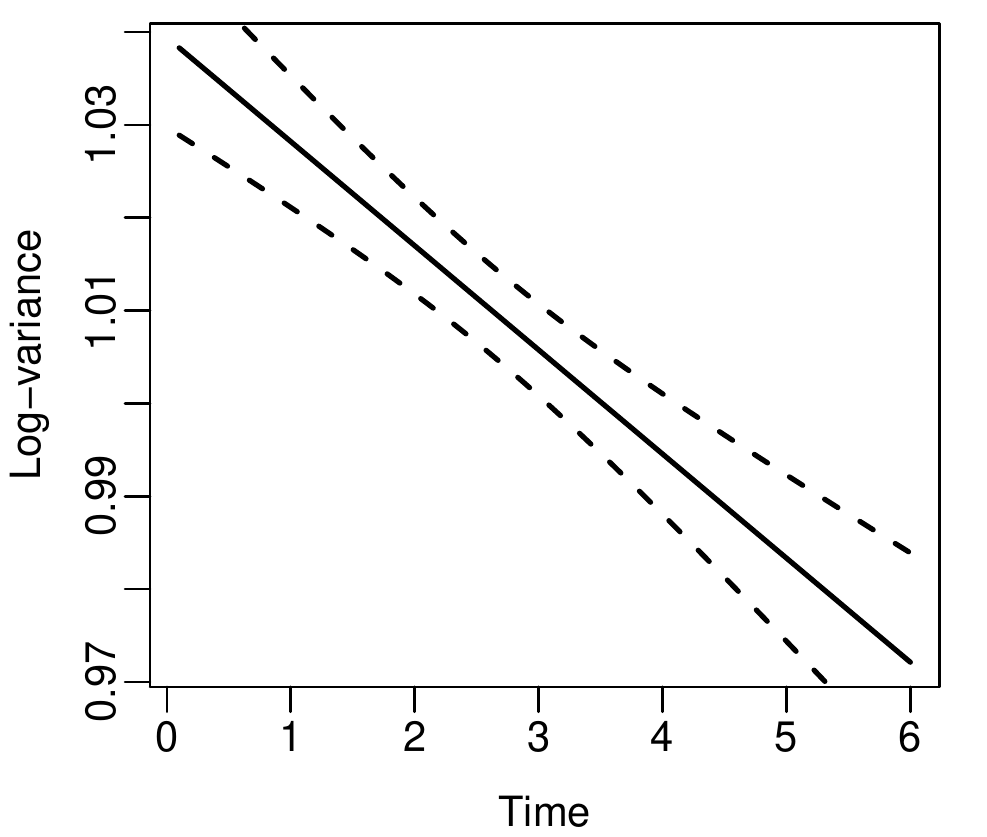}
	}\\
	\subfigure[ ]{
		\includegraphics[width=5cm]{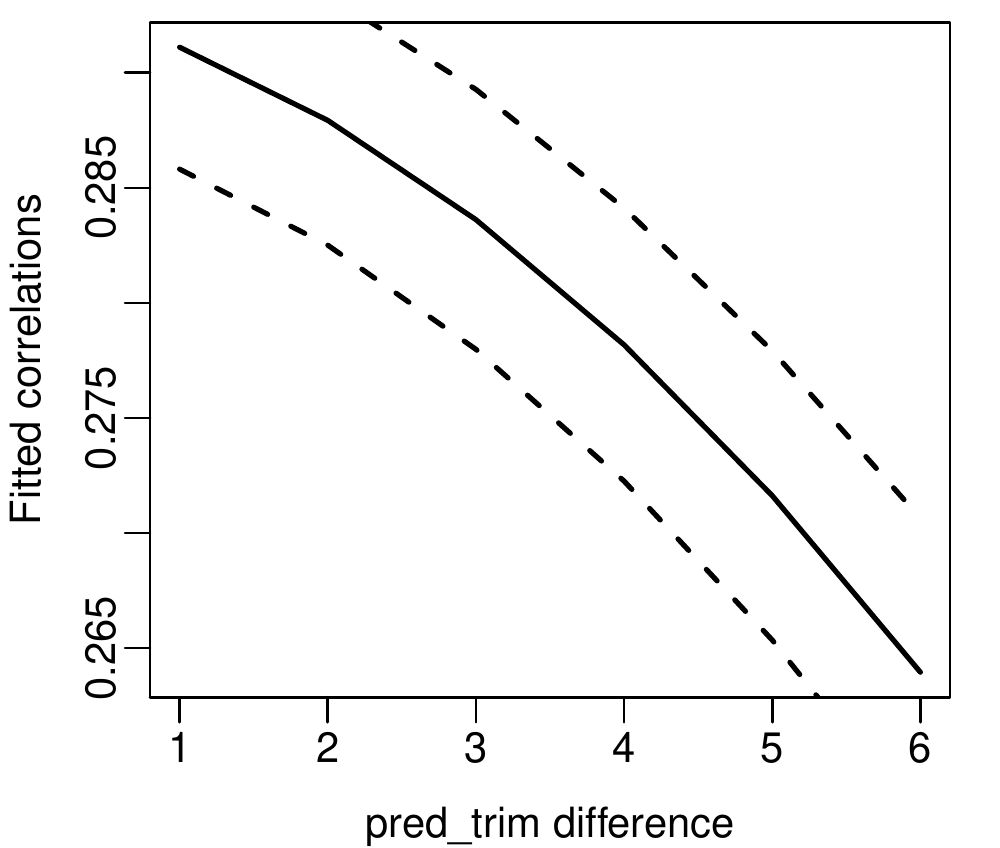}
	}\quad
	\subfigure[ ]{
		\includegraphics[width=5cm]{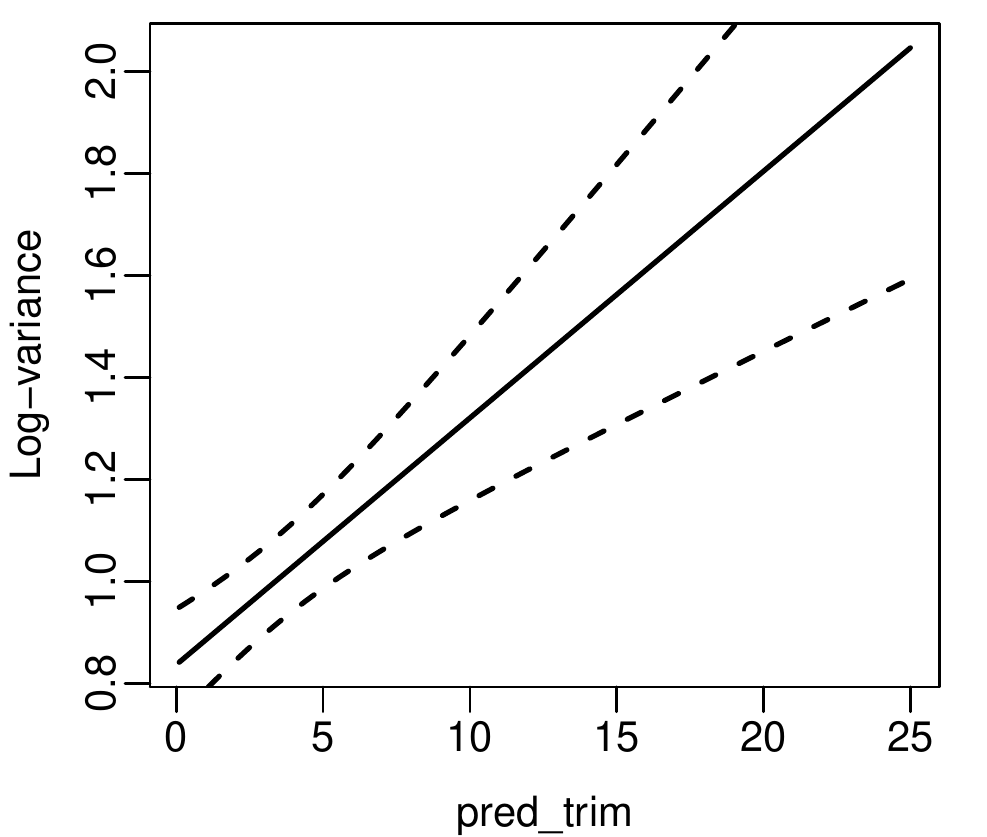}
	}
	\caption{ Malaria immune response data: the dotted lines indicate the asymptotic 95\% confidence intervals: (a) the fitted correlations against time lag; (b) the fitted log-variances against time; (c) the fitted correlations against pred\_trim difference; (d) the fitted log-variances against pred\_trim. }
	\label{fig2}
\end{figure}

\subsection{Simulation}\label{simu}
In this section, we investigate the finite sample performance of the proposed method via three studies. In Study 1, the data are generated from our model to validate our theory with a view to check its robustness.  In Study 2 and Study 3, we generate data inspired by the classroom data and compare our approach to the linear mixed-effects model to demonstrate the usefulness of our approach. 

\textit{Study 1}. The data sets are generated from the following model
$$
\left\{\begin{array}{c}
	y_{i j}=\beta_{0}+x_{i j 1} \beta_{1}+x_{i j 2} \beta_{2}+e_{i j} \quad \left(i=1, \ldots, n ; j=1, \ldots, m_{i}\right) \\
	\gamma_{i j k}=\alpha_{0}+w_{i j k 1} \alpha_{1}+w_{i j k 2} \alpha_{2} \quad \left(i=1, \ldots, n ; 1\le k<j\le m_{i}\right) \\
	\log(\sigma_{i j}^{2})=\lambda_{0}+z_{i j 1} \lambda_{1}+z_{i j 2} \lambda_{2} \quad \left(i=1, \ldots, n ; j=1, \ldots, m_{i}\right)
\end{array} \right.
$$
where $m_{i}-1 \sim$ binomial$(6, 0.8)$ gives rise to different numbers of repeated measurements $m_{i}$ for each subject, the covariate $\mathbf{x}_{i j}=\left(x_{i j 1}, x_{i j 2}\right)^{\prime}$ is generated from a multivariate normal distribution with mean $\mathbf{0}$, marginal variance $1$ and correlation $0.5$, and we set $\mathbf{z}_{i j}=\mathbf{x}_{i j}$. For $\mathbf{w}$, we set $\mathbf{w}_{i j k}=\left(1, u_{i j}-u_{i k},\left(u_{i j}-u_{i k}\right)^{2}\right)^{\prime}$ with $u_{i j}$ generated from the uniform $(0,1)$ distribution.   
The values of the parameters are set to be $\bm{\beta}=(1.0, -0.5, 0.5)^{\prime}$, $\bm{\alpha}=(0.3,-0.2,0.3)^{\prime}$ and $\bm{\lambda}=(-0.5, 0.5, -0.3)^{\prime}$. We generate 1,000 data sets and consider sample sizes $n=50, 100,$ or $200$.

To examine the robustness of the method against mis-specification of the Gaussian assumption,   we further consider generating $\mathbf{e}_i=(e_{i1}, \ldots, e_{im_i})^{\prime}$ from multivariate $t$ distribution with a covariance matrix $\Sigma_i$ and degrees of freedom of 10, 5 and 3, respectively.  Here, the smaller the degree of freedom is, the heavier the tail of the error distribution is.  

The results of the simulations are reported in Table \ref{tb:3}, summarized by the following measures: the  mean absolute distance (MAD) from the estimations to the truth for each of the model parameters, the mean prediction $\parallel \widehat{\mu} \parallel=\frac{1}{n}\sum_{a=1}^{n}\parallel  \mathbf{x}_{a}^{\prime}(\widehat{\bm{\beta}}-\bm{\beta}_0) \parallel $  and the covariance prediction $  \parallel \widehat{\Sigma} \parallel=\frac{1}{n}\sum_{a=1}^{n}\parallel \widehat{\bm{\Sigma}}_a-\bm{\Sigma}_{0a} \parallel$.

Overall, the performance of our method in this study is satisfactory, showing substantial improvement when the sample size increases. When the error distribution is $t$ instead of normal, the likelihood function is mis-specified and the performance suffers. This is reflected in the larger values of all the error measures.   The estimation of the correlation parameter is particularly affected, especially when the degree of freedom of the $t$ distribution is small.   Thus, when there is evidence of severe violation of the distribution assumption, caution is needed when using our likelihood based estimation approach.  As an alternative, we may resort to estimation equations based methods that are more robust, which is best left for a future project.

\begin{table}
	\begin{center}
		\caption{Simulation results for Study 1. All the results have been multiplied by 100 and the corresponding standard errors are in subscripts.}
		\label{tb:3} 	
		\resizebox{\textwidth}{36mm}{
			\begin{tabular}{cccccccccccc}
				
				\toprule[1.2pt]
				\multicolumn{12}{c}{Multivariate normal distribution} \\	 
				& $\beta_{0}$ & $\beta_{1}$ & $\beta_{2}$ &$\alpha_{0}$& $\alpha_{1}$&$\alpha_{2}$&$\lambda_{0}$&$\lambda_{1}$&$\lambda_{2}$&$\parallel \widehat{\mu} \parallel$&$\parallel\widehat{\Sigma}\parallel$ \\	
				\specialrule{0.1em}{3pt}{3pt}
				
				n=50& $7.23 (8.79)$ &$3.98 (5.00)$ & $3.78 (4.57)$ &$4.28 (5.16) $&$6.85 (8.73)$ &$15.17 (17.82)$&$12.10 (14.93)$ &$7.23 (8.86)$ & $6.31 (7.98)$ &$7.90 (4.50)$  &$8.99 (4.91)$\\
				n=100& $4.70 (6.12)$ &$2.33 (2.88)$ & $2.16 (2.68)$ &$2.62 (3.34) $&$5.04 (6.00)$ &$9.36 (11.43)$&$7.28 (9.37)$ &$5.05 (6.66)$ & $4.30 (5.31)$ &$5.35 (3.71)$  &$5.93 (3.02)$\\
				n=200&$3.32 (4.10)$ &$1.66 (2.04)$ & $1.51 (1.85)$ &$1.90 (2.27) $&$3.33 (4.02)$ &$6.40 (7.74)$&$4.93 (6.06)$ &$3.25 (4.05)$ & $3.30 (4.03)$ &$3.74 (2.17)$  &$3.99 (1.94)$\\
				\specialrule{0.1em}{3pt}{3pt}
				\multicolumn{12}{c}{Multivariate t distribution (df=10)}\\
				& $\beta_{0}$ & $\beta_{1}$ & $\beta_{2}$ &$\alpha_{0}$& $\alpha_{1}$&$\alpha_{2}$&$\lambda_{0}$&$\lambda_{1}$&$\lambda_{2}$&$\parallel \widehat{\mu} \parallel$&$\parallel\widehat{\Sigma}\parallel$ \\	
				\specialrule{0.1em}{3pt}{3pt}
				n=50&$8.85 (11.01)$ &$4.56 (5.52)$ & $4.00 (5.04)$ &$4.48 (5.45) $&$7.86 (9.56)$ &$17.90 (21.81)$&$23.05 (17.46)$ &$7.21 (9.06)$ & $7.89 (10.08)$ &$9.73 (5.81)$  &$15.37 (9.84)$\\		
				n=100&$5.31 (6.50)$ &$2.44 (2.94)$ & $2.06 (2.64)$ &$2.95 (3.87) $&$4.44 (5.53)$ &$10.29 (12.68)$&$21.48 (12.48)$ &$5.48 (7.14)$ & $5.04 (6.32)$ &$6.05 (3.64)$  &$13.37 (8.21)$\\
				n=200&$3.85 (5.01)$ &$1.91 (2.29)$ & $1.69 (2.16)$ &$1.88 (2.34) $&$3.32 (4.22)$ &$6.64 (8.27)$&$19.60 (7.71)$ &$3.85 (4.70)$ & $4.12 (5.01)$ &$4.42 (2.90)$  &$10.33 (5.00)$\\
				\specialrule{0.1em}{3pt}{3pt}
				\multicolumn{12}{c}{Multivariate t distribution (df=5)}\\
				& $\beta_{0}$ & $\beta_{1}$ & $\beta_{2}$ &$\alpha_{0}$& $\alpha_{1}$&$\alpha_{2}$&$\lambda_{0}$&$\lambda_{1}$&$\lambda_{2}$&$\parallel \widehat{\mu} \parallel$&$\parallel\widehat{\Sigma}\parallel$ \\	
				\specialrule{0.1em}{3pt}{3pt}
				n=50&$10.53 (12.90)$ &$5.77 (7.13)$ & $4.24 (5.53)$ &$4.91 (5.84) $&$8.37 (10.41)$ &$21.20 (25.99)$&$47.23 (23.89)$ &$10.65 (13.30)$ & $8.80 (10.95)$ &$11.33 (6.75)$  &$31.63 (22.32)$\\		
				n=100&$5.76 (7.42)$ &$3.01 (3.81)$ & $2.37 (2.90)$ &$4.02 (5.01) $&$6.24 (7.81)$ &$14.34 (18.58)$&$47.13 (15.34)$ &$7.30 (8.84)$ & $6.31 (7.90)$ &$6.68 (4.29)$  &$29.63 (13.67)$\\
				n=200&$4.30 (5.40)$ &$1.96 (2.50)$ & $1.94 (2.42)$ &$2.98 (3.52) $&$4.72 (6.29)$ &$10.89 (13.67)$&$48.75 (11.69)$ &$5.05 (6.35)$ & $4.82 (5.97)$ &$4.88 (2.94)$  &$28.84 (10.39)$\\
				\specialrule{0.1em}{3pt}{3pt}
				\multicolumn{12}{c}{Multivariate t distribution (df=3)}\\
				& $\beta_{0}$ & $\beta_{1}$ & $\beta_{2}$ &$\alpha_{0}$& $\alpha_{1}$&$\alpha_{2}$&$\lambda_{0}$&$\lambda_{1}$&$\lambda_{2}$&$\parallel \widehat{\mu} \parallel$&$\parallel\widehat{\Sigma}\parallel$ \\	
				\specialrule{0.1em}{3pt}{3pt}
				n=50&$10.66 (13.14)$ &$4.31 (5.73)$ & $4.50 (5.72)$ &$6.34 (8.16) $&$11.50 (14.78)$ &$26.34 (32.72)$&$86.53 (33.95)$ &$12.48 (16.06)$ & $11.77 (14.31)$ &$11.56 (7.25)$  &$71.82 (51.60)$\\			
				n=100&$8.95 (11.61)$ &$3.84 (4.94)$ & $3.28 (4.26)$ &$4.61 (5.76) $&$8.41 (10.42)$ &$14.43 (18.47)$&$92.31 (29.28)$ &$8.66 (11.03)$ & $7.31 (9.54)$ &$9.87 (7.06)$  &$80.20 (49.47)$\\
				n=200&$4.65 (6.01)$ &$2.28 (2.87)$ & $2.37 (2.94)$ &$4.50 (5.58) $&$6.61 (8.61)$ &$13.90 (17.56)$&$98.66 (22.32)$ &$7.79 (9.81)$ & $6.65 (8.39)$ &$5.42 (3.38)$  &$83.86 (41.50)$\\
				\bottomrule[1.2pt]	
		\end{tabular}}	
	\end{center}
\end{table}

\textit{Study 2}. In this study, we compare the estimation accuracy of the proposed approach with the linear mixed-effects model and generalized estimating equations (GEE) under four different scenarios.
We set sample sizes as $n=50$, $100$ or $200$ and generate clustered data from the   model
$y_{ijk}=\beta_{0}+x_{ijk,1}\beta_{1}+x_{ijk,2}\beta_{2}+\varepsilon_{ijk}$, 
where $i=1, \ldots, n; j=1, \ldots, m_i; k=1, \ldots, k_{ij} $. For all the cases considered here, we set $\bm{\beta}_{0}=(1,-0.5,0.5)^{\prime}$ and generate the covariate $\mathbf{x}_{ijk}=\left(x_{ijk, 1}, x_{ijk, 2}\right)^{\prime}$  from a bivariate normal distribution with mean $\mathbf{0}$, marginal variance $1$ and correlation $0.5$.

In case I, we  generate the error from a linear mixed-effects model by taking $\varepsilon_{ijk}=u_{i}+v_{ij}+\epsilon_{ijk}$, where $u_{i}, v_{ij}$ and $\epsilon_{ijk}$ are independent and all follow the standard normal distribution.
To be consistently comparable,  we consider balanced data sets where $m_i=2$ and $k_{ij}=5$ for all $i$ and $j$. In this case, the linear mixed-effects model and our approach estimate the same covariance structure, while the covariance model for  GEE is misspecified.

In case II, we generate data similar to the unbalanced classroom data where $m_i$ is uniform on $\{2, 3, 4\}$ and $k_{ij}-1\sim$ binomial$(4,0.8)$. We  generate  $\varepsilon_{ijk}$ from our model (2) in the paper with $\gamma_{ikk'}=\alpha_{0}+\alpha_{1}w_{ikk'}+ \alpha_{2}|t_{ik}-t_{ik'}|$ with $t_{ik} $ following the uniform distribution on $[0,1]$. Here $w_{ikk'}=1$ if $k$th and $k'$th observations are in the same group and $w_{ikk'}=0$ otherwise. For simplicity, we set $\log(\sigma_{ijk}^2)=\lambda_{0}$, and set $\bm{\alpha}_0=(0.2, 0.3, -0.2)^\prime$ and $\lambda_{0}=1$.  In this case, the covariance model for the linear mixed-effects model and GEE are misspecified.

In case III, we take the same setting on $m_i$ and $k_{ij}$ as in case II and let $\varepsilon_{ijk}=u_{i}+v_{ij}+\epsilon_{ijk}$, where $u_{i}$ and $v_{ij}$ both follow the standard normal distribution as in case I, but with $\epsilon_{ijk}$ having an  autoregressive AR(1) correlation structure in the sense that $\operatorname{corr}(\epsilon_s, \epsilon_{s'})=0.85\rho^{|s-s'|}$ for $\rho=0.6$ corresponding to moderately correlated errors. In this case, all approaches use misspecified models for the covariance.

In case IV, we take the same setting as in case III, but with $\epsilon_{ijk}$ generated from an ARCH(1) process specified as $\sigma_{ij(k+1)}^2=1+0.5\epsilon_{ijk}^2$. In this case, all approaches use misspecified models for the covariance.

We use the following error measurements to compare the performance of the two competing methods
$$\parallel \widehat{\mu} \parallel=\frac{1}{n}\sum_{a=1}^{n}\parallel  \mathbf{x}_{a}^{\prime}(\widehat{\bm{\beta}}-\bm{\beta}_0) \parallel \quad \text{and} \quad  \parallel \widehat{\Sigma} \parallel=\frac{1}{n}\sum_{a=1}^{n}\parallel \widehat{\bm{\Sigma}}_a-\bm{\Sigma}_{0a} \parallel,$$
which are the $\ell_2$-norm of the difference between the estimated mean and the true mean, and the Frobenius norm of the difference between the true covariance matrix and its estimate.
Table \ref{tb:6} presents the average norms with their standard errors over 1000 replicates. In case I where data are generated from the linear mixed-effects model, our approach performs almost the same as the mixed-effects model approach. In case II when the data is generated from our model, it is not surprising that our method performs much better. In case III and case IV when all methods are misspecified, our method still performs well, especially for estimating the covariance matrices.  The simulation results together with our real data examples clearly demonstrate that the proposed approach is more adaptive and flexible for capturing the correlations of correlated data, even when the model is misspecified.

We repeat the experiment with the same setting but with the model errors $\epsilon_{ijk}$ generated from a $t$ distribution with 3 degrees of freedom.  The results are reported in Table \ref{tb:6a}.  We summarize the impact of model misspecification as follows:
\begin{enumerate}
	\item The estimation of the mean model is generally accurate for all methods, as expected; 
	\item Larger errors in estimating the covariance models are present when the respective models are misspecified, e.g., in Case II for the linear mixed-effects model;  
	\item All methods performed worse when the models are misspecified, and as seen in Table \ref{tb:6}, the proposed method is more robust to the misspecification of the covariance models, even in Case IV with an ARCH setting; this finding re-affirms that  the proposed approach is more adaptive and flexible for capturing the correlations of correlated data, which is the primary focus of this work;
	\item When the Gaussian  assumption on $\epsilon_{ijk}$ is violated as in Table 8, the performance of our method in estimating the covariance underperforms in general in comparison to the linear mixed-effects model. This is perhaps related to the fact that our estimator of the matrix log-correlation parameters is on the logarithmic scale In case this violation is of concern, we recommend the use of robust alternatives.
\end{enumerate}

\begin{table}[ht]
	\begin{center} 
		\caption{The mean errors and their standard errors (in parentheses) for the method proposed, the linear mixed-effects model and GEE with an AR(1) working correlation when the errors $\epsilon_{ijk}$ were generated from Gaussian distribution.}
		\label{tb:6} 
		\resizebox{\textwidth}{34mm}{	
			\begin{tabular}{cccccccccc} 
				
				\toprule[1.2pt]
				\multirow{2}{*}{Case}	& \multirow{2}{*}{Size}	& \multicolumn{2}{c}{Proposed model} & & \multicolumn{2}{c}{Mixed-effects model}& & \multicolumn{2}{c}{GEE (ar1)}\\
				
				\cline{3-4} \cline{6-7} \cline{9-10}
				\specialrule{0.0em}{3pt}{3pt}	
				&	& $\parallel \widehat{\mu} \parallel$  &  $ \parallel \widehat{\bm{\Sigma}} \parallel$  & &$ \parallel\widehat{\mu} \parallel$  &  $\parallel\widehat{\bm{\Sigma}} \parallel$& &$ \parallel\widehat{\mu} \parallel$  &  $\parallel\widehat{\bm{\Sigma}} \parallel$ \\	
				
				\specialrule{0.05em}{3pt}{3pt}
				\multirow{3}{*}{Case I}	 &$n=50$& $0.153 (0.096)$ & $ 0.288 (0.172)$& &$0.153 (0.096)$  &$0.287 (0.167)$& &$0.166 (0.103)$  &$0.552 (0.160)$  \\
				&$n=100$& $0.115 (0.075)$ & $ 0.201 (0.107)$& &$0.115 (0.075)$  &$0.199 (0.102)$ & &$0.120 (0.068)$  &$0.503 (0.090)$ \\	
				&$n=200$& $0.084 (0.058)$ & $ 0.132 (0.075)$& &$0.084 (0.058)$  &$0.130 (0.075)$& &$0.087 (0.053)$  &$0.470 (0.044)$  \\
				\specialrule{0.05em}{3pt}{3pt}
				\multirow{3}{*}{Case II}	 &$n=50$& $0.159 (0.105)$ & $ 0.227 (0.149)$& &$0.166 (0.108)$  &$0.993 (0.180)$  & &$0.169 (0.113)$  &$0.670 (0.272)$\\
				&$n=100$& $0.101 (0.066)$ & $ 0.131 (0.082)$& &$0.112 (0.088)$  &$0.674 (0.085)$& &$0.121 (0.087)$  &$0.607 (0.172)$  \\
				&$n=200$& $0.078 (0.043)$ & $ 0.101 (0.053)$& &$0.084 (0.056)$  &$0.560 (0.067)$& &$0.083 (0.052)$  &$0.551 (0.106)$  \\
				\specialrule{0.05em}{3pt}{3pt}
				\multirow{3}{*}{Case III}	 &$n=50$& $0.138 (0.111)$ & $ 0.371 (0.143)$& &$0.141 (0.110)$  &$1.228 (0.370)$& &$0.140 (0.089)$  &$0.512 (0.124)$  \\
				&$n=100$& $0.098 (0.049)$ & $ 0.309 (0.081)$& &$0.103 (0.056)$  &$1.151 (0.303)$& &$0.093 (0.068)$  &$0.468 (0.086)$  \\
				&$n=200$& $0.080 (0.041)$ & $ 0.282 (0.078)$& &$0.083 (0.052)$  &$1.090 (0.206)$& &$0.086 (0.046)$  &$0.436 (0.057)$  \\  
				\specialrule{0.05em}{3pt}{3pt}	 
				\multirow{3}{*}{Case IV}	 &$n=50$& $0.146 (0.079)$ & $ 0.421 (0.224)$& &$0.145 (0.079)$  &$0.681 (0.159)$& &$0.195 (0.088)$  &$0.774 (0.147)$  \\
				&$n=100$& $0.103 (0.065)$ & $ 0.380 (0.122)$& &$0.102 (0.063)$  &$0.691 (0.122)$& &$0.104 (0.066)$  &$0.755 (0.107)$  \\
				&$n=200$& $0.076 (0.046)$ & $ 0.321 (0.070)$& &$0.075 (0.045)$  &$0.687 (0.089)$& &$0.079 (0.046)$  &$0.721 (0.068)$  \\  	
				\bottomrule[1.2pt]
				
		\end{tabular}}
		
	\end{center}
\end{table}

\begin{table}[ht]
	\begin{center} 
		\caption{The mean errors and their standard errors (in parentheses) for the method proposed, the linear mixed-effects model and GEE when the errors $\epsilon_{ijk}$ were generated from $t$ distribution with 3 degrees of freedom.} 
		\label{tb:6a} 
		\resizebox{\textwidth}{34mm}{	
			\begin{tabular}{cccccccccc} 
				
				\toprule[1.2pt]
				\multirow{2}{*}{Case}	& \multirow{2}{*}{Size}	& \multicolumn{2}{c}{Proposed model} & & \multicolumn{2}{c}{Mixed-effects model}& & \multicolumn{2}{c}{GEE (ar1)}\\
				
				\cline{3-4} \cline{6-7} \cline{9-10}
				\specialrule{0.0em}{3pt}{3pt}	
				&	& $\parallel \widehat{\mu} \parallel$  &  $ \parallel \widehat{\bm{\Sigma}} \parallel$  & &$ \parallel\widehat{\mu} \parallel$  &  $\parallel\widehat{\bm{\Sigma}} \parallel$& &$ \parallel\widehat{\mu} \parallel$  &  $\parallel\widehat{\bm{\Sigma}} \parallel$ \\	
				
				\specialrule{0.05em}{3pt}{3pt}
				\multirow{3}{*}{Case I}	 &$n=50$& 0.323 (0.232) & 3.848 (3.479)&& 0.323 (0.232)& 1.018 (0.912) && 0.315 (0.222)& 4.095 (3.547)  \\
				&$n=100$& 0.168 (0.104) &  2.859 (1.076)& &0.168 (0.104)  &0.698 (0.249) & &0.183 (0.100) &3.124  (1.094) \\	
				&$n=200$& 0.107 (0.059) & 3.137 (1.513) & & 0.107 (0.059) & 0.759 (0.413) & & 0.106 (0.061)& 3.344 (1.371) \\
				\specialrule{0.05em}{3pt}{3pt}
				\multirow{3}{*}{Case II}	 &$n=50$& 0.304 (0.266) &  4.231 (2.075)& &0.294 (0.254)  &1.180 (0.228)  & &0.270 (0.224)  &2.874 (1.093)\\
				&$n=100$& 0.191 (0.131) &  3.499 (1.968)& &0.182 (0.118)  &1.101 (0.444)& &0.184 (0.129)  &3.592 (2.159)  \\
				&$n=200$& 0.143 (0.080) & 2.639 (1.292)& &  0.137 (0.054) & 1.055 (0.133)& & 0.171 (0.058) & 3.680 (1.384) \\
				\specialrule{0.05em}{3pt}{3pt}
				\multirow{3}{*}{Case III}	 &$n=50$&  0.215 (0.269) &2.940 (1.214)&&0.213 (0.235)&2.394 (0.891) &&0.230 (0.237)& 3.217 (1.441) \\
				&$n=100$& 0.223 (0.157)& 2.612 (0.712) &&0.232 (0.149) &1.373 (0.658)&& 0.228 (0.161) &3.102 (0.757)  \\
				&$n=200$& 0.108 (0.072)& 2.916 (1.091) &&0.130 (0.063) &1.345 (0.148) &&0.112 (0.053)& 2.993 (0.803) \\  
				\specialrule{0.05em}{3pt}{3pt}	 
				\multirow{3}{*}{Case IV}	 &$n=50$& 0.433 (0.198) &1.838 (0.977) &&0.434 (0.102) &0.673 (0.466) &&0.462 (0.090) &2.809 (1.778) \\
				&$n=100$& 0.140 (0.113) &2.229 (0.840) &&0.143 (0.122)& 0.692 (0.355) &&0.151 (0.125)& 3.178 (1.265)   \\
				&$n=200$& 0.117 (0.052) &2.020 (0.266)&& 0.118 (0.053)& 0.770 (0.357) &&0.112 (0.081) &2.517 (0.207)  \\  	
				\bottomrule[1.2pt]
				
		\end{tabular}}
		
	\end{center}
\end{table}

\textit{Study 3}. This study is for demonstrating the use of our approach for inference when testing the presence of random effects. 

The data are generated in a way similar to case I in Study 2 with sample size $n=50$ where $u_{i}\sim N(0, \sigma_u^2), v_{ij} \sim N(0, \sigma_v^2)$ and $\epsilon_{ijk}\sim N(0,1)$ are independent. That is, the data are generated from the linear mixed-effects model. For each simulation setup, we repeat the experiments 200 times. 
We are interested in testing various hypotheses regarding the magnitude of $\sigma^2_u$ and $\sigma_v^2$. A recommended test for the linear mixed-effects model is the LRT as implemented in ${\verb+lme4+}$ whose performance we now examine. As we have shown in Example 2 in Section 2.2 of the paper, a test regarding $\sigma^2_u$ or $\sigma_v^2$ corresponds to a test regarding $\alpha_0$ and $\alpha_1$ which are defined therein. A trivial application of Corollary 1 in the paper suggests that the LRT via our approach is asymptotically $\chi^2$ distributed with its degrees of freedom determined by the number of constraints under the null hypothesis. For the mixed-effects model approach however, statistical inference is not easy due to the absence of analytical results for the null distributions of parameter estimates \citep{Bates:etal:2014}. When testing the existence of random effects for example, the true parameter value is at the boundary of the support of the variance parameter, making the asymptotic distribution of the LRT not generally tractable.

With the above discussion in mind, we examine various scenarios to test the magnitude and the existence of the random effects. 
First, we test $H_0: \sigma_{u}^2=\sigma_{v}^2=1$, where $\sigma_u^2$ and $\sigma_v^2$ are in the interior of their respective parameter space. When data are generated under this null, we know that the LRT will follow $\chi^2_2$ asymptotically, regardless whether the mixed-effects model or our model is used. This is confirmed in Figure \ref{fig5} (a) and (e) when 
the quantile-quantile (Q-Q) plots of the LRT statistics versus the $\chi^2_2$ distribution is examined. 
In the next three settings, we examine the existence of the random effects by testing the following hypotheses: 

\begin{itemize}
	\item
	$H_0: \sigma_{v}^2=0$ for the linear mixed-effects model, or equivalent  $H_0:\alpha_{1}=0$; 
	
	\item 
	$H_0: \sigma_{u}^2=0$ for the linear mixed-effects model, or equivalent $H_0: \alpha_{0}=0$;   
	\item  $H_0: \sigma_{u}^2=\sigma_{v}^2=0$ for the linear mixed-effects model, or equivalent $H_0: \alpha_{0}=\alpha_{1}=0. $ 
\end{itemize}

The Q-Q plots of the LRT statistics versus the corresponding $\chi^2$ distributions under the three hypotheses above are plotted in the second to the last columns in 
Figure \ref{fig5}.   
It is clear that while the variance components are at the boundary of their parameter spaces, a substantial discrepancy exists between the empirical distribution of the test statistics and the reference distribution for the linear mixed-effects model. In contrast, our approach remains valid. Indeed, \citet{Baayen:2008} commented, in a different context, that though the LRT is often chosen as the test statistic to use for these tests in linear mixed-effects models,  the asymptotic reference distribution of a $\chi^2$ does not apply,  giving rise to mis-calibrated $p$-values for variance parameters if these $p$-values are  computed using the $\chi^2$ reference distribution.

We close with a remark that the  limiting distribution of the LRT for the mixed-effects model in this case is known non-standard; see  \citet{Self:Liang:1987} and \citet{ChenLiang:2010}.  
The analytical procedure is case-by-case for testing various random effects, and its finite sample approximation is known inaccurate
\citep{Crainiceanu:2004}, rendering substantial barriers for testing the random effects in the mixed-effects model in practice.  Our approach, in contrast, provides a simple and justified solution. 

\begin{figure}[htbp]
	\centering
	\subfigure[$\chi_2^2$]{
		\includegraphics[width=3cm]{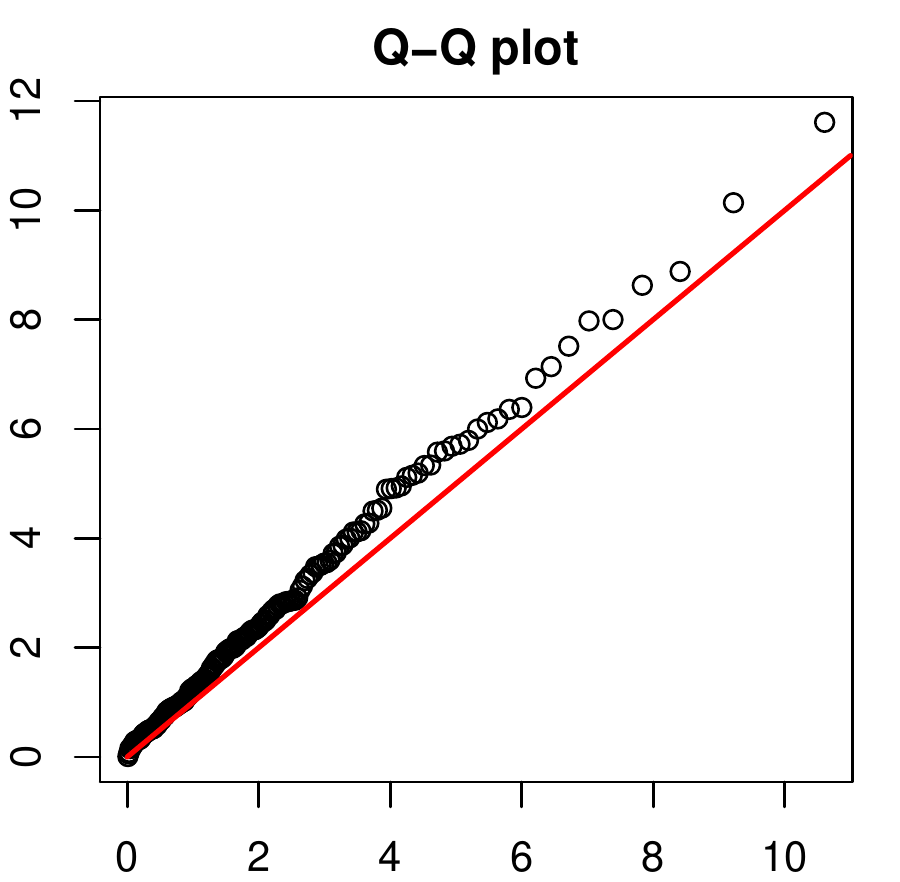}
	}	\quad
	\subfigure[$\chi_1^2$]{
		\includegraphics[width=3cm]{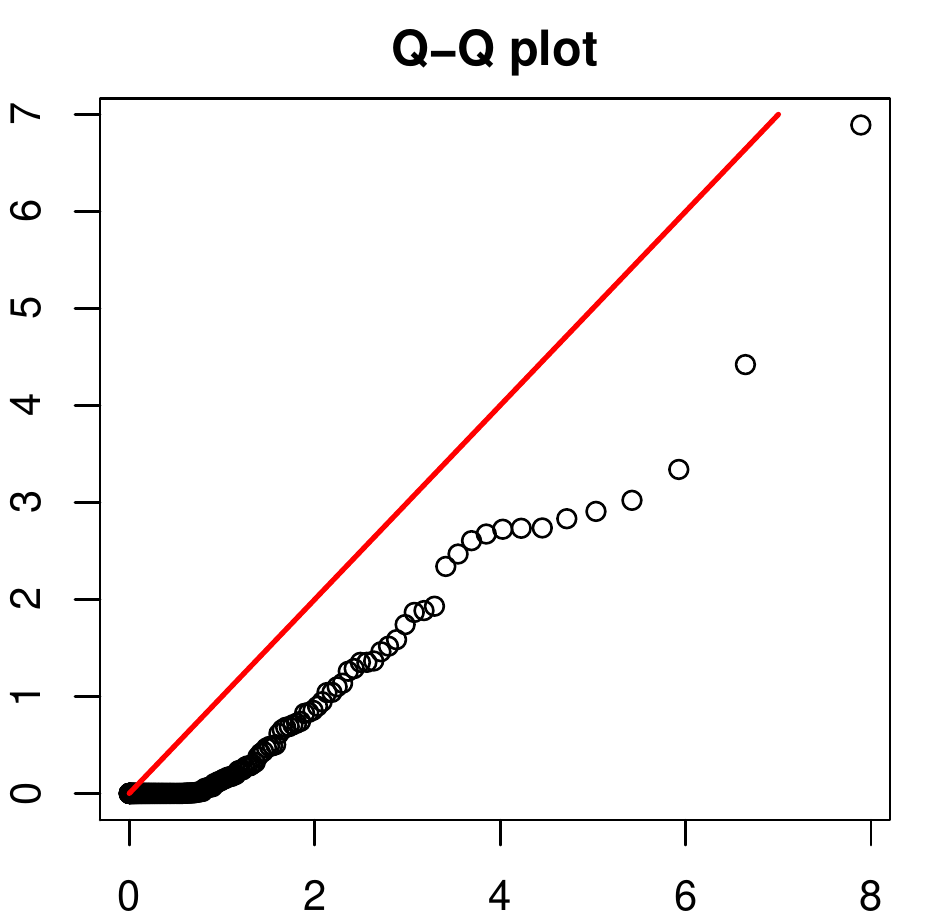}
	}\quad
	\subfigure[$\chi_1^2$]{
		\includegraphics[width=3cm]{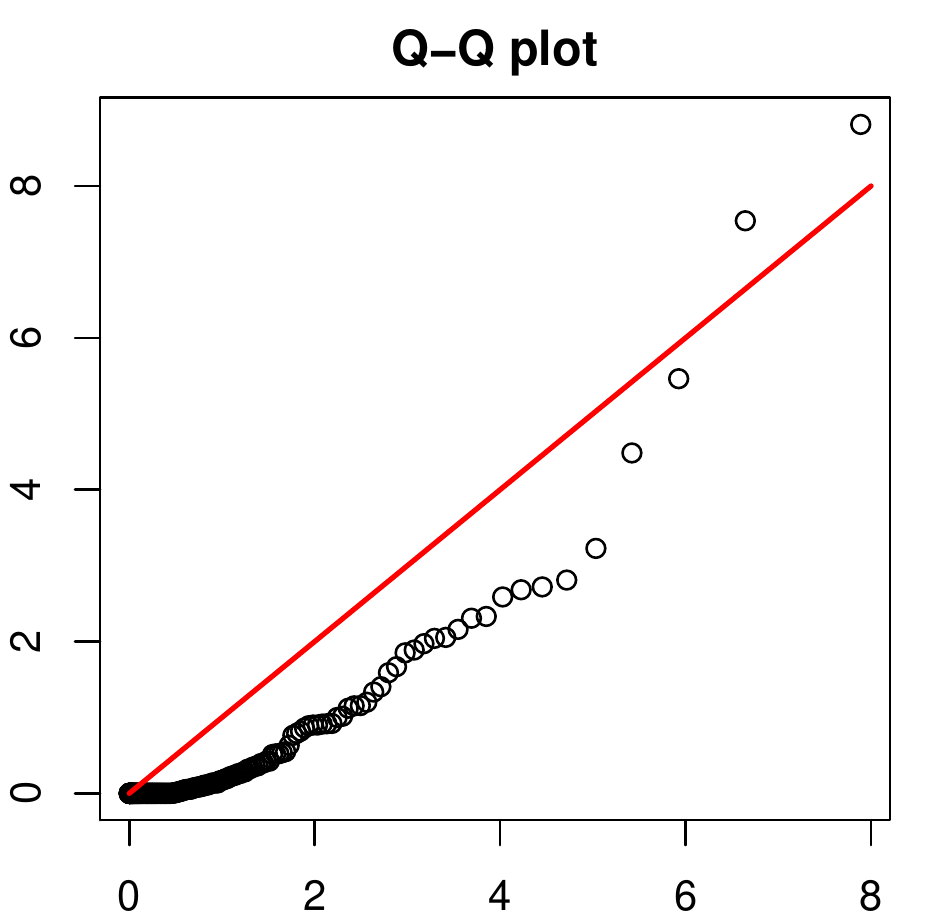}
	}\quad
	\subfigure[$\chi_2^2$]{
		\includegraphics[width=3cm]{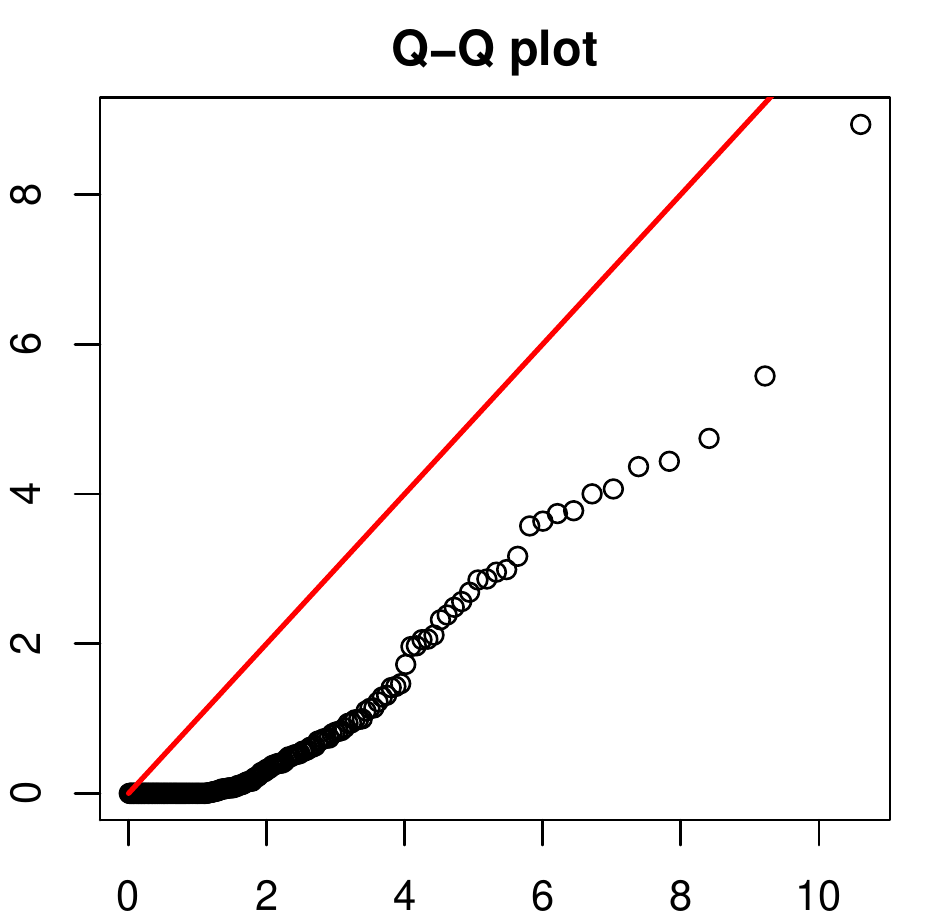}
	}
	\\
	\subfigure[$\chi_2^2$]{
		\includegraphics[width=3cm]{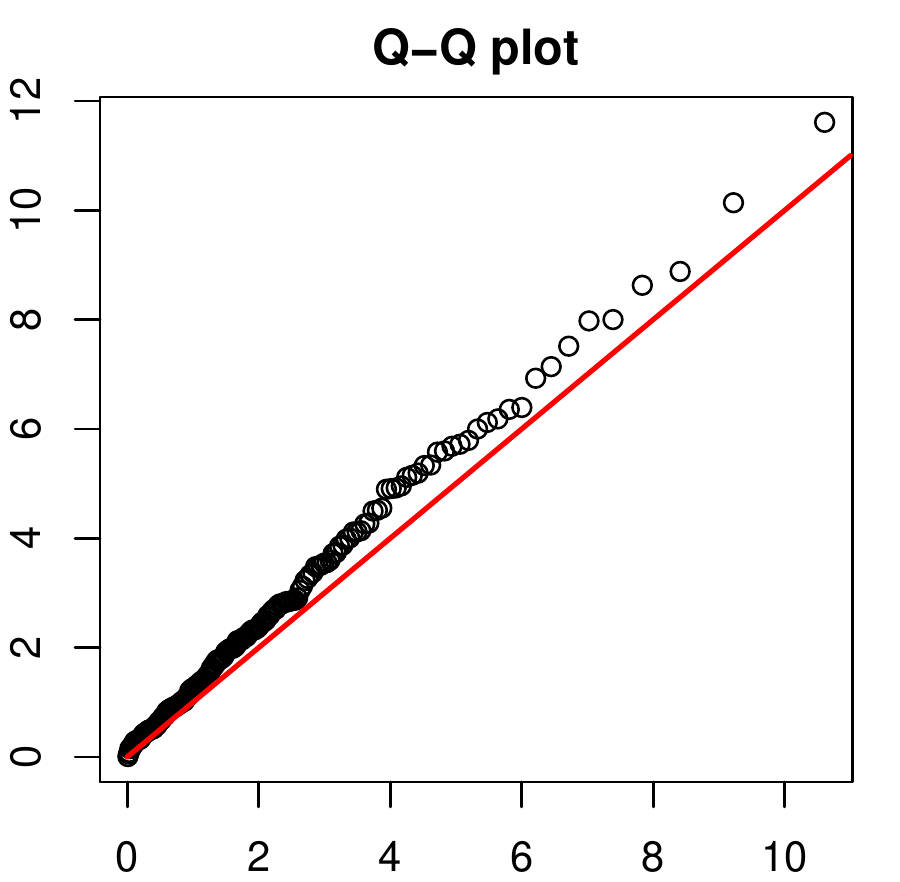}
	}
	\subfigure[$\chi_1^2$]{
		\includegraphics[width=3cm]{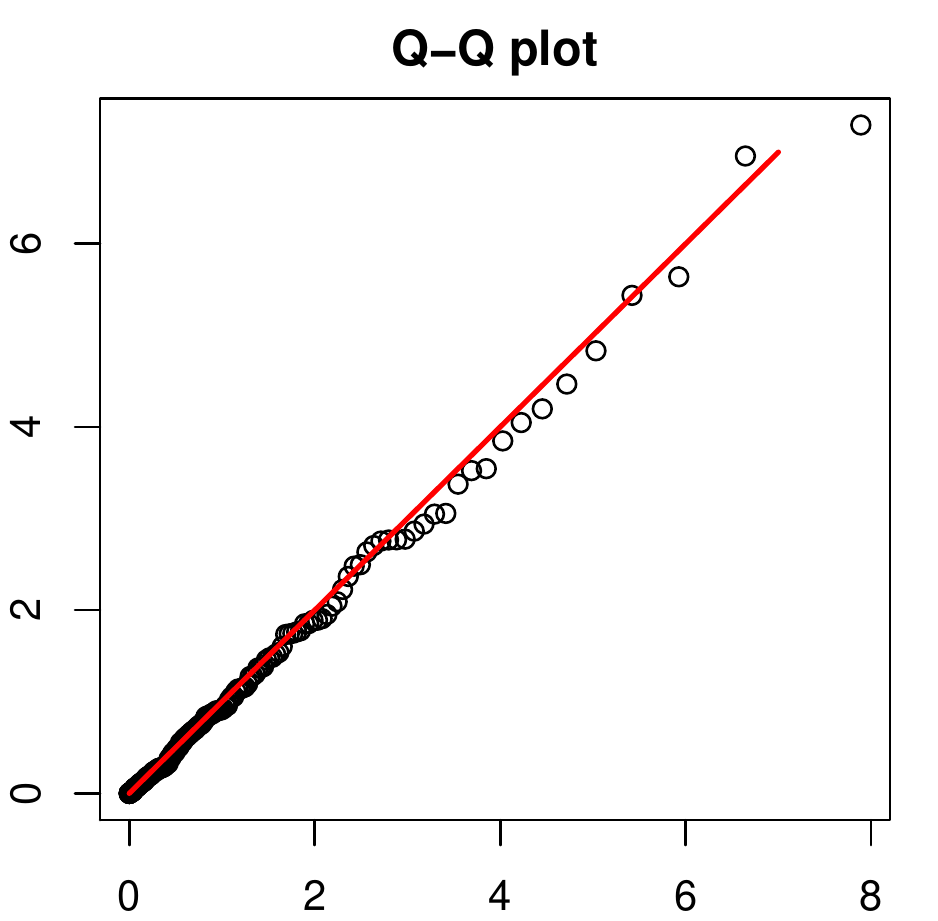}
	}
	\quad
	\subfigure[$\chi_1^2$]{
		\includegraphics[width=3cm]{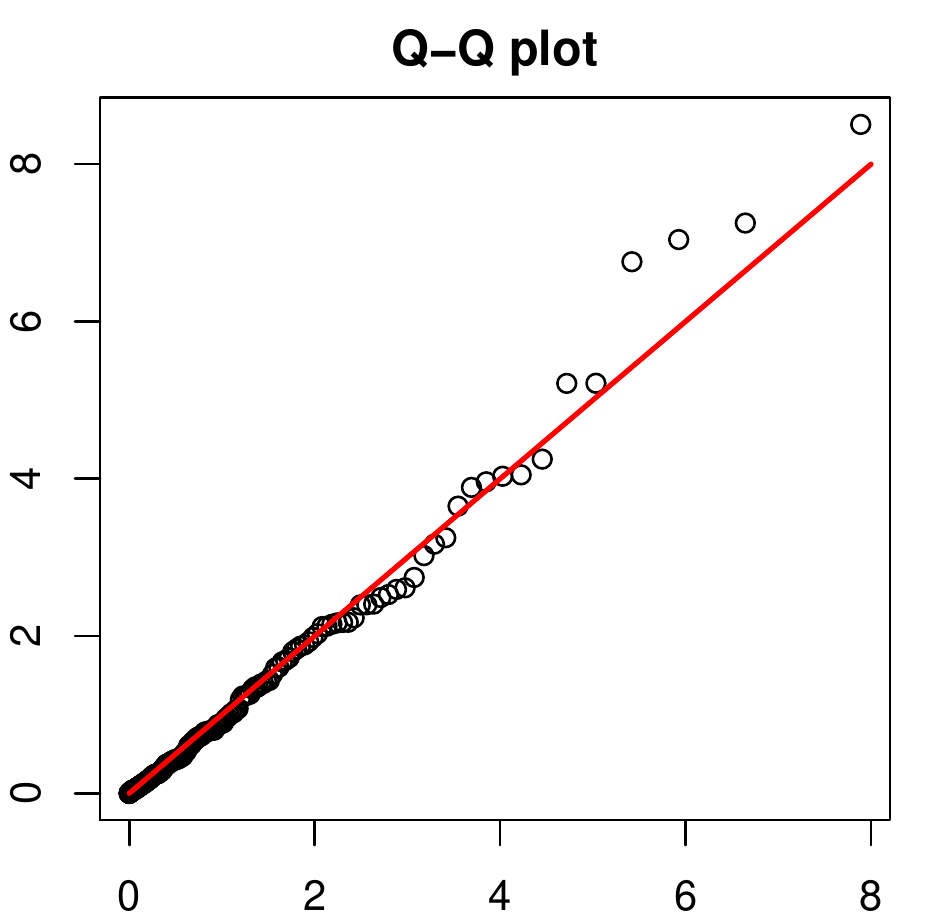}
	}
	\quad
	\subfigure[$\chi_2^2$]{
		\includegraphics[width=3cm]{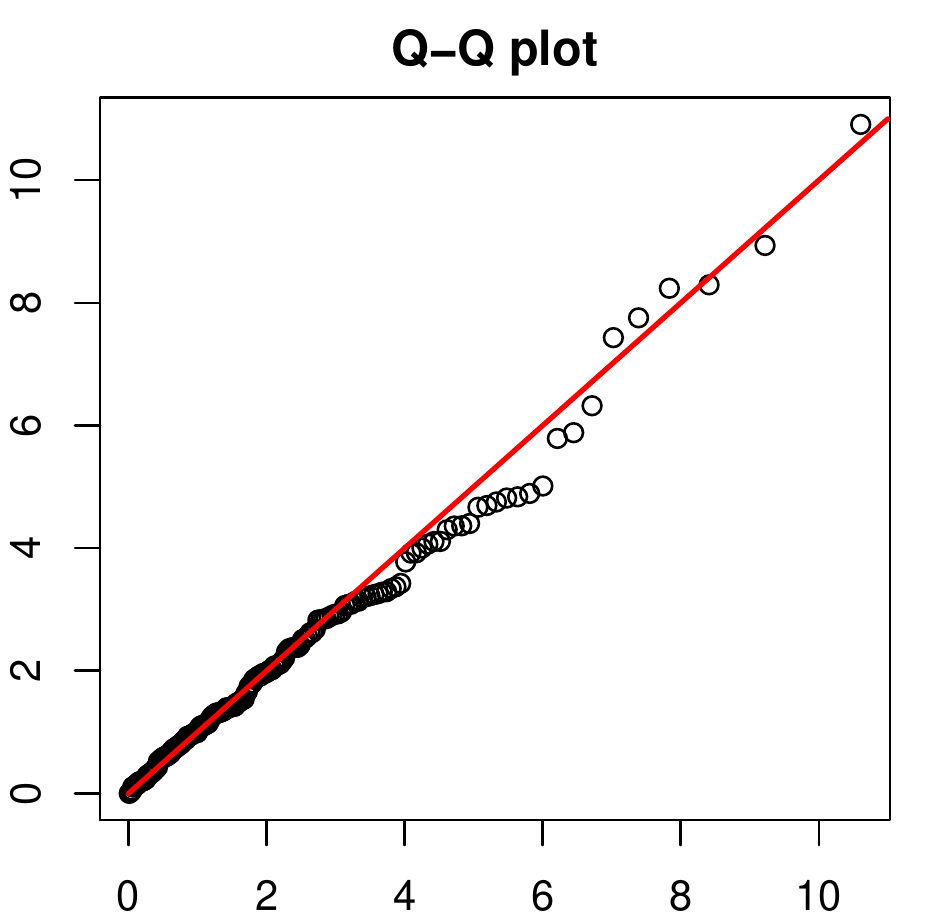}
	}
	\\
	\caption{ The Q-Q plots of the LRT statistics versus the corresponding $\chi^2$ distribution for four tests. First row: the linear mixed-effects model; Second row: the proposed approach.
		First column: $H_0: \sigma_u^2=\sigma_v^2=1$; Second column: $H_0: \sigma_v^2=0$; Third column: $H_0: \sigma_u^2=0$; Fourth column: $H_0: \sigma_u^2=\sigma_v^2=0$. The red lines are the 45 degree reference lines.}
	\label{fig5}
\end{figure}

\end{supplement}



\end{document}